\documentclass[11pt, oneside]{article} 

\usepackage[utf8]{inputenc}
\usepackage[T1]{fontenc}
\usepackage{xcolor} 
\usepackage{color} 

\usepackage{amsmath}
\usepackage{amssymb} 
\usepackage{mathrsfs} 
\usepackage{bm} 
\usepackage{bbm} 
\usepackage{float}
\usepackage{multirow}
\usepackage{comment}

\usepackage{amsthm} 
\usepackage{amsmath}
\theoremstyle{plain}
\newtheorem{theorem}{Theorem}[section] 

\newtheorem{condition}[theorem]{Condition}

\newtheorem{lemma}[theorem]{Lemma}
\newtheorem{proposition}[theorem]{Proposition}
\newtheorem{remark}[theorem]{Remark}

\usepackage{algorithm2e} 
\SetKwInput{KwInputs}{Inputs}
\usepackage[noend]{algpseudocode}
\makeatletter
\def\BState{\State\hskip-\ALG@thistlm}
\makeatother

\usepackage{graphicx} 
\usepackage{subcaption} 
\graphicspath{{figures/}}

\usepackage{tikz} 
\usetikzlibrary{positioning ,shadows,matrix}
\usepackage{pgfplots}
\pgfplotsset{compat=1.18}
\usepgfplotslibrary{groupplots}
\usepackage[english]{babel} 

\usepackage[]{natbib}
\usepackage{bookmark}

\IfFileExists{xurl.sty}{\usepackage{xurl}}{} 
\hypersetup{
	pdftitle={Title},
	pdfauthor={Author 1; Author 2},
	pdfkeywords={3 to 6 keywords, that do not appear in the title},
	colorlinks=true,
	linkcolor={blue},
	filecolor={Maroon},
	citecolor={Blue},
	urlcolor={Blue},
	pdfcreator={LaTeX via pandoc}}

\usepackage{hyperref} 
\hypersetup{
	pdfpagemode={UseOutlines},
	bookmarksopen,
	pdfstartview={FitH},
	colorlinks,
	linkcolor={blue}, 
	citecolor={blue}, 
	urlcolor={blue} 
}

\usepackage[letterpaper,margin=1.25in]{geometry} 

\usepackage{enumitem} 

\def\iid{\overset{\textnormal{iid}}{\sim}} 

\makeatletter  
\@ifundefined{@currsizeindex} 
{\RequirePackage{relsize}\let\dolarger\relsize} 
{\def\dolarger#1{\larger[#1]}} 
\newcommand*\@@bigtimes[2]{\vphantom{\prod} 
	\vcenter{\hbox{\dolarger{4}$\m@th#1\mkern-2mu\times\mkern-2mu$}}} 
\newcommand*\bigtimes{\mathop{\mathpalette\@@bigtimes\relax}\displaylimits} 
\makeatother

\def\iid{\overset{\textnormal{iid}}{\sim}} 

\def\N{\mathbb{N}}\def\R{\mathbb{R}}\def\1{\mathbbm{1}}

\def\Bcal{\mathcal{B}}\def\Ccal{\mathcal{C}}\def\Hcal{\mathcal{H}}\def\Kcal{\mathcal{K}}\def\Ncal{\mathcal{N}}\def\Pcal{\mathcal{P}}\def\Rcal{\mathcal{R}}\def\Scal{\mathcal{S}}\def\Tcal{\mathcal{T}}\def\Wcal{\mathcal{W}}\def\Zcal{\mathcal{Z}}

\def\Enorm{\textnormal{E}}\def\Varnorm{\textnormal{Var}}
\def\Po{\textnormal{Po}}\def\vol{\textnormal{vol}}

\title{\bf A nonparametric Bayesian analysis of
independent and identically distributed
observations of covariate-driven Poisson
processes}
\author{Patric Dolmeta\thanks{
    M.G. has been partially supported by MUR, PRIN project 2022CLTYP4. The authors gratefully acknowledge the support from the ``de Castro" Statistics Initiative.}\hspace{.2cm}\\
    ESOMAS Department, University of Turin\\
    and \\
    Matteo Giordano \\
    ESOMAS Department, University of Turin}
\date{} 

\begin{document}

\maketitle

\begin{abstract}
	An important task in the statistical analysis of inhomogeneous point processes is to investigate the influence of covariates on the point-generating mechanism. We consider the nonparametric Bayesian approach to this problem in the setting where $n$ independent and identically distributed realizations of the point pattern and of the covariate random field are available over a fixed observation window. We employ hierarchical priors based on multi-bandwidth Gaussian processes, and prove that the resulting posterior distributions contract around the ground truth as $n\to\infty$ at a rate that adapts to its possibly anisotropic smoothness. For stationary covariates, this rate is shown to be minimax optimal up to logarithmic factors, via a matching lower bound. Posterior inference is implemented through a Metropolis-within-Gibbs Markov chain Monte Carlo algorithm, which combines dimension-robust updates of the functional component of the model with non-centered updates of the length-scales. Our theoretical results are supported by extensive numerical simulation studies. Further, we present an application to a Canadian wildfire dataset with meteorological covariates.
\end{abstract}


\tableofcontents

\section{Introduction}\label{Sec:Intro}

Inhomogeneous point processes are fundamental mathematical models for describing the distribution of events that take place randomly over space and time. In many applications, the occurrence of the events is determined, or heavily influenced, by covariates. It is then of interest to investigate the relationship between the points and the covariates. This can be formalized as an intensity estimation problem in the following way: let $N$ be a point process over some Euclidean domain $\Wcal$, and for any $A\subseteq \Wcal$, denote by $N(A)$ the number of events within $A$. The (first-order) intensity function of $N$ is a map $\lambda:\Wcal\to [0,\infty)$ with the property that $\Enorm[N(A)] = \int_A\lambda(x)dx$, and $N$ is said to be inhomogeneous if $\lambda$ is non-constant. Additionally, let $Z=\{Z(x), \ x\in\Wcal\}$ be a multivariate field, with values in some subset $\Zcal\subseteq \R^d$, $d\in\N$, representing the measurements of the covariates at each location in the domain. The connection between $N$ and $Z$ is then modeled by assuming that the intensity is driven by the covariates, namely
\begin{equation}
	\label{Eq:BasicIntensity}
	\lambda(x) = \rho(Z(x)), 
	\qquad x\in\Wcal,
\end{equation}
for some unknown function $\rho:\Zcal\to[0,\infty)$. Note that this allows for `purely spatial' effects by listing the `coordinates' of $x$ within the covariates $Z(x)$. The goal is to estimate $\rho$ from observations of $N$ and $Z$. Such covariate-based intensity estimation problems arise in a variety of scientific fields, including environmental statistics (e.g.~\citep{BGMMM20}), geology (e.g.~\citep{BCST12}) and ecology (e.g.~\citep{G08}), to mention a few. See Section \ref{Sec:RealData} for an application to the study of wildfire occurrences in relation to meteorological covariates.

When the covariates are (all or partly) random, the resulting point process is `doubly stochastic'. In particular, we focus on the case where the overarching point-generating mechanism is of Poisson type, whereby $N$ defines a Cox process \citep{C55}. In this framework, covariate-based intensity estimation has been widely studied under parametric models for the function $\rho$ in \eqref{Eq:BasicIntensity}, both in the frequentist (e.g.~\citep{B78,D90,W07}) and Bayesian literature (e.g.~\citep{RMC09, ISR12,LG17}). See also the monograph by \citet{D14}, and references therein. For instance, the celebrated log-Gaussian Cox model of \citet{MSW98} postulates that $Z$ is a multivariate Gaussian process, and that $\lambda(x) = e^{\beta^TZ(x)}$ for some $\beta\in\R^d$. Log-linear specifications of this type are often combined with latent Gaussian random fields capturing residual spatial variation, e.g.~\citep{DMRT13}, or with semiparametric components \citep{YL11}.

In contrast, the nonparametric literature on the subject is less developed. An early frequentist investigation is due to \citet{G08}, who constructed covariate-based kernel methods and derived asymptotic point-wise consistency results under the assumption that $Z$ is a stationary and ergodic random field, in an `increasing domain' asymptotic regime in which the volume of $\Wcal$ diverges and a single realization of $N$ and $Z$ is observed. Similar estimators were later defined by \citet{BCST12} and by \citet{BGMMM20}, and studied under related sampling schemes. More recently, machine-learning approaches, e.g.~based on boosted trees \citep{LGvLX25} and random forests \citep{BL25}, have been proposed.

Nonparametric Bayesian methods for intensity estimation have been mostly developed for non-covariate-driven point processes. Seminal methodological advances were provided by, among others, \citet{L82,KG97,dMGK01,KS07,AMM09}, covering a variety of prior distributions, ranging from gamma process-based ones to beta process, spline, kernel mixture and Gaussian process priors, respectively. Building on the landmark developments in the theory of Bayesian nonparametrics from the early 2000s \citep{GGvdV00}, several articles have investigated the asymptotic properties of nonparametric Bayesian procedures in models without covariates. \citet{BSvZ15} developed the Hellinger testing approach for independent and identically distributed (i.i.d.)~observations of inhomogeneous Poisson processes over fixed domains, and derived minimax optimal posterior contraction rates for spline priors with uniform coefficients. These results were extended to Gaussian process priors by \citet{KvZ15} under similar i.i.d.~schemes; see also \citet{GS13}. Procedures with piecewise-constant priors were investigated by \citet{GvdMSS20}. Lastly, \citet{DRRS17} obtained optimal performance guarantees for general Aalen point processes under various types of smoothness and shape constraints. Covariates have instead been incorporated in nonparametric Bayesian models for closely related problems, notably in survival analysis, via semiparametric proportional hazards models with nonparametric priors on the baseline hazard \citep{IJ04}, mixture models with covariate-dependent components \citep{DJMR09}, and Gaussian process priors on hazard functions \citep{FRT16}.

Nonparametric Bayesian approaches to covariate-driven intensity estimation have been considered more recently. \citet{KAT22} devised fast inference algorithms based on Gaussian process priors for the square root of the intensity, while \citet{GKR26} derived optimal global and local posterior contraction rates for several classes of prior distributions; see also the related contribution by \citet{DG26}. These works, like the frequentist ones discussed above, assume the availability of a single realization of the points and covariates, the latter two considering increasing domain asymptotic regimes similar to that of \citet{G08}. While large observation windows are common in spatial statistics, many natural applications with point processes are confined to fixed domains, and rather entail multiple observations of the event pattern and the covariates, each carrying individual information that needs to be effectively combined; see Section~\ref{Sec:RealData} for an example with yearly records of wildfires and meteorological conditions. In this setting, Bayesian nonparametrics offers an attractive approach, as it allows for flexible modeling of intensity functions that may depend nonlinearly on the covariates and vary at different scales across them, beyond the reach of standard parametric models, while providing uncertainty quantification for the estimated intensity. Despite interest in the joint analysis of multiple realizations having been raised since at least \citet{DLB91}, this setting appears to be under-explored in the nonparametric literature, and the results and proof techniques of \citet{GKR26}, based on concentration inequalities for stationary and ergodic spatial random fields, do not apply to it. This gap represents the main motivation for our work.

%
%
%

\subsection{Our contributions}

In this article, we develop a nonparametric Bayesian analysis of i.i.d.~observations of covariate-driven Poisson processes. We model $\rho$ in \eqref{Eq:BasicIntensity} via a suitable prior distribution and then form, via Bayes' theorem, the corresponding posterior, which encodes our updated belief about $\rho$, providing point estimates and uncertainty quantification. See \citet[Chapter 1]{GvdV17} for an overview of Bayesian nonparametrics.

For the specification of the prior, we employ `multi-bandwidth' Gaussian processes, obtained by scaling stationary covariance functions at different levels along distinct directions; see Section \ref{Subsec:Prior}. This construction, which includes the widely used Automatic Relevance Determination (ARD) kernel from machine learning \citep[Section 5.1]{RW06}, offers desirable modeling flexibility for `anisotropic' functions whose variations in response to changes in different inputs may occur at distinct characteristic length-scales, or according to diverse smoothness levels; see Section \ref{Subsec:MainNotation} for precise definitions. Multi-bandwidth Gaussian processes were shown by \citet{BPD14}, in simpler statistical models, to be able to achieve optimal reconstructions over anisotropic function spaces.

In our main theoretical result, Theorem \ref{Theo:GPRates}, we derive posterior contraction rates towards the (possibly anisotropic) true intensity function generating the data, in the asymptotic regime where the number of observed realizations of $N$ and $Z$ increases. The obtained rates automatically adapt to the smoothness of the intensity, and, in a setting with stationary covariates with sufficiently regular marginal distribution, are shown to be minimax optimal up to logarithmic factors via a matching lower bound (Proposition~\ref{Prop:LowerBound}). The proofs are based on the Hellinger testing approach for i.i.d.~sampling schemes, which we specialize to the case of covariate-driven Poisson processes adapting ideas from \citet{BSvZ15} and \citet{KvZ15}, and which we then pursue for the proposed multi-bandwidth Gaussian process methods, over anisotropic function spaces.  Adaptation is achieved through a hierarchical procedure where we model the length-scales in the ARD kernel via a tilted gamma hyper-prior that penalizes the overall complexity of the prior draws. This is inspired by the hyper-prior employed by \citet{BPD14}, who rather control the complexity by raising gamma-distributed length-scales to random exponents jointly drawn from a Dirichlet distribution.

A second contribution of this work concerns the implementation of the nonparametric Bayesian approach to covariate-based intensity estimation. In Section~\ref{Subsec:Algorithm}, we devise a Metropolis-within-Gibbs Markov chain Monte Carlo (MCMC) algorithm to approximately sample from the posterior distribution, alternating draws from the full conditional distributions of the various parameters. It combines dimension-robust pre-conditioned Crank--Nicolson updates \citep{CSRW13} for the functional component of the model with non-centered updates of the length-scales \citep{PRS07}, and relies on numerical approximations of the likelihood, whose impact is tested through a sensitivity analysis.

To assess our method in practice, we conducted extensive numerical simulations, presented in Section~\ref{Sec:Simulations}, testing its performance against kernel-type estimators and log-polynomial Poisson regression. The empirical results are in close agreement with the theory, illustrating the ability of the proposed procedure to reconstruct the true intensity function in a variety of experimental setups, with increasing accuracy as the number of observations grows. The obtained reconstructions were found to be competitive against the alternative methods, with lower estimation errors from moderate sample sizes onwards. Lastly, in Section~\ref{Sec:RealData}, we develop an application to a Canadian wildfire dataset containing yearly recordings of hotspots and meteorological covariates, where our approach achieves an effective aggregation of the information across the observation period, delivering interpretable estimates of the wildfire intensity as a function of the covariates and capturing year-specific trends in the spatial distribution of the wildfires.

The rest of the paper is organized as follows: in Section~\ref{Sec:GPMethods}, we describe the statistical problem and our approach in detail, present our main theoretical results, and outline the MCMC sampler. The numerical simulations and data analysis are presented in Sections~\ref{Sec:Simulations} and~\ref{Sec:RealData}, respectively. Section~\ref{Sec:Discussion} contains a summary and a discussion of open problems and possible extensions. The proofs of all the results are deferred to the Supplementary Materials \citep{SUPPL}, where additional simulations, a sensitivity analysis, MCMC diagnostics and more details on the data analysis can also be found.

%
%
%

\subsection{Main notation}
\label{Subsec:MainNotation}

For positive integers $m\in\N$, we denote $m$-dimensional vectors of real numbers by $t=(t_1,\dots,t_m)\in\R^m$, intended as column vectors unless otherwise stated. For $r,t\in\R^m$, we write $r\odot t:=(r_1t_1,\dots,r_mt_m)$ for the component-wise product of $r$ and $t$. The symbols $a\land b$ and $a\vee b$ are used for the minimum and maximum of real numbers $a,b\in\R$, respectively. The notation $\lesssim$, $\gtrsim$ and $\simeq$ indicates one- and two-sided inequalities holding up to multiplicative constants that do not depend on the quantities of interest, and $\propto$ indicates the proportionality of a function with respect to its arguments.

Given a measure space $(\Tcal,\mathfrak{T},\tau)$ and any $ p\in[1,\infty]$, let $L^p(\Tcal,\tau)$ be the Lebesgue space of real-valued $p$-integrable (essentially bounded if $p=\infty$) functions defined on $\Tcal$, equipped with the usual norm $\|\cdot\|_{L^p(\Tcal,\tau)}$. When $\Tcal\subseteq\R^m$ and $\tau$ equals the Lebesgue measure $dt$, we write for brevity $L^p(\Tcal)$ in place of $L^p(\Tcal,\tau)$.

For $\Tcal\subseteq\R^m$, denote by $C(\Tcal)$ the space of continuous functions defined on $\Tcal$, equipped with the sup-norm $\|\cdot\|_\infty$. For $q\in(0,\infty)$, write $C^q(\Tcal)$ for the usual Hölder space of $(\lceil q\rceil-1)$-times differentiable functions whose $(\lceil q\rceil-1)^{\text{th}}$ derivatives are $(q-\lceil q\rceil+1)$-Hölder continuous, with norm $\|\cdot\|_{C^q(\Tcal)}$. Next, we define the family of anisotropic Hölder spaces, containing functions whose degree of smoothness may differ along distinct directions, cf.~\citet[Section 4.1.3]{BBM99}. For all $t\in\Tcal$ and $k=1,\dots,m$, let $\Scal_{t,k}:=\{s\in\R:(t_1,\dots,t_{k-1},s,t_{k+1},\dots,t_m)\in\Tcal\}$, and for any $f\in C(\Tcal)$, define the univariate functions $f_{t,k}:s\mapsto f(t_1,\dots,t_{k-1},s,t_{k+1},\dots,t_m)$ on $\Scal_{t,k}$. For $\alpha=(\alpha_1,\dots,\alpha_m)\in(0,\infty)^m$, let $C^\alpha(\Tcal)$ be the subset of all $f\in C(\Tcal)$ such that
$$
\max_{k=1,\dots,m}\sup_{t\in\Tcal}\|f_{t,k}\|_{C^{\alpha_k}(\Scal_{t,k})}<\infty.
$$
The above definition recovers the traditional (isotropic) Hölder spaces if $\alpha_k=\alpha_h$ for all $h,k=1,\dots,m$. When there is no risk of confusion, we sometimes omit the dependence of the function spaces on the underlying domain, writing for example $L^p$ for $L^p(\Tcal)$.

%
%
%
%
%

\section{Multi-bandwidth Gaussian process methods for covariate-based intensities}
\label{Sec:GPMethods}

On a compact `observation window' $\Wcal\subset\R^D$, $D\in\N$, consider a $d$-dimensional `covariate' random field $Z:=\{Z(x), \ x\in\Wcal\}$ with values in some `covariate space' $\Zcal\subseteq\R^d$, $d\in\N$, and a stochastic point pattern $N:=\{X_1,\dots,X_K\}$ arising, conditionally given $Z$, as an inhomogeneous Poisson process with first-order intensity $\lambda_\rho (x) := \rho(Z(x))$, $x\in\Wcal$, for some unknown (measurable and bounded) function $\rho:\Zcal\to[0,\infty)$. That is, $N$ is a Cox process, \citet{C55}, directed by the random measure $\lambda_\rho(x)dx$, and we have
\[
K|Z\sim \Po\Big(\int_{\Wcal} \lambda_\rho(x)dx  \Big),
\qquad X_1,\dots,X_K|Z,K\iid \frac{\lambda_\rho(x)dx}{\int_{\Wcal} \lambda_\rho(x)dx}.
\]

For some $n\in\N$, we assume that we observe $n$ i.i.d.~copies of the pair $(N,Z)$, denoted by $D^{(n)}:=\{(N^{(i)},Z^{(i)})\}_{i=1}^n$, where $N^{(i)}:=\{X^{(i)}_1,\dots,X^{(i)}_{K^{(i)}}\}$. We then seek to estimate $\rho$ from data $D^{(n)}$. Throughout, we denote by $P^{(n)}_\rho$ the law of $D^{(n)}$, by $E^{(n)}_\rho$ the expectation with respect to it, and write $P_\rho := P^{(1)}_\rho$, $E_\rho := E^{(1)}_\rho$.
The law $P^{(n)}_\rho$ is absolutely continuous with respect to the distribution $P_1^{(n)}$ of the standard Poisson case (where $\rho\equiv 1$), with likelihood (see e.g.~\citet[Theorem 1.3]{K98}),
\begin{equation}
	\label{Eq:Likelihood}
	L^{(n)}(\rho) = \prod_{i=1}^n p_\rho(N^{(i)},Z^{(i)}),
	\qquad p_\rho(N,Z) = e^{\sum_{k=1}^K\log \rho(Z(X_k)) 
		- \int_{\Wcal} (\rho(Z(x))-1) dx}.
\end{equation}

\begin{remark}[Repeated observations of covariates and points]\label{Rem:ConCov}
	Throughout, we are interested in the setting where multiple realizations of $Z$ and $N$ are available. For example, this is the case in the data analysis we develop in Section \ref{Sec:RealData} below, where we have access to yearly observations. Depending on the application at hand, different measurement schemes may be relevant, such as ones where multiple point patterns are driven by shared (possibly deterministic) covariates. In other cases, a single realization of a covariate-driven point process with large intensity may be available, giving rise to a so-called `in-fill' asymptotics. These settings are closely related to the increasing domain regime recently studied by \citet{GKR26}, where optimal posterior contraction rates for various nonparametric Bayesian procedures were obtained. See \citet{G08} for an earlier related reference. Here, we focus on i.i.d.~sampling schemes for covariate-driven Poisson processes, whose theoretical analysis requires different tools and techniques. 
\end{remark}

%
%
%

\subsection{The prior model}
\label{Subsec:Prior}

We adopt the nonparametric Bayesian approach, modeling $\rho$ with a prior $\Pi$ based on Gaussian processes. To do so, we maintain the assumption that the covariate space equals $\Zcal=[0,1]^d$ (cf.~Remark \ref{Rem:BoundCov}) and introduce the parametrization
\begin{equation}
	\label{Eq:Param}
	\rho(z) = \rho^* \sigma( w(z)), \qquad z\in\Zcal,
\end{equation}
where $\rho^*>0$ is an upper bound for the values of the intensity, $w:\Zcal\to\R$ is some unknown function and $\sigma:\R\to(0,1)$ is a fixed, smooth and strictly increasing link function. Throughout, we employ the sigmoid link $\sigma(t) = (1 + e^{-t})^{-1}, \ t\in\R$.

Under $\eqref{Eq:Param}$, we specify $\Pi$ by assigning independent priors $\Pi_{\rho^*}$ and $\Pi_W$ to $\rho^*$ and $w$, respectively. Specifically, for some fixed $a_{\rho^*},b_{\rho^*},c_{\rho^*}>0$, we take $\Pi_{\rho^*}$ to be a $\Gamma(a_{\rho^*},b_{\rho^*})$ distribution (with shape $a_{\rho^*}$ and rate $b_{\rho^*}$) truncated to the interval $[0,c_{\rho^*}+\log n]$. Its probability density function (p.d.f., also denoted by $\Pi_{\rho^*}$) equals
$$
\Pi_{\rho^*}(r) = \frac{b_{\rho^*}^{a_{\rho^*}}}{\gamma(a_{\rho^*},b_{\rho^*}c_{\rho^*}+b_{\rho^*}\log n)} r^{a_{\rho^*} - 1}e^{-b_{\rho^*}r},
\qquad r\in[0,c_{\rho^*}+\log n],
$$
where $\gamma(a_{\rho^*},b_{\rho^*}c_{\rho^*}+b_{\rho^*}\log n) := \int_0^{b_{\rho^*}c_{\rho^*}+b_{\rho^*}\log n}r^{a_{\rho^*} - 1}e^{-r}dr$ is positive and bounded above by $\Gamma(a_{\rho^*})$ for all $n\in\N$. This leads to a conjugate full conditional distribution on $\rho^*$ that is computationally convenient, cf.~Section \ref{Subsec:Algorithm}, and also implies a bound on the sup-norm of $\rho$, used in the theoretical analysis, cf.~Remark~\ref{Rem:BoundInt}.

Next, we model $w$ via a family of centered Gaussian processes $W_\ell:=\{W_\ell(z), \ z\in\Zcal\}$ with Automatic Relevance Determination (ARD) kernel,
\begin{equation}
	\label{Eq:ARD}
	\Enorm [W_\ell(z)W_\ell(z')] = e^{-\sum_{j=1}^d  \ell_j^2 (z_j - z_j')^2},
	\qquad z=(z_1,\dots,z_d), \ z' = (z_1',\dots,z_d'),
\end{equation}
where $ \ell_1,\dots, \ell_d>0$ are (inverse) length-scale hyper-parameters and $\ell = ( \ell_1,\dots, \ell_d)$. This covariance function represents the anisotropic generalization of the standard squared-exponential kernel, which prescribes $\ell_j =\ell_h$ for all $h,j=1,\dots,d$. It offers desirable modeling flexibility in the present setting, where distinct covariates may have diverse physical nature and vary over vastly different ranges, possibly resulting in intensities with distinct smoothness levels along different directions. The ARD kernel is widely used in machine learning in such situations, e.g.~\citet[Section 5.1]{RW06}, and was shown by \citet{BPD14}, in simpler statistical models, to lead to optimal reconstruction of anisotropic functions.

We conclude the specification of $\Pi_W$ (and of $\Pi$) by randomizing the length-scales $\ell$ in \eqref{Eq:ARD} via a tilted gamma hyper-prior $\Pi_\ell$, with p.d.f.~(also denoted by $\Pi_\ell$) satisfying
\begin{equation}
	\label{Eq:HyperPrior}
	\Pi_\ell( l )
	\propto\prod_{j=1}^d l_j^{a_\ell-1}e^{-b_\ell l_j}
	\times e^{-\lambda\prod_{j=1}^d(1+l_j)},\qquad l = (l_1,\dots,l_d)\in(0,\infty)^d,
\end{equation}
for some $a_\ell,b_\ell,\lambda>0$. That is, $\Pi_\ell$ is given by the joint law of $d$ independent $\Gamma(a_\ell,b_\ell)$-random variables tilted by a penalty on the quantity $\prod_{j=1}^d(1+l_j)$. The latter is a natural complexity proxy, since larger inverse length-scales shorten correlation ranges and allow prior draws to fluctuate more abruptly.  The full conditional prior distributions for the individual length-scales retain a gamma form,
\begin{equation}
	\label{Eq:EllFullCond}
	\ell_j\mid\ell_{1},\dots,\ell_{j-1},\ell_{j+1},\dots,\ell_d\sim
	\Gamma\Big(a_\ell,
	b_\ell+\lambda\prod_{k\ne j}(1+\ell_k)\Big),
	\qquad j=1,\dots,d,
\end{equation}
whereby finer resolution in the other directions $k\neq j$ increases shrinkage in the direction $j$. 
This construction is inspired by the hyper-prior from \citet[Section 3.1]{BPD14}, where the length-scales are modeled as gamma random variables raised to random exponents jointly drawn from a Dirichlet distribution. In their results, the simplex constraint controls the overall complexity, while the random allocation allows adaptation to anisotropic smoothness. Here, complexity is penalized directly. This eliminates the need for auxiliary random exponents, while retaining, as shown in Theorem \ref{Theo:GPRates} below, enough flexibility to effectively model anisotropic functions. For $d=2$, we note that the joint law $\Pi_\ell$ from \eqref{Eq:HyperPrior} belongs to the classical auto-gamma family \citep{moller1999perfect}.

In the specification of $\Pi$, the parameters $a_{\rho^*}, b_{\rho^*},c_{\rho^*},a_\ell,b_\ell,\lambda$ are arbitrary positive quantities. They play no role in our proofs, only possibly affecting the constants pre-multiplying the rates. The concrete values employed in the simulation study and data analysis are reported below, and the robustness of the finite-sample performance with respect to these choices is investigated in Section~\ref{Subsec:Sensitivity} of the Supplement.

Following the Bayesian paradigm, given data $D^{(n)}$ arising as described in Section \ref{Sec:GPMethods}, and $\Pi$ as above, the posterior distribution $\Pi(\cdot|D^{(n)})$ is given by the conditional distribution of $\rho|D^{(n)}$. By Bayes' theorem (e.g.~\citet[p.~7]{GvdV17}), 
\begin{equation}
	\label{Eq:Posterior}
	\Pi(A|D^{(n)}) 
	=\frac{\int_A L^{(n)}(\rho)d\Pi(\rho)}
	{\int_{\Rcal} L^{(n)}(\rho)d\Pi(\rho)},
	\qquad A\subseteq \Rcal\ \text{measurable},
\end{equation}
where $L^{(n)}$ is the likelihood from \eqref{Eq:Likelihood}, and where $\Rcal$ is the collection of all measurable, bounded and nonnegative-valued functions defined on the covariate space $\Zcal$.

\begin{remark}[Bounded covariate spaces]
	\label{Rem:BoundCov}
	The assumption that $\Zcal=[0,1]^d$ is a convenient working assumption that can always be met in practice since, for covariates taking values in a general (possibly unbounded) subset $\Zcal$ of $\R^d$, we might `pre-process' them via a smooth one-to-one map $\Phi:\R^d\to[0,1]^d$, setting $\tilde Z^{(i)}(x):=\Phi(Z^{(i)}(x))$, $x\in\Wcal$, $i=1,\dots,n$. We would then proceed in the statistical analysis using the transformed covariates in place of the original ones, and then translate the obtained estimates back onto the original covariate space via the inverse map $\Phi^{-1}$. A standard choice is
	\begin{equation}\label{Eq:PreProc}
		\Phi(z)=(\phi(z_1),\dots,\phi(z_d)),\qquad z=(z_1,\dots,z_d)\in\R^d,
	\end{equation}
	where $\phi:\R\to[0,1]$ is a smooth and strictly increasing cumulative distribution function (c.d.f.). Analogous standardization steps are common practice in spatial statistics applications, for example being embedded within the popular \textnormal{\texttt{R}} package   \textnormal{\texttt{spatstat}} (\citet{BRT16}) for kernel-based intensity estimation; see also \citet[Section 3.2.1]{G08}. 
\end{remark}

%
%
%

\subsection{Adaptive anisotropic posterior contraction rates}\label{Subsec:Theory}

We present our main theoretical results concerning the asymptotic behavior of the posterior distribution \eqref{Eq:Posterior} as $n\to\infty$, under the paradigm of the `frequentist analysis of Bayesian procedures' (e.g.~\citet{GvdV17}). We assume observations $D^{(n)}\sim P^{(n)}_{\rho_0}$ generated by some fixed (possibly anisotropic) ground truth $\rho_0$, and study the convergence of $\Pi(\cdot|D^{(n)})$ towards $\rho_0$. In the following result, we quantify the speed of such concentration with respect to the distance
\begin{equation}
	\label{Eq:dZ}
	d_Z (\rho_1,\rho_2)
	:= \sqrt{\Enorm\left\|\sqrt{\lambda_{\rho_1}} - \sqrt{\lambda_{\rho_2}}\right\|_{L^2(\Wcal)}^2}
	=\sqrt{\Enorm \int_{\Wcal} \left(\sqrt{\rho_1(Z(x))} - \sqrt{\rho_2(Z(x))}\right)^2 dx},
\end{equation}
where the expectation is with respect to the law of $Z$. This is a natural metric for the problem at hand, as it is closely related to the Hellinger distance between the observational densities \eqref{Eq:Likelihood}, cf.~Section \ref{Subsec:HellDist} in the Supplement \citep{SUPPL}. In the important case where $Z$ is stationary, under mild conditions on its marginal distribution, $d_Z$ is equivalent to a standard $L^2$-type metric, under which the obtained rates match, up to a log-factor, the minimax rates over anisotropic Hölder classes. See Section \ref{Subsec:StatCov}.

\begin{theorem}\label{Theo:GPRates}
	For $\alpha = (\alpha_1,\dots,\alpha_d)\in (0,\infty)^d$, 
	let $\rho_0\in C^\alpha([0,1]^d)$ be strictly positive,  and consider data $D^{(n)}\sim P_{\rho_0}^{(n)}$ arising as described at the beginning of Section \ref{Sec:GPMethods}. Let $\Pi$ be a prior for $\rho$ constructed as in Section \ref{Subsec:Prior}, and let $\Pi(\cdot|D^{(n)})$ be the resulting posterior distribution.
	Then, in $P^{(\infty)}_{\rho_0}$-probability, as $n\to\infty$,
	$$
	\Pi\left(\rho : d_Z (\rho,\rho_0 ) > L n^{-\alpha_0/(2\alpha_0+1)}\log^C n
	\Big|D^{(n)}\right) \to 0,
	\qquad \alpha_0 = 1/\sum_{j=1}^d\alpha_j^{-1},
	$$
	for all sufficiently large $L>0$ and some large enough $C>0$.
\end{theorem}

The proof is given in Section \ref{Sec:ProofGPRates} of the Supplement. It is developed under a general condition on the hyper-prior on the length-scales, cf.~Condition \ref{Cond:GP} therein.

Theorem \ref{Theo:GPRates} asserts that if $\rho_0\in C^\alpha([0,1]^d)$, then, with probability tending to one, any posterior sample $\rho\sim\Pi(\cdot|D^{(n)})$ is within a small $d_Z$-neighborhood of $\rho_0$ with radius shrinking at the order $n^{-\alpha_0/(2\alpha_0+1)}$ (up to log factors). The quantity $\alpha_0$ is called `effective smoothness' \citep[p.~326]{HL02}, and is known to characterize the minimax optimal rates of estimation over anisotropic function spaces; see e.g.~\citet{N87} for results in nonparametric regression. Note that $\alpha_0$ is increasing with respect to each component of $\alpha$; in particular, the sequence $n^{-\alpha_0/(2\alpha_0+1)}$ can be made arbitrarily close to the parametric rate $n^{-1/2}$ if $\rho_0$ is infinitely differentiable along each direction. Since the employed prior $\Pi$ requires no information about $\alpha$ (or $\alpha_0$) for its construction, we conclude that it is able to automatically `adapt' to the (possibly) anisotropic smoothness. This is in line with the findings from \citet[Section 3.1]{BPD14}, which we build on to investigate the present covariate-based intensity estimation problem.

In the isotropic case where $\alpha_j = a$ for some $a>0$ and all $j=1,\dots,d$, we have $\alpha_0 = a/d$, and Theorem \ref{Theo:GPRates} recovers the usual nonparametric rate $n^{-a/(2a+d)}$, up to a log-factor. On the other hand, in the presence of a genuine anisotropy, since $\alpha_0\ge\min_{j=1,\dots,d}\alpha_j/d$, treating the smoothness of $\rho_0$ as being isotropic (and equal to $\min_{j=1,\dots,d}\alpha_j$) generally results in slower rates, with greater loss of efficiency in higher dimensions. Multi-bandwidth Gaussian process priors are suited to both scenarios, while it was shown by \citet[Section 3.5]{BPD14} that single-bandwidth procedures can lead to sub-optimal performance if the ground truth is anisotropic.

\begin{remark}[Bounded intensity ceiling]
	\label{Rem:BoundInt}
	The bounded link $\sigma$ and the truncation of $\Pi_{\rho^*}$ within the prior construction in Section~\ref{Subsec:Prior} are used in the proof of Theorem~\ref{Theo:GPRates} to control the complexity of the prior and to relate the Hellinger distance to the metric $d_Z$. The $n$-dependent truncation is not essential: for an untruncated gamma prior on $\rho^*$, the same rate, up to logarithmic factors, holds for the cut-off distance $d_Z(\rho\land M,\rho_0\land M)$, for any fixed $M>1$; see Remark~\ref{Rem:Cutoff} in the Supplement for details. On the other hand, since the employed truncation level $c_{\rho^*}+\log n$ increases with $n$, it imposes only a mild restriction on the prior, while allowing us to remove the cut-off in Theorem~\ref{Theo:GPRates}, with no knowledge required of an upper bound on $\|\rho_0\|_\infty$. We also note that boundedness of $\rho_0$ is automatic for continuous intensities on $[0,1]^d$, and does not preclude highly localized peaks, which the prior can capture through large values of $\rho^*$ combined with steep variations of $w$. For unbounded links such as $\rho=e^{w}$, the posterior sampling algorithm developed in Section~\ref{Subsec:Algorithm} readily applies with $\rho^*$ removed, but the proof of Theorem~\ref{Theo:GPRates} would require different arguments, since the Lipschitz constant of $w\mapsto e^{w/2}$ grows exponentially with $\|w\|_\infty$. We leave this to future work.
\end{remark}

\begin{remark}[Boundedness away from zero]
	\label{Rem:PosCov}
	The proof of Theorem \ref{Theo:GPRates} requires that $\rho_0$ be bounded away from zero. This condition similarly underpins previous results for nonparametric Bayesian intensity estimation (in non-covariate-based models), e.g.~\citep{GS13,BSvZ15,KvZ15}. However, it can be circumvented in practice: reasoning similarly to the discussion after Theorem 1 in \citet{BSvZ15}, if $\rho_0$ is not (or not known to be) bounded away from zero, one may modify the observed point patterns $\{N^{(i)}\}_{i=1}^n$ by adding independent standard Poisson processes. The law of the resulting data would then be characterized by a covariate-based intensity equal to $1+\rho_0$, which is bounded below by one and has the same smoothness as $\rho_0$. Using this, the above multi-bandwidth Gaussian process methods could be employed to make inference on the function $1 + \rho_0$ (and therefore also on $\rho_0$), with precise theoretical guarantees provided by Theorem \ref{Theo:GPRates}.
\end{remark}

\begin{remark}[Deterministic covariates]\label{Rem:DetCov}
	Our approach readily allows for the case where both random and deterministic covariates are of interest, say $Z_{\text{rand}}:=\{Z_{\text{rand}}(x),\ x\in\Wcal\}$ and $Z_{\text{det}}:=\{Z_{\text{det}}(x), \ x\in\Wcal\}$, respectively. Letting $Z(x) :=(Z_{\text{rand}}(x),Z_{\text{det}}(x))$ and considering observations $Z^{(i)}(x):=(Z_{\text{rand}}^{(i)}(x),Z_{\text{det}}(x))$, $i=1,\dots,n$, where $Z_{\text{rand}}^{(1)},\dots,$ $Z_{\text{rand}}^{(n)}$ are i.i.d.~copies of $Z_{\text{rand}}$, the posterior distribution is again given by \eqref{Eq:Posterior}, and can be approximately sampled from via the MCMC algorithm from Section \ref{Subsec:Algorithm} below. Furthermore, inspection of the proof of Theorem \ref{Theo:GPRates} shows that its conclusion remains valid in this setting, with the distance $d_Z$ from \eqref{Eq:dZ} now equaling
	$$
	d_Z^2 (\rho_1,\rho_2)
	=\Enorm \int_{\Wcal} \left(\sqrt{\rho_1(Z_{\text{rand}}(x),Z_{\text{det}}(x))} - \sqrt{\rho_2(Z_{\text{rand}}(x),Z_{\text{det}}(x))}\right)^2 dx,
	$$
	with the expectation being with respect to the law of $Z_{\text{rand}}$. Among others, this allows one to study purely spatial effects on the intensity by taking $Z_{\text{det}}(x) = x$. See Section~\ref{Subsec:ExpDetCovariates} in the Supplement for an illustration of this with synthetic data.
\end{remark}

\begin{remark}[Discrete covariates]\label{Rem:DiscrCov}
	Since our primary focus is on Gaussian process methods for  covariate-based intensities, we do not consider in detail the case of discrete covariates, as these would require completely different priors. However, we note that our general concentration result, Theorem~\ref{Theo:GenContrRates} in the Supplement, imposes no restrictions on the covariate space $\Zcal$, and thus can be used to treat this setting as well. In particular, arguing similarly to the proof of Proposition 3.20 in \citet{GKR26} would lead to near-parametric posterior contraction rates under mild conditions on the prior distribution. Combining this with the results derived in the present article, mixed scenarios with both continuous and discrete covariates could be further investigated. 
\end{remark}

%
%
%

\subsection{Posterior contraction rates in the case of stationary covariates}\label{Subsec:StatCov}

Stationarity is a common assumption for the analysis of spatially correlated data, e.g.~\citep{C15}. In the present setting, stationarity of the covariates entails the often realistic scenario, which can be tested (e.g.~\citep{BS17}), where the marginal distribution of the random field $Z$ is homogeneous across the observation window, namely that $Z(x)\sim \nu_Z$ for each $x\in\Wcal$, for some probability measure $\nu_Z$ supported on $[0,1]^d$.

For stationary covariates, the metric $d_Z$ appearing in Theorem \ref{Theo:GPRates} can be made more explicit. Indeed, an application of Fubini's theorem yields, for all $\rho_1,\rho_2\in\Rcal$,
\begin{align*}
	d_Z^2 (\rho_1,\rho_2)
	&=\int_{\Wcal}
	\int_{[0,1]^d}\left(\sqrt{\rho_1(z)} - \sqrt{\rho_2(z)}\right)^2d\nu_Z(z) dx\\
	&=\vol(\Wcal)\|\sqrt{\rho_1} - \sqrt{\rho_2}\|^2_{L^2([0,1]^d,\nu_Z)}.
\end{align*}
Further, if $\nu_Z$ is absolutely continuous with bounded and bounded away from zero p.d.f., then $d_Z$ is equivalent to the standard $L^2$-distance between square-rooted intensities, whereby $\|\rho_1 - \rho_2\|_{L^2}\lesssim (\|\rho_1\|_\infty\vee\|\rho_2\|_\infty)^{1/2}d_Z(\rho_1,\rho_2)$. For example, this is the case if the stationary distribution $\nu_Z$ is known and has independent components, and if we pre-process $\{Z^{(i)}\}_{i=1}^n$ as described in Remark \ref{Rem:BoundCov} via the associated marginal c.d.f.'s, yielding a uniform stationary distribution. When $\nu_Z$ is not known, a pre-processing step involving the empirical c.d.f.~of the covariates is often used in practice, e.g.~\citep{BCST12}.

Under the latter assumptions on $Z$, since $\|\rho\|_\infty\le c_{\rho^*}+\log n$ for all $\rho$ in the support of $\Pi$, the conclusion of Theorem \ref{Theo:GPRates} implies that
$$
\Pi\left(\rho : \|\rho - \rho_0\|_{L^2} > L n^{-\alpha_0/(2\alpha_0+1)}\log^C n
\big|D^{(n)}\right) \to 0,
$$
in $P^{(\infty)}_{\rho_0}$-probability as $n\to\infty$, for all sufficiently large $L>0$ and some large enough $C>0$. The rate $n^{-\alpha_0/(2\alpha_0+1)}$ is known to be minimax optimal, in various statistical models, for estimating, in $L^2$-risk, functions with anisotropic H\"older regularity $\alpha$, including nonparametric regression (e.g.~\citep{N87}) and density estimation (e.g.~\citep{BBM99}). The next proposition shows that a lower bound of the same order holds in the present setting.

\begin{proposition}
	\label{Prop:LowerBound}
	Assume that the covariate random field $Z$ is stationary, with marginal distribution possessing a p.d.f.~that is bounded and bounded away from zero on $[0,1]^d$. Then, for $\alpha\in(0,\infty)^d$ and all sufficiently large $R>0$,
	\[
	\inf_{\hat\rho_n}\sup_{\substack{\rho_0\in C^\alpha([0,1]^d):\\ \rho_0\ge0,\ \|\rho_0\|_{C^\alpha}\le R}}
	E^{(n)}_{\rho_0}\|\hat\rho_n-\rho_0\|_{L^2}
	\gtrsim n^{-\alpha_0/(2\alpha_0+1)},
	\]
	where the infimum is over all estimators $\hat\rho_n=\hat\rho_n(D^{(n)})$ of $\rho$.
\end{proposition}

The proof, which relies on standard lower bound techniques (e.g.~\citep[Chapter 2]{T09}), is given in Section~\ref{Subsec:LowerBound} of the Supplement. Hence, in the case of stationary covariates, the proposed methods adaptively achieve the minimax optimal rate, up to a log-factor.

%
%
%

\subsection{Posterior sampling via a Metropolis-within-Gibbs algorithm}\label{Subsec:Algorithm}

Noting that the posterior distribution from \eqref{Eq:Posterior} is not available in closed form, we build an MCMC algorithm of Metropolis-within-Gibbs type to approximately draw from $\Pi(\cdot|D^{(n)})$. Following the usual MCMC methodology, we then employ the generated samples to concretely compute Bayesian point estimates and credible sets.

A delicate aspect for likelihood-based nonparametric procedures for inhomogeneous point processes is the analytical intractability of the likelihood, since the latter involves an integral of the intensity over the spatial domain $\Wcal$, cf.~\eqref{Eq:Likelihood}, which cannot generally be computed in closed form. In our implementation, we tackle this difficulty by resorting to numerical integration, specifically via piece-wise constant quadrature over a regular partition of $\Wcal$ into pixels. Further, to reduce the computational cost, the covariate values at the quadrature nodes and at the observed points are binned on a fine regular grid of the covariate space $\Zcal = [0,1]^d$, so that the function $w$ from \eqref{Eq:Param} only needs to be evaluated at the bin centers. 
The sensitivity of the posterior to these two approximations, as well as to the discretization of $w$ described below, is investigated in Section~\ref{Subsec:Sensitivity} of the Supplement. In the context of nonparametric Bayesian intensity estimation for models without covariates, more sophisticated methods based on data augmentation have been proposed to handle the resulting `doubly-intractable' posteriors; see~\citet{AMM09}. These could conceivably be adapted to the present covariate-based setting; however, we did not pursue such extensions here. Devising `exact' MCMC samplers for the problem at hand, and comparing their performance to our approach based on numerical likelihood approximations, is an interesting direction for future work.

The employed MCMC algorithm alternates samples from the full conditional distributions, given below, of the quantities $\rho^*,\ell_1,\dots,\ell_d$ and $w$, cf.~Section \ref{Subsec:Prior}. For the functional parameter $w:\Zcal\to\R$, we introduce the `high-dimensional' discretization 
\begin{equation}
	\label{Eq:Discret}
	w(z) = \sum_{v=1}^V w_v\psi_v(z),
	\qquad V\in\N,
	\qquad w_1,\dots,w_V\in\R,
	\qquad z\in \Zcal,
\end{equation}
where $\psi_1,\dots,\psi_V$ are piecewise (multi-)linear interpolation functions associated to a pre-determined regular tensor grid $z_1,\dots,z_V\in \Zcal$. We then have $w(z_v) = w_v$ for all $v=1,\dots,V$, and for any $z\in\Zcal$, the value $w(z)$ is found by (multi-)linearly interpolating the pairs $\{(z_v,w_v)\}_{v=1}^V$. Accordingly, we identify $w$ with the vector $(w_1,\dots,w_V)$, which under $\Pi$, conditionally on $\ell = (\ell_1,\dots,\ell_d)$, is assigned the centered Gaussian prior
\begin{equation}
	\label{Eq:DiscretPrior} 
	w\sim N_V(0,C_\ell),
	\qquad (C_\ell)_{v,v'}= e^{-\sum_{j=1}^d\ell_j^2(z_{v,j} - z_{v',j})^2},
	\qquad v,v'=1,\dots,V.
\end{equation}
Since the kernel \eqref{Eq:ARD} is separable across coordinates, $C_\ell$ factorizes over regular tensor grids as the Kronecker product of $d$ one-dimensional covariance matrices, which substantially reduces the cost of the required linear algebra operations involving $C_\ell$.

Starting from some initialization (which we draw from the prior), and given the current state of the chain, each step of the algorithm alternates samples from:
\begin{enumerate}
	\item The full conditional distribution of the intensity ceiling $\rho^*$,
	\begin{align*}
		\Pi(\tilde\rho^{*} | D^{(n)},&\ell_1,\dots,\ell_d,w)\\
		& \propto \Pi_{\rho^{*}}(\tilde\rho^*)
		\prod_{i=1}^n e^{\sum_{k=1}^{K^{(i)}}\log(\tilde\rho^*\sigma(w(Z^{(i)}(X^{(i)}_k))))-\int_{\Wcal} \tilde\rho^*\sigma(w(Z^{(i)}(x)))dx}
		\\
		& \propto \Pi_{\rho^{*}}(\tilde\rho^{*}) (\tilde\rho^{*})^{\sum_{i = 1}^n K^{(i)}} e^{ - \tilde\rho^*  \int_{\Wcal} \sum_{i = 1}^n\sigma(w(Z^{(i)}(x)))dx},
	\end{align*}
	which, since $\Pi_{\rho^*}$ is truncated $\Gamma(a_{\rho^*}, b_{\rho^*})$ on $[0,c_{\rho^*}+\log n]$, is again truncated gamma with updated parameters $a_{\rho^*} + \sum_{i = 1}^n K^{(i)}$ and $b_{\rho^*} +  \int_{\Wcal}\sum_{i = 1}^n \sigma(w(Z^{(i)}(x)))dx$. The latter integral is efficiently computed in practice as described above.

	\item The full conditional distribution of each length-scale $\ell_j$, $j=1,\dots,d$. To improve mixing, we employ the non-centered parametrization $w=C_\ell^{1/2}\eta$, where $C_\ell^{1/2}$ is the Cholesky factor of $C_\ell$ and $\eta\sim N_V(0,I_V)$ is independent of $\ell$; see \citet{PRS07}. Keeping $\eta$ fixed, the length-scales are updated sequentially from
	\begin{equation}
		\label{Eq:FullCondEll}    
		\begin{split}
			\Pi(\tilde\ell_j | D^{(n)},\rho^*,&\ell_1,\dots,\ell_{j-1},\ell_{j+1},\dots,\ell_d,\eta)\\ 
			&\propto 
			L^{(n)}\big(\rho^*\sigma\circ (C_{\tilde\ell}^{1/2}\eta)\big)\,
			\tilde\ell_j^{a_\ell-1}e^{-(b_\ell+\lambda\prod_{k\neq j}(1+\ell_k))\tilde\ell_j},
		\end{split}
	\end{equation}
	where $\tilde\ell:=(\ell_1,\dots,\ell_{j-1},\tilde\ell_j,\ell_{j+1},\dots,\ell_d)$, and where the last two factors stem from the gamma full conditionals of $\Pi_\ell$ from \eqref{Eq:EllFullCond}. Above, $L^{(n)}$ is the likelihood from \eqref{Eq:Likelihood}, and we have used the notation, for any $\tilde w = (\tilde w_1,\dots,\tilde w_V) \in\R^V$,
	$$
	(\rho^*\sigma\circ \tilde w) (z)= \rho^*\sigma\Big(\sum_{v=1}^V \tilde w_v\psi_v(z) \Big),
	\qquad z\in\Zcal.
	$$
	Sampling from \eqref{Eq:FullCondEll} is achieved via random walk Metropolis-Hastings on the log scale, whose step-size is tuned during the burn-in phase and then kept fixed. Upon acceptance, $w$ is updated to $C_{\tilde\ell}^{1/2}\eta$ in view of the non-centered parametrization.

	\item The full conditional distribution of the high-dimensional parameter $w$,
	\begin{equation}
		\label{Eq:FullCondW}
		\Pi(\tilde w|D^{(n)},\rho^*,\ell_1,\dots,\ell_d)
		\propto L^{(n)}(\rho^*\sigma\circ \tilde w)
		\textnormal{det}^{-1/2}(C_\ell)e^{-\frac{1}{2}\tilde w^T(C_\ell)^{-1}\tilde w}.
	\end{equation}
	Approximate samples from \eqref{Eq:FullCondW} are obtained via the `pre-conditioned Crank-Nicolson' (pCN) algorithm, which is a dimension-robust MCMC sampling method of Metropolis-Hastings type, commonly used in inverse problems and data assimilation; see \citet{CSRW13}. It generates an $\R^V$-valued Markov chain $\{\omega_u, \ u\in\N\}$ through the repetition of the following two operations:
	\begin{itemize}
		\item Draw a prior sample $\xi\sim N_V(0,C_{\ell})$ and construct the proposal $\omega:=\sqrt{1-2\zeta}$ $\omega_{u-1} + \sqrt{2\zeta}\xi$, where $\zeta\in(0,1/2)$ is the step-size, tuned during burn-in.
		\item Construct the new element in the pCN chain by
		$$
		\omega_{u}:=
		\begin{cases}
			\omega, & \textnormal{with probability}\ 1\land \frac{L^{(n)}(\rho^*\sigma\circ \omega)}
			{L^{(n)}(\rho^*\sigma\circ \omega_{u-1})},\\
			\omega_{u-1}, & \text{otherwise}.
		\end{cases}
		$$
	\end{itemize}
	The resulting pCN chain can be shown to be reversible and to have stationary measure equal to the full conditional distribution \eqref{Eq:FullCondW}, e.g.~\citep[Proposition 1.2.2]{N23}. Further, the pCN acceptance probabilities are known to be stable with respect to the discretization dimension \citep{CSRW13}, implying desirable mixing properties for statistical applications with functional unknowns \citep{HSV14}. In each iteration of the overall Metropolis-within-Gibbs algorithm, we initialize the pCN chain at the current value of $w$ and perform several (10 in our implementation) pCN moves, which are computationally inexpensive relative to the length-scale updates.
	
\end{enumerate}

%
%
%
%
%

\section{Simulation studies}\label{Sec:Simulations}

We assess our approach in extensive numerical simulations. We take the centered unit square $\Wcal = [-1/2,1/2]^2$ as the observation window, fix the ground truth $\rho_0$ and, for $i=1,\dots,n$, draw an independent realization $Z^{(i)}$ of a $d$-variate random field $Z$, conditionally on which we sample the point pattern $N^{(i)}$. We consider univariate and bivariate scenarios, with sample sizes $n\in\{50,250,500,1000\}$, simulating, for each combination, 50 independent datasets over which the results are averaged. We then implement posterior inference via the Metropolis-within-Gibbs MCMC algorithm described in Section \ref{Subsec:Algorithm}, using the posterior mean $\hat\rho_\Pi^{(n)}:=E^\Pi[\rho|D^{(n)}]$ as the point estimate.

We compare the obtained results to the performance of alternative kernel-type and parametric methods, which are standard approaches in spatial statistics; see e.g.~the monograph \citep{BRT16}. To our knowledge, the frequentist literature on nonparametric covariate-based intensity estimation mostly considers settings with a single observation of the covariates and points (possibly over a large domain or under an increasing intensity assumption), e.g.~\citep{G08,BCST12,BGMMM20}. There appears to be no standard approach for repeated observations, despite interest in this case having been raised since at least \citep{DLB91}. An overview of aggregation strategies was presented by \citet[Chapter 4]{IPSS08}. Following the latter, we consider a simple average of individual kernel estimators, 
\begin{equation}
	\label{Eq:AverageKernel}
	\hat\rho_\kappa(z) = \frac{1}{n}\sum_{i=1}^n
	\hat \rho_\kappa^{(i)}(z),
	\qquad z\in\Zcal,
\end{equation}
where each $\hat\rho_\kappa^{(i)}$ is defined by a ratio-type kernel estimator following \citet{G08},
\begin{equation}
	\label{Eq:IndividualKernels}
	\hat \rho_\kappa^{(i)}(z) 
	= \frac{1}{g^{(i)}(z)} 
	\sum_{k = 1}^{K^{(i)}} \kappa (Z^{(i)}(X^{(i)}_k) - z),
	\qquad z\in\Zcal.
\end{equation}
Above, $\kappa$ is a $d$-dimensional smoothing kernel and $g^{(i)}(z)=\int_{\Wcal}\kappa\!\left(Z^{(i)}(x)-z\right)\,dx$ is the kernel-smoothed, non-normalized spatial occupation density of $Z^{(i)}$. See \citep{G08,BCST12,BGMMM20} for similar procedures. In the experiments, $\kappa$ is a Gaussian product kernel with bandwidths $h=(h_1,\dots,h_d)$, selected over a pre-determined grid of values either by likelihood cross-validation across replicates, or by minimizing the $L^2$-error to $\rho_0$; see Section~\ref{Sec:MoreSym} of the Supplement for details. The latter `oracle' choice requires knowledge of the ground truth, and thus provides a benchmark that is not implementable in practice.

As parametric competitors, we consider log-polynomial Poisson regression models $\rho(z)=e^{P_k(z)}$, $z\in\Zcal$, with $P_k$ a (full) polynomial of total degree $k\in\{1,2,3,4\}$. For these, we estimate the coefficients by maximizing the likelihood \eqref{Eq:Likelihood}, cf.~\citep[Chapter 9]{BRT16}, where the integral over $\Wcal$ is approximated via the quadrature scheme from above.

%
%
%

\subsection{Data generation}
\label{Subsec:DataGen}

For the univariate setting, we take $Z$ as a (centered) Gaussian process with squared-exponential covariance $e^{-|x-x'|^2/(2\cdot 0.005)}$, $x,x'\in \Wcal = [-1/2,1/2]^2$, transformed via the standard normal c.d.f.~as described in Remark \ref{Rem:BoundCov}. With this choice, $Z$ is supported on $\Zcal = [0,1]$, is stationary, and has invariant measure equal to the uniform distribution on $[0,1]$, falling within the framework of Section \ref{Subsec:StatCov}. The ground truth is set to be proportional to the restriction on $[0,1]$ of a univariate skew normal p.d.f.,
\begin{equation}
	\label{Eq:1DTruth}
	\rho_0(z) = 5f_{SN}(z; 0.8, 0.3 ,-5),
	\qquad z\in[0,1].
\end{equation}
This results in point patterns concentrated around the regions with covariate value near $0.7$; see Figure \ref{Fig:1DPattern}. The expected number of points per observation is (slightly below) 5.

\begin{figure}[H]
	\centering
	\includegraphics[width=\linewidth]{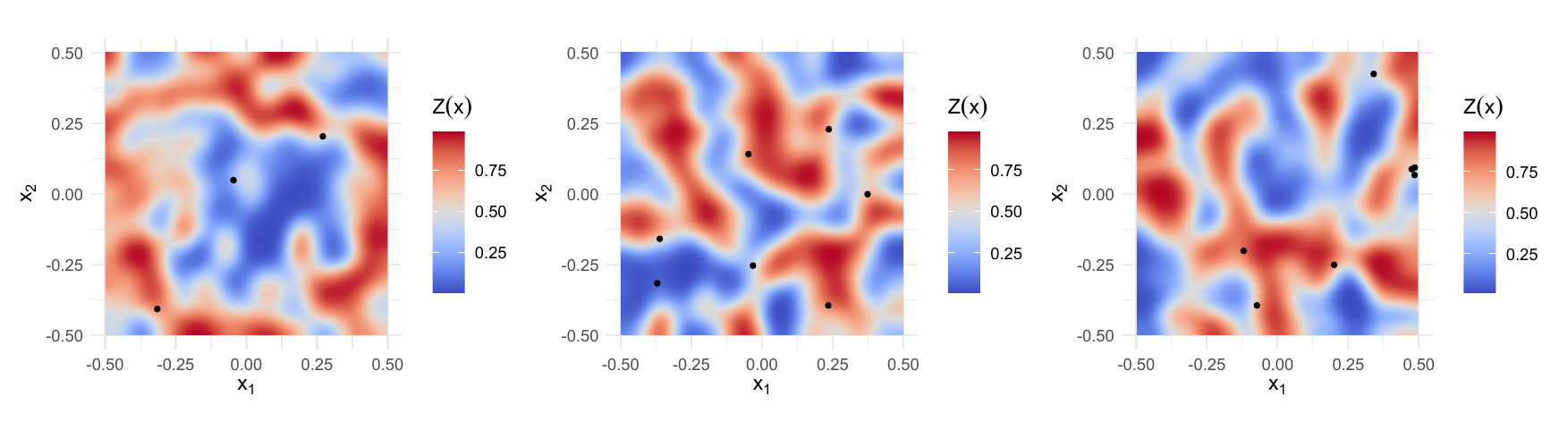}
	\caption{Independent draws of the covariate and the point pattern with intensity \eqref{Eq:1DTruth}.}
	\label{Fig:1DPattern}
\end{figure}

In the bivariate scenario, we take $Z=(Z_1,Z_2)$, with $Z_1$ as above and $Z_2$ an independent process of the same type with covariance $e^{-|x-x'|^2/(2\cdot 0.05)}$; see Figure \ref{Fig:2DPattern} in the Supplement for a visual comparison of the different length-scales. The ground truth is
\begin{equation}
	\label{Eq:2DTruth}
	\rho_0(z) = 10 - 10\, g_{(0.8,0.3)}(z) + 10\, g_{(0.3,0.8)}(z),
	\qquad z\in[0,1]^2,
\end{equation}
where $g_\mu(z) := e^{-\frac12 (z-\mu)^T\Sigma^{-1}(z-\mu)}$ with $\Sigma = \textnormal{diag}(0.08^2,0.5^2)$. The resulting intensity is anisotropic, with noticeably smaller characteristic length-scales in the first argument; cf.~Figure \ref{Fig:2DResults}. The expected number of points per observation is approximately 10.

In both settings, $Z$ is simulated as a pixel-wise constant field by sampling its values at the $50\times50$ pixel centers of $\Wcal$. Conditionally on $Z$, the point patterns are obtained by sampling, in each pixel, an independent Poisson number of points with mean equal to the integral of $\rho_0(Z(x))$ over the pixel, and placing these points uniformly within it.

%
%
%

\subsection{Computation}
\label{Subsec:SimComputation}

For each dataset, we run a single instance of 60,000 iterations of the MCMC algorithm from Section \ref{Subsec:Algorithm}, initialized at a prior draw, discarding the first 20,000 as burn-in and retaining every tenth of the remaining ones. During the burn-in, the step-sizes of the pCN and length-scale updates are adapted by a Robbins--Monro scheme on the logarithmic scale \citep{AT08}, targeting acceptance rates of $30\%$ and $44\%$, respectively, and are then kept fixed. The adequacy of this sampling framework was validated against longer multi-chain runs on a
sub-selection of the datasets. The posterior means differed from those of the long runs by at most
$0.7\%$ (in relative $L^2$-norm) and the relative estimation errors by at most $0.005$ in absolute value (at most $3.7\%$ of their value); see Section~\ref{S:Sec:Diagnostics} of the Supplement. In particular, across
all datasets, the Monte Carlo standard error of the posterior mean, in $L^2$-norm and relative to the
norm of the posterior mean, ranged between $0.07\%$ and $0.8\%$, negligible
compared to the estimation errors reported below.

Compared to the intensity $\rho$, the hyper-parameters, particularly the length-scales, mix more slowly, as is typical for hierarchical Gaussian priors where they are strongly coupled to the latent function \citep{PRS07}. To obtain reliable posterior summaries for them, we additionally employed long runs of two independent chains on five datasets for each setting and sample size. Each chain was initially run for 100,000 iterations in the univariate
setting and 200,000 in the bivariate one (250,000 for $n=1000$), discarding the first 50,000 as burn-in
and retaining every tenth of the remaining ones. The chains were extended in blocks of 50,000 iterations until the rank-normalized $\hat R$ was below $1.01$ and the bulk and tail effective sample sizes (ESS) were at least $200$ \citep{VGSCB21} for all monitored quantities, namely $\rho^*$, $\ell$, the log-likelihood, the mean intensity and the values of $\rho$ at a set of fixed covariates. The resulting lengths range between 100,000 and 300,000.

Numerical integration over $\Wcal$ is performed via piece-wise constant quadrature using a uniform square grid with 2500 nodes. The discretization \eqref{Eq:Discret} is based on $V=200$ equally spaced nodes for $d=1$ and on a $20\times 20$ tensor grid for $d=2$, and the covariate values are binned on $2000$ and $200\times 200$ bins, respectively. Owing to this, the computational times grow mildly in $n$, although the total number of events increases on average from 250 to 5000 (when $d=1$) and from 500 to 10000 ($d=2$). All experiments were carried out in \texttt{R} on a MacBook Pro with M1 Processor (8 cores) and 8 GB of RAM. The individual 60,000-iteration chains required a median of 1.7 minutes per dataset in the univariate setting and 10.7 minutes in the bivariate one.

The hyper-parameters of the prior are set to $a_{\rho^*}=1$, $b_{\rho^*}=0.1$, $c_{\rho^*}=100$,
$a_\ell=1$, $b_\ell=2.5$ and $\lambda=0.1$ throughout. The sensitivity of the results to these choices, as well as to the numerical approximations described above, is investigated in Section~\ref{Subsec:Sensitivity} of the Supplement. The accuracy of an estimator $\hat \rho$ is measured by the relative $L^2$-error $\|\hat\rho-\rho_0\|_{L^2}/\|\rho_0\|_{L^2}$. Trace plots, acceptance rates, effective sample sizes and further diagnostics are reported in Section \ref{S:Sec:Diagnostics} of the Supplement.

%
%
%

\subsection{Results}
\label{Subsec:SimResults}

Figures \ref{Fig:1DResults} and \ref{Fig:2DResults} display the obtained posterior means $\hat\rho_\Pi^{(n)}$ on representative datasets for the univariate and bivariate scenarios (with pointwise $95\%$-credible intervals in the former case), under increasing sample sizes. For visual comparison, kernel estimates, with both cross-validated and oracle bandwidth choices, and log-polynomial fits of degree 4 are also included. In line with Theorem \ref{Theo:GPRates}, the posterior appears to concentrate around $\rho_0$ as $n$ increases. The kernel estimates display some spurious oscillations, while the log-polynomial models cannot fully adapt to the shape of the ground truth, which is particularly visible near the low-intensity region of the bivariate setting.

\begin{figure}[H]
	\centering
	\includegraphics[width=\linewidth]{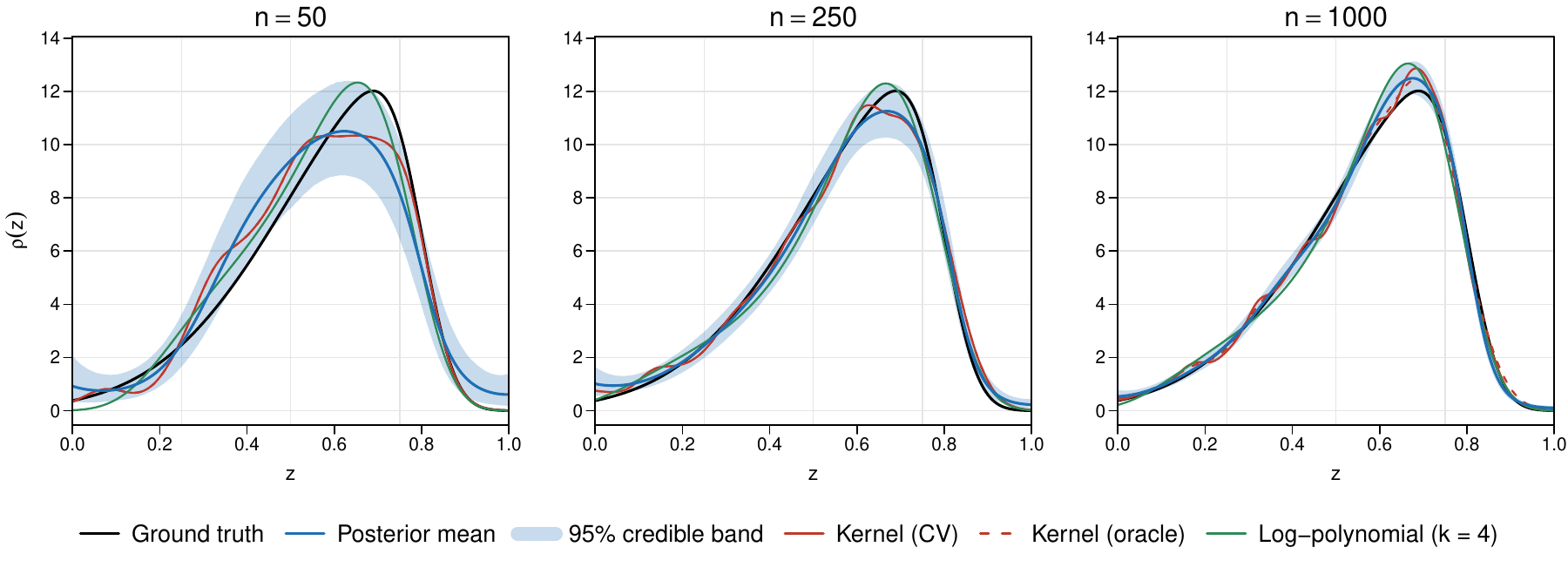}
	\caption{Univariate setting, $n = 50, 250, 1000$: posterior mean with pointwise $95\%$-credible intervals, kernel estimates with oracle and cross-validated bandwidths, and log-polynomial estimate of degree 4. The ground truth \eqref{Eq:1DTruth} is shown in black.}
	\label{Fig:1DResults}
\end{figure}

\begin{figure}[H]
	\centering
	\includegraphics[width=\linewidth]{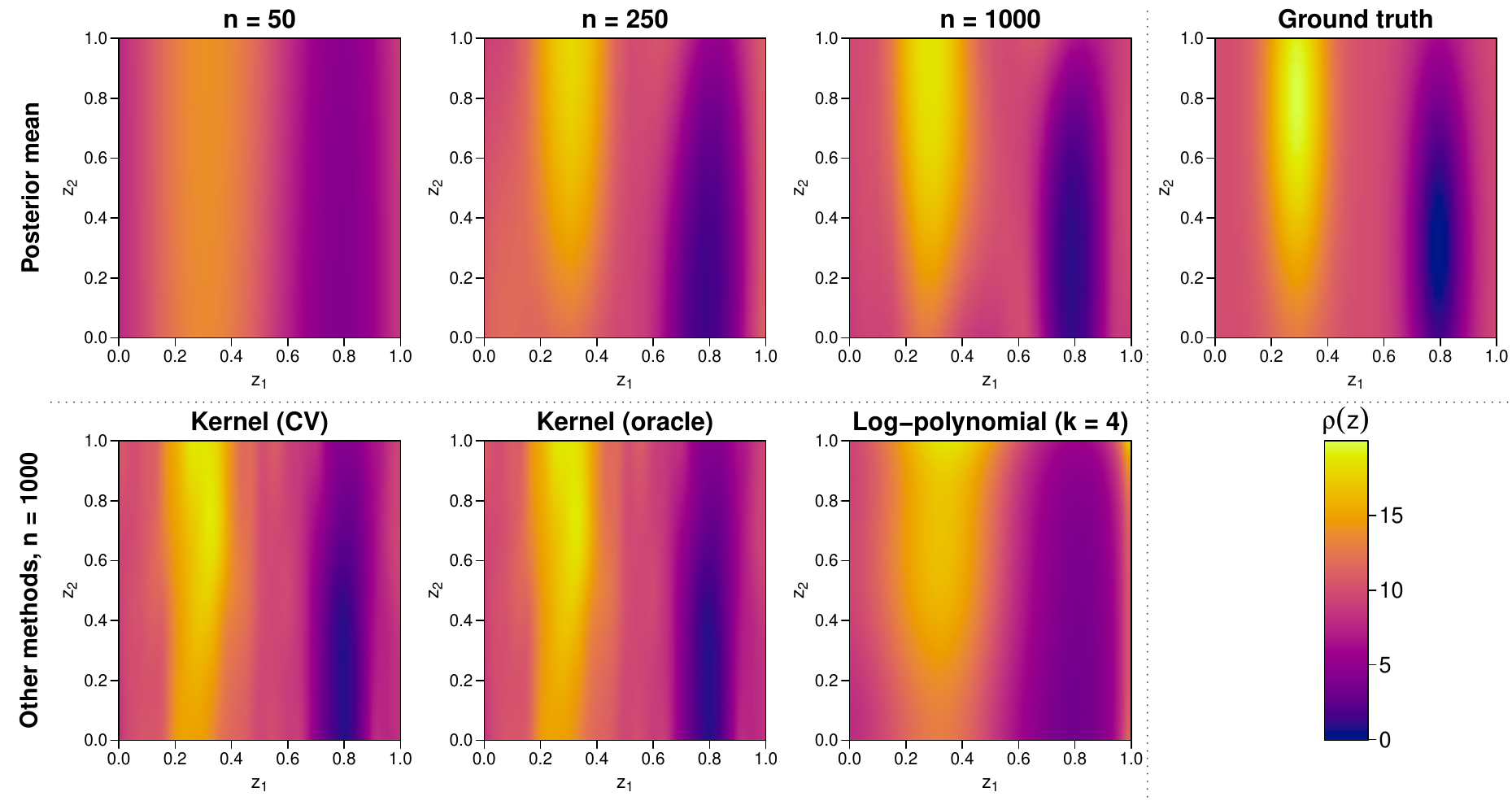}
	\caption{Bivariate setting. Top row: posterior means for $n = 50, 250, 1000$ and the ground truth \eqref{Eq:2DTruth}. Bottom row: kernel estimates with cross-validated and oracle bandwidths, and log-polynomial estimate of degree 4, for $n=1000$.}
	\label{Fig:2DResults}
\end{figure}

Table \ref{Tab:Errors} reports the relative $L^2$-errors of all methods, averaged over the 50 independent datasets simulated for each scenario and sample size. From $n=250$ onwards, the posterior mean attains the lowest error in both the univariate and bivariate settings, also when compared to the kernel estimator with oracle bandwidth. For $n=50$, the kernel estimators and, in the univariate scenario, the log-polynomial model of degree 4 attain lower errors. With few observations, the hyper-prior $\Pi_\ell$, which favours smoother functions, induces a bias towards flatter estimates; this effect vanishes for larger $n$, and is reduced by a smaller rate parameter $b_\ell$, cf.~Section~\ref{Subsec:Sensitivity}. The errors of the log-polynomial models saturate around a level reflecting their misspecification, whereas the error of the posterior mean keeps decreasing, as expected in view of Theorem \ref{Theo:GPRates}.

\begin{table}[H]
	\centering
	\small
	\begin{tabular}{l|cccc}
		& $n=50$ & $n=250$ & $n=500$ & $n=1000$ \\
		\hline
		\noalign{\vskip 3pt}
		\multicolumn{5}{c}{\textbf{Univariate setting}}\\
		\noalign{\vskip 3pt}
		Posterior mean & 0.145 (0.041) & \textbf{0.055} (0.013) & \textbf{0.046} (0.015) & \textbf{0.034} (0.008) \\
		\noalign{{\color{gray!60}\hrule height 0.3pt}}
		Kernel, CV & 0.129 (0.036) & 0.070 (0.018) & 0.059 (0.014) & 0.042 (0.009) \\
		Kernel, oracle & 0.120 (0.035) & 0.064 (0.014) & 0.053 (0.012) & 0.040 (0.009) \\
		\noalign{{\color{gray!60}\hrule height 0.3pt}}
		Log-polynomial, $k=1$ & 0.625 (0.006) & 0.624 (0.003) & 0.623 (0.002) & 0.624 (0.002) \\
		Log-polynomial, $k=2$ & 0.256 (0.023) & 0.240 (0.010) & 0.242 (0.007) & 0.240 (0.006) \\
		Log-polynomial, $k=3$ & 0.166 (0.031) & 0.143 (0.012) & 0.143 (0.010) & 0.141 (0.006) \\
		Log-polynomial, $k=4$ & \textbf{0.114} (0.032) & 0.072 (0.011) & 0.070 (0.015) & 0.064 (0.007) \\
		\hline
		\noalign{\vskip 3pt}
		\multicolumn{5}{c}{\textbf{Bivariate setting}}\\
		\noalign{\vskip 3pt}
		Posterior mean & 0.178 (0.022) & \textbf{0.090} (0.012) & \textbf{0.069} (0.008) & \textbf{0.052} (0.006) \\
		\noalign{{\color{gray!60}\hrule height 0.3pt}}
		Kernel, CV & 0.162 (0.027) & 0.102 (0.014) & 0.084 (0.009) & 0.066 (0.007) \\
		Kernel, oracle & \textbf{0.146} (0.021) & 0.095 (0.011) & 0.078 (0.008) & 0.064 (0.007) \\
		\noalign{{\color{gray!60}\hrule height 0.3pt}}
		Log-polynomial, $k=1$ & 0.324 (0.009) & 0.318 (0.003) & 0.317 (0.002) & 0.316 (0.001) \\
		Log-polynomial, $k=2$ & 0.284 (0.014) & 0.267 (0.005) & 0.264 (0.004) & 0.262 (0.003) \\
		Log-polynomial, $k=3$ & 0.223 (0.019) & 0.183 (0.005) & 0.179 (0.003) & 0.175 (0.002) \\
		Log-polynomial, $k=4$ & 0.224 (0.022) & 0.164 (0.006) & 0.156 (0.004) & 0.151 (0.002) \\
		\hline
	\end{tabular}
	\caption{Mean relative $L^2$-errors (and between-dataset standard deviations) over 50 independent datasets for each setting and sample size.}
	\label{Tab:Errors}
\end{table}

Finally, Table \ref{Tab:Hyper} reports the posterior summaries of the hyper-parameters obtained from the long runs described in Section \ref{Subsec:SimComputation}. In the bivariate setting, from $n=250$ onwards the estimate of $\ell_1$ is six to eight times larger than that of $\ell_2$, reflecting the anisotropy of the ground truth \eqref{Eq:2DTruth}, whose first argument varies on a scale about six times smaller than the second. Across the experiments, the estimates of the intensity ceiling $\rho^*$ lie above, and relatively close to, the maximum of $\rho_0$, which equals approximately $12$ in the univariate setting and $20$ in the bivariate one.

\begin{table}[H]
	\centering
	\small
	\begin{tabular}{ll|cccc}
		& $n$ & $\ell_1$ & $\ell_2$ & $\rho^*$ & min.~ESS / max.~$\hat R$ \\
		\hline
		\textbf{Univariate} & 50 & 3.9 & -- & 13.3 & 459 / 1.007 \\
		\textbf{setting}& 250 & 4.2 & -- & 14.5 & 206 / 1.010 \\
		& 500 & 4.5 & -- & 14.1 & 201 / 1.007 \\
		& 1000 & 4.4 & -- & 15.9 & 205 / 1.009 \\
		\hline
		\textbf{Bivariate} & 50 & 3.1 & 0.5 & 22.2 & 513 / 1.003 \\
		\textbf{setting}& 250 & 6.0 & 0.7 & 21.5 & 507 / 1.009 \\
		& 500 & 6.1 & 0.9 & 21.4 & 249 / 1.009 \\
		& 1000 & 6.7 & 1.1 & 21.6 & 202 / 1.010 \\
		\hline
	\end{tabular}
	\caption{Medians over 5 datasets of the posterior means of the hyper-parameters. Last column: minimal bulk ESS and maximal rank-normalized $\hat R$ of the hyper-parameters.}
	\label{Tab:Hyper}
\end{table}

%
%
%
%
%

\section{Applications to a Canadian wildfire dataset}\label{Sec:RealData}

The study of the distribution of wildfires and their relationship with environmental factors is well established in the spatial statistics community, with recent contributions by \citet{JMS12,BGMMM20,KPDO23}, among others. This literature highlights that the spread of wildfires is heavily influenced by meteorological conditions such as high temperatures, prolonged dry periods and moderate-to-strong winds.

In this section, we analyze data from the Canadian Wildland Fire Information System (\url{https://cwfis.cfs.nrcan.gc.ca/home}), which provides detailed daily records over the last two decades, comprising the geographical coordinates of the hotspots and complete environmental information. We extracted from this large dataset annual recordings from 2004 to 2022 of the locations of the wildfires in June (the month of peak activity in Canada, cf.~\citep{BGMMM20}), alongside pixel-wise monthly average temperatures, precipitation levels and wind speeds. We focused on a few selected provinces, specifically Ontario, in the eastern part of Canada, British Columbia, on the western coast, and Saskatchewan, in the central region. Here, we present the results for Ontario, deferring the rest of the analysis to Section \ref{Sec:MoreRealData} of the Supplement.

\begin{figure}[H]
	\centering
	\begin{subfigure}[b]{\textwidth}
		\centering
		\includegraphics[width=\linewidth]{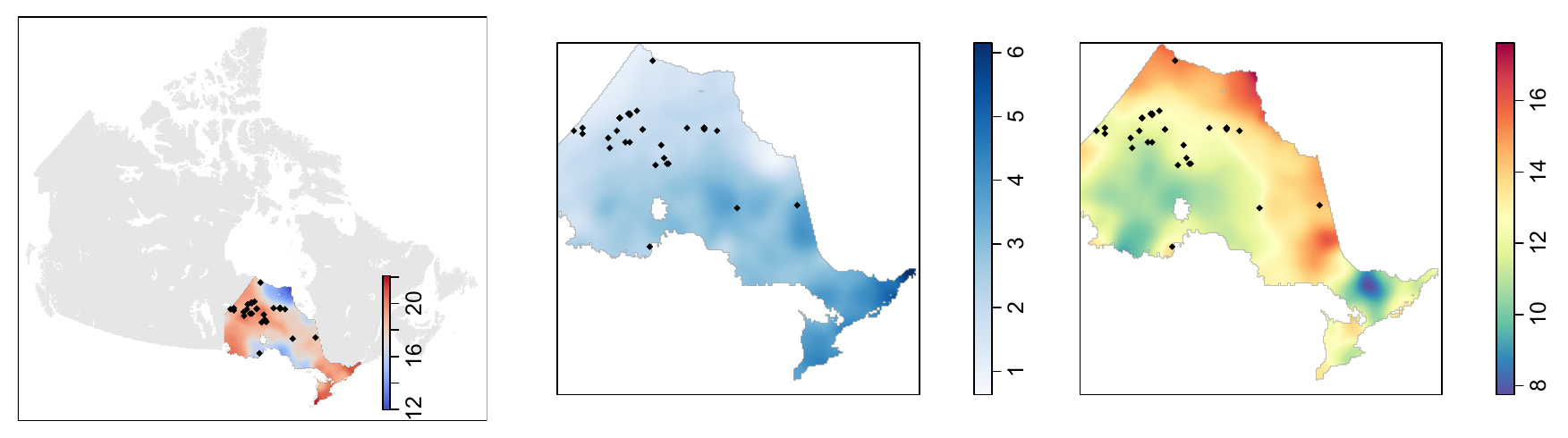}
	\end{subfigure}
	\begin{subfigure}[b]{\textwidth}
		\centering
		\includegraphics[width=\linewidth]{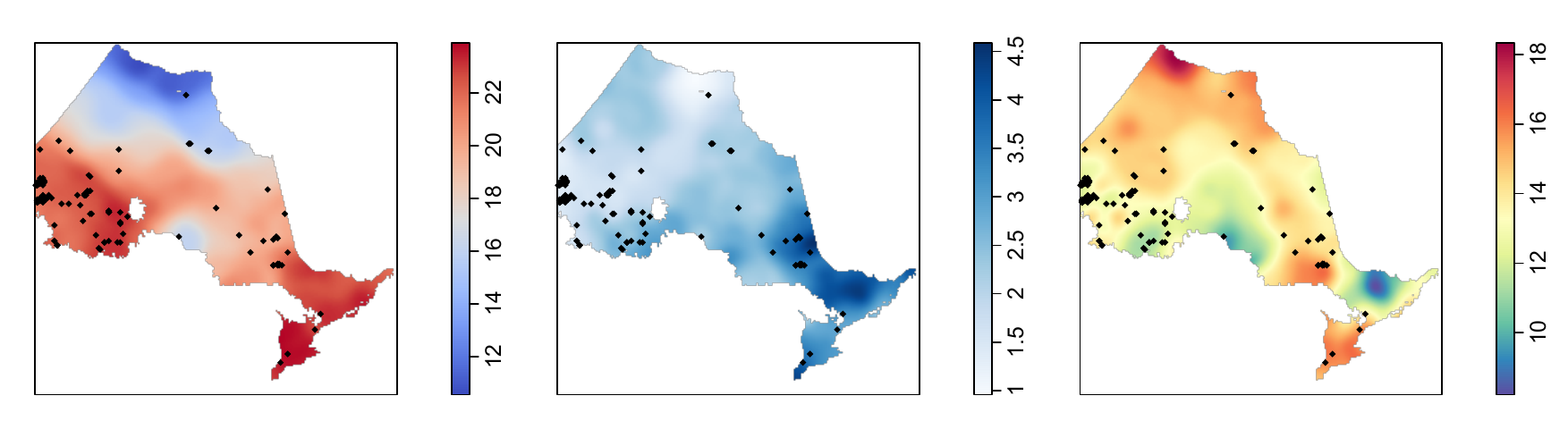}
	\end{subfigure}
	\caption{Top row, left to right: Average temperatures (in Celsius), precipitation (in mm) and wind speeds (in km/h) in Ontario during June 2013. Bottom row: Data for 2021. Black dots indicate the wildfires (respectively, $34$ and $130$ in total).}     
	\label{Fig:Temp}
\end{figure}

The Ontario dataset comprises $n=19$ spatial point patterns $\{N^{(i)}\}_{i=2004}^{2022}$ representing the wildfire locations, cf.~Figure \ref{Fig:Temp}, and the same number of three-dimensional spatial covariate fields $\{Z^{(i)}\}_{i=2004}^{2022}$, where $Z^{(i)}=(Z^{(i)}_{\text{temp}},Z^{(i)}_{\text{prec}},Z^{(i)}_{\text{wind}})$. 
The data display some strong variability, with the number of wildfires ranging from 2 (in June 2004) to 130 (in 2021) and a wide spectrum of covariate values. Another distinctive characteristic is that the covariates exhibit heterogeneous behaviors: while the temperature fields mostly change smoothly over space, precipitation and winds tend to display more abrupt variations, possibly as a result of currents, orographic features and other environmental factors. This heterogeneity motivates the use of a multi-bandwidth method.

%
%
%

\subsection{Exploratory univariate analysis}\label{Subsec:1DOntario}

For a preliminary analysis, Figure \ref{Fig:1DOntario} shows the posterior means $\hat\rho_{\Pi,\text{temp}}, \hat\rho_{\Pi,\text{prec}}, \hat\rho_{\Pi,\text{wind}}$ obtained using each covariate individually. In line with the literature, the results capture a positive association between temperature and wildfire risk, with a steep rise between about $16^\circ$C and $22^\circ$C. A strong negative impact is inferred for the precipitation level, with the intensity dropping between about 1.5 and 3 mm, while windy conditions appear to increase the intensity only around intermediate speeds of about $13$ km/h. The figure also includes the averaged kernel estimates \eqref{Eq:AverageKernel}, with bandwidths selected by cross-validation across the years as in Section~\ref{Sec:Simulations}. These are in general agreement with the posterior means, but display a more erratic behavior, with spurious oscillations between $15^\circ$C and $21^\circ$C and boundary effects at the extremes of the covariate ranges.

\begin{figure}[H]
	\centering
	\includegraphics[width=\linewidth]{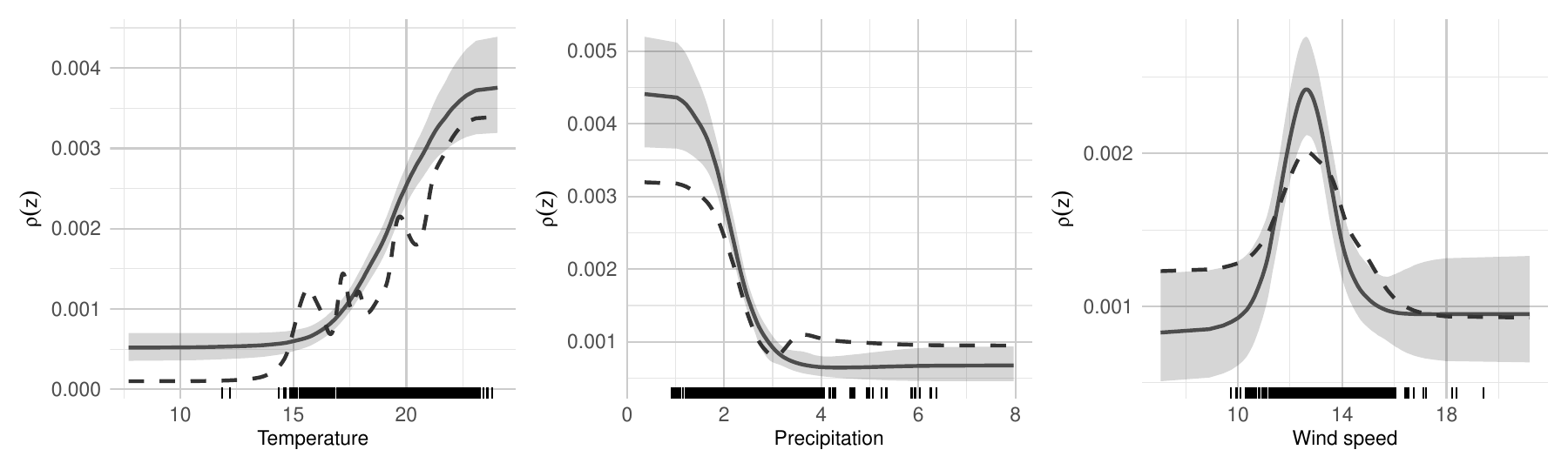}
	\caption{
		Left to right: Posterior means (solid line) and point-wise $95\%$-credible intervals (shaded region) for the intensity as a function of the average temperature, precipitation and wind speed. The dashed lines show the kernel estimates with cross-validated bandwidths, and the ticks mark the covariate values at the observed wildfires.}
	\label{Fig:1DOntario}
\end{figure}

In all analyses, each covariate was mapped onto $[0,1]$ via the empirical c.d.f.~of its pooled values over the 19 years, and the estimates were transformed back to the original scale for display. The hyper-parameters of the prior were set as in the simulation study of Section~\ref{Subsec:SimComputation} ($a_{\rho^*}=1$, $b_{\rho^*}=0.1$, $c_{\rho^*}=100$,
$a_\ell=1$, $b_\ell=2.5$ and $\lambda=0.1$). The latent function discretization \eqref{Eq:Discret} was based on $V=200$ equally spaced nodes, and the covariate values were binned on 2000 bins. For each fit, we ran two chains of 60,000 iterations, discarding the first 20,000 as burn-in and retaining every tenth of the remaining ones. The chains were extended in blocks of 50,000 iterations until the rank-normalized $\hat R$ was below $1.01$ and the bulk and tail ESS were at least $200$ for $\rho^*$, $\ell$, the log-likelihood and the values of $\rho$ at a fixed grid of covariate values, and the Monte Carlo standard error of the posterior mean at these points was below $1\%$ of its norm. The resulting lengths range between 60,000 and 210,000 iterations. On the same hardware as for the simulation study presented in Section~\ref{Subsec:SimComputation}, the computation time was about 1.6 seconds per 1000 iterations, i.e.~between 1.5 and 6 minutes per chain.

%
%
%

\subsection{Full analysis}\label{Subsec:3DOntario}

Next, we present the full analysis based on the joint information on temperature, precipitation and wind. For ease of visualization, Figure~\ref{Fig:3DOntarioMarginals} shows two-dimensional slices of the obtained posterior mean, with the wind speed fixed at its $0.05$-quantile, median and $0.95$-quantile (10.72, 13.50 and 16.21 km/h, respectively). These reinforce the findings of the exploratory analysis, with the greatest intensities associated with higher temperatures (above $19^\circ$C) and drier conditions (precipitation below 2 mm). The intensity is highest at the median wind speed, in agreement with Figure~\ref{Fig:1DOntario} (right), and lower at both extreme quantiles; for strong winds, relatively high intensities are also detected at high temperatures despite heavy precipitation. Kernel estimates are not reported for the three-dimensional model; a comparison with the kernel estimator for the bivariate model based on temperature and precipitation is given in Section~\ref{Sec:MoreRealData} of the Supplement.

\begin{figure}[H]
	\centering
	\includegraphics[width=\linewidth]{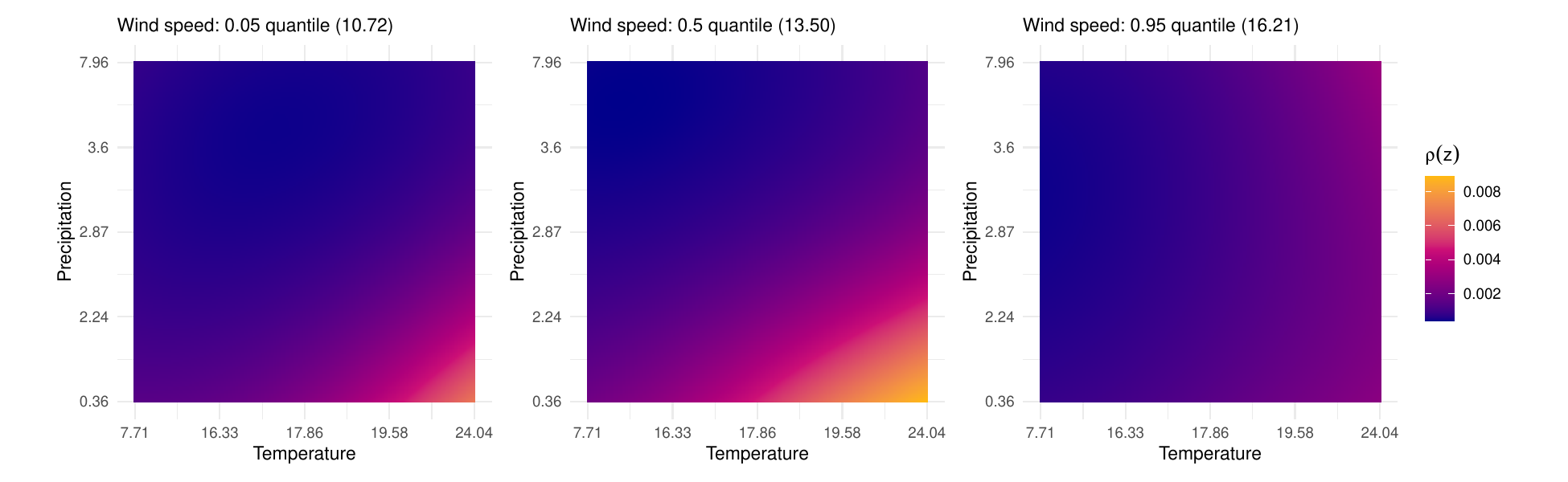}
	\caption{Left to right: Posterior mean of the wildfire intensity as a function of the average temperature and precipitation level, with the average wind speed fixed at its $0.05$-quantile (10.72 km/h), median (13.50 km/h) and $0.95$-quantile (16.21 km/h), respectively. The axes are on the transformed scale, labelled with the original values.}
	\label{Fig:3DOntarioMarginals}
\end{figure}

Figure~\ref{Fig:3DOntationSpatial} displays, for the years 2013, 2015 and 2021, the plug-in posterior means $\hat\lambda^{(i)}:=\hat\rho_\Pi\circ Z^{(i)}$ of the spatial intensity, obtained by composing the posterior mean $\hat\rho_\Pi$ with the covariate fields of the respective years. While $\hat\rho_\Pi$ is based on the combined information from 2004 to 2022, the yearly variability of the covariates yields different spatial intensity estimates, which reproduce the main features of the observed point patterns and capture year-specific trends, such as the concentration of wildfires in the north-western part of the province in 2013 and in the western part in 2021. 

\begin{figure}[H]
	\centering
	\includegraphics[width=\linewidth]{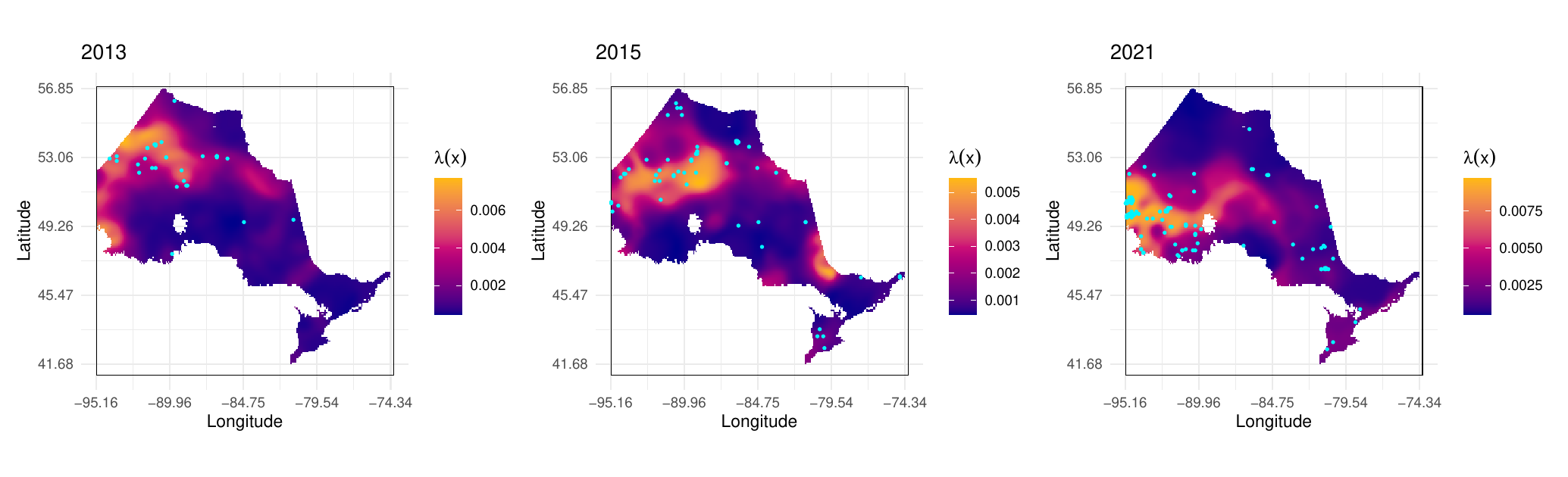}
	\caption{Left to right: Plug-in posterior means of the spatial intensity based on average temperature, precipitation level and wind speed, for the years $2013, 2015$ and $2021$.}
	\label{Fig:3DOntationSpatial}
\end{figure}

The three-dimensional model was fitted with the same pre-processing, prior hyper-parameters, MCMC settings and stopping rule as in Section~\ref{Subsec:1DOntario}, with the discretization \eqref{Eq:Discret} now based on a $20\times20\times20$ tensor grid ($V=8000$ nodes) and the covariate values binned on $50\times50\times50$ bins; the chains met the stopping criteria after 110,000 iterations, each requiring about 55 minutes (around 30 seconds per 1000 iterations).

We conclude by acknowledging that yearly wildfire and meteorological data are likely to exhibit temporal dependence, so that treating the yearly observations as independent is a modeling simplification, and the above results should be understood as pooled inference over 2004--2022 under this simplified model. That said, since the posterior depends on the data only through the conditional likelihood of the point patterns given the covariates, temporal dependence in the meteorological fields does not affect the likelihood on which the procedure is based. What the model neglects is possible residual dependence of the wildfire locations across years, given the covariates. The implications of this, and possible extensions, are discussed in Section~\ref{Subsec:ExtDataAnalysis} below.

%
%
%
%
%

\section{Summary and discussion}\label{Sec:Discussion}

We have studied the estimation of the intensity function of a covariate-driven point process from i.i.d.~observations. We devised multi-bandwidth Gaussian process methods and showed that they achieve optimal adaptive posterior contraction rates towards (possibly) anisotropic ground truths (cf.~Theorem~\ref{Theo:GPRates}). For implementation, we have constructed a Metropolis-within-Gibbs MCMC algorithm based on numerical likelihood evaluations and a dimension-robust sampling scheme (cf.~Section~\ref{Subsec:Algorithm}). The methods have been assessed through numerical simulations (cf.~Section~\ref{Sec:Simulations}) and applied to a Canadian wildfire dataset (cf.~Section~\ref{Sec:RealData}). Overall, the proposed strategy combines optimal theoretical guarantees with a feasible implementation and good practical performance.

%
%
%

\subsection{Theoretical open problems}

An important open question concerns the frequentist validity of the associated uncertainty quantification, since in nonparametric models credible sets may have asymptotically vanishing coverage even when the posterior is consistent, e.g.~\citet{DF86}. Possible approaches include the radius inflation strategy of \citet{SvdVvZ15} and (adaptive) nonparametric Bernstein–von Mises theorems \citep{R17}, for which the contraction rates derived here could provide a useful initial `localization’ step.

It would also be of interest to extend our results to other priors, such as the covariate-based Pólya tree-type priors of \citet{GKR26}, which achieve adaptive optimal point-wise rates in an increasing domain regime. Similarly, to our knowledge no theoretical guarantees are currently available for kernel-type estimators in the present setting with i.i.d.~observations, such as the averaged estimator \eqref{Eq:AverageKernel} built from the single-observation procedures of \citet{G08,BCST12,BGMMM20}; for these, our theoretical and empirical results may serve as a benchmark.

%
%
%

\subsection{Extensions of the data analysis}
\label{Subsec:ExtDataAnalysis}

While our approach appears to satisfactorily capture the relationship between the occurrence of wildfires and the considered covariates, several refinements are possible. Firstly, yearly wildfire data are likely to exhibit temporal dependence and long-term trends. In the presence of residual dependence, the posterior mean remains interpretable as an estimate of the average relationship over the observation period, while the credible sets may understate the uncertainty. Trends could be accounted for by including time as an additional covariate, which the proposed algorithm accommodates directly, although our theoretical framework would then no longer apply. Temporal dependence could be incorporated via auto-regressive components for the point patterns and the covariates, or via spatio-temporal models such as the one recently investigated by \citet{MS26} (where, however, no covariates are considered).

Additional covariates available from the Canadian Wildland Fire Information System, as well as residual spatial effects along the lines of Remark~\ref{Rem:DetCov}, could also improve predictive power, albeit at the risk of over-parametrizing the model. Lastly, latent random effects could be included by setting $\lambda(x)=\rho(Z(x),Y(x))$, $x\in\Wcal$, where $Y$ is an unobserved random field, modeled e.g.~via a Gaussian process as in the log-Gaussian Cox process of \citet{MSW98}. This could provide robustness against latent spatial variability and dependencies, but would require substantial modifications of the present methodological and theoretical developments, which we leave for future research.



%
%
\section*{Acknowledgments}\label{acks}
	We thank the Editor, the Associate Editor and two Referees for insightful comments that led to an improvement of the manuscript. M.G. is also affiliated with Collegio Carlo Alberto, Turin.


\section*{Funding}\label{funding}	We gratefully acknowledge support from Collegio Carlo Alberto and the ``de Castro'' Statistics Initiative.
	M.G. was partially supported by MUR, PRIN project 2022CLTYP4.


\section*{Data Availability Statement}\label{data-availability-statement}	The data and \texttt{R} code are available at  \href{https://github.com/PatricDolmeta/Covariate-based-nonparametric-Bayesian-intensity-estimation}{\texttt{github.com/PatricDolmeta}}.

\bibliography{Bibliography.bib}

\clearpage 
\begin{center}
\LARGE \textbf{Supplementary Material}
\end{center}
In this supplement, we present the proofs of the theoretical results, auxiliary lemmas, additional simulations, sensitivity studies and further details on the data analysis.

\appendix

\section{Proof of Theorem \ref{Theo:GPRates}}\label{Sec:ProofGPRates}

As mentioned in Section \ref{Subsec:Theory} of the main article, the result holds under a slightly more flexible prior class, summarized in the following condition. The prior constructed in Section \ref{Subsec:Prior}, which appears in the statement of Theorem \ref{Theo:GPRates}, represents a concrete instance to which the general theory applies, cf.~Lemma \ref{Lem:TiltedGamma} below. We refer to e.g.~\citet[Chapter 11]{GvdV17} for background information on stationary Gaussian processes.

\begin{condition}[Multi-bandwidth Gaussian process priors for covariate-based intensities]
	\label{Cond:GP}  
	Let $\Zcal=[0,1]^d$, $d\in\N$, and let $\Rcal$ be the set of measurable, bounded and nonnegative-valued functions defined on $\Zcal$. Let $\Pi$ be a prior supported on $\Rcal$ given by the law of the random function $\rho(z) = \rho^*\sigma(w(z)), \ z\in\Zcal$, where:
	\begin{enumerate}
		
		\item $\sigma:\R\to(0,1)$ is smooth and bijective, with $\sigma'>0$, and $\sqrt{\sigma}$ is uniformly Lipschitz;
		
		\item $\rho^*\sim \Pi_{\rho^*}$, for an absolutely continuous prior $\Pi_{\rho^*}$ on $[0,c_{\rho^*}+\log n]$, for some fixed $c_{\rho^*}>0$, whose p.d.f.~(also denoted by $\Pi_{\rho^*}$) is bounded away from zero on compact subsets of $(0,c_{\rho^*}+\log n]$, uniformly in $n$;
		
		\item independently of $\rho^*$, $\ell=(\ell_1,\dots,\ell_d)\sim\Pi_\ell$, for some prior $\Pi_\ell$ on $(0,\infty)^d$ satisfying the following for some $a_0\ge1$, $C_\ell,c_\ell>0$ and $k_\ell\ge0$:
		\begin{itemize}
			\item for all $a=(a_1,\dots,a_d)\in[a_0,\infty)^d$, letting $a^*:=\prod_{j=1}^d a_j$,
			\begin{equation}\label{Eq:H1}
				\Pi_\ell\big(\ell : \ell_j\in[a_j,2a_j]\ \text{for all } j=1,\dots,d\big)\ge e^{-C_\ell a^*\log^{k_\ell}a^*};
			\end{equation}
			\item for all sufficiently large $R>0$,
			\begin{equation}\label{Eq:H2}
				\Pi_\ell\Big(\ell : \prod_{j=1}^d(\ell_j\vee 1)>R\Big)\le e^{-c_\ell R}.
			\end{equation}
		\end{itemize}

		\item independently of $\rho^*$, $w\sim \Pi_W$, defined as follows: conditionally on $\ell\sim \Pi_\ell$, let $
		w|\ell\sim \Pi_{W_\ell}$, given by the law of the restriction $W_\ell =\{W_\ell(z), \ z\in[0,1]^d\}$ to $[0,1]^d$ of a centered and stationary Gaussian process on $\R^d$ with kernel having spectral expansion
		\begin{equation}
			\label{Eq:GenCovKernel}
			\Enorm[W_\ell(z)W_\ell(z')] = \int_{\R^d}e^{-i\sum_{j=1}^d\xi_j  \ell_j(z_j - z_j')}d\mu(\xi),
			\ z=(z_1,\dots,z_d), \ z' = (z_1',\dots,z_d'),
		\end{equation}
		and whose spectral measure $\mu$ admits a Lebesgue density that is bounded away from zero on a neighbourhood of the origin, and satisfies, for some $c_\mu>0$
		\begin{equation}
			\label{Eq:SpMeas}
			\int_{\R^d} e^{c_\mu|\xi|}d\mu(\xi)<\infty.
		\end{equation}
	\end{enumerate}
\end{condition}

Condition \eqref{Eq:H1} requires that the hyper-prior assigns sufficient mass to every configuration of large length-scales, at a cost driven only by their product $a^*=\prod_j a_j$, irrespective of how $a^*$ is distributed across directions. Condition \eqref{Eq:H2} requires that the overall complexity $\prod_j(\ell_j\vee1)$ rarely exceeds large levels. This is used to construct sieves of suitably bounded metric entropy, cf.~Lemma \ref{Lem:GPSieves}.

Below, for sets $\Theta$, semi-metrics $\delta$ on $\Theta$ and any $\varepsilon>0$, the covering numbers $\Ncal(\varepsilon;\Theta,\delta)$ are defined as the smallest number of balls of $\delta$-radius equal to $\varepsilon$ required to cover $\Theta$.

\begin{proof}[Proof of Theorem~\ref{Theo:GPRates}]
	
	As argued below, cf.~Lemma \ref{Lem:TiltedGamma}, the prior $\Pi$ considered in Section \ref{Subsec:Prior} satisfies Condition \ref{Cond:GP}. We prove the claim for any prior satisfying the latter, verifying the assumptions of the general concentration result, Theorem \ref{Theo:GenContrRates}, with $\varepsilon_n = n^{-\alpha_0/(2\alpha_0+1)}\log^{c_1} n$ and $\bar\varepsilon_n = n^{-\alpha_0/(2\alpha_0+1)}\log^{(2c_1+2+d)/2} n$ for some sufficiently large $c_1>0$. Since $\rho_0$ is continuous and $[0,1]^d$ is compact, we have $\|\rho_0\|_\infty+1<c_{\rho^*}+\log n$ provided that $n$ is large enough, whence, for all such $n$'s, $\|\rho_0\|_\infty+1$ is included in the interior of the support of $\Pi_{\rho^*}$. Following the argument in Section 4.1 of \citet{KvZ15}, we may then write $\rho_0 = (\|\rho_0\|_\infty+1)\sigma \circ w_0$, where $w_0:[0,1]^d\to\R$ is given by $w_0:=\sigma^{-1}\circ (\rho_0/(\|\rho_0\|_\infty+1))$. Since $\rho_0$ is bounded away from zero by assumption, the function $\rho_0/(\|\rho_0\|_\infty+1)$ too is bounded away from zero. Noting that also $\rho_0/(\|\rho_0\|_\infty+1)<1$, we may conclude that $w_0\in C^\alpha(\Zcal)$ in view of the fact that $\sigma^{-1}$ is smooth over $(0,1)$ (since $\sigma'>0$).

	We start with the verification of the prior mass condition \eqref{Eq:SmallBall}. We have
	\begin{align*}
		&\Pi(\rho : \|\rho - \rho_0\|_\infty\le \varepsilon_n)\\
		&\ =\Pi\left((\rho^*,w) : \|(\rho^* - (\|\rho_0\|_\infty+1))\sigma\circ w + (\|\rho_0\|_\infty+1)(\sigma\circ w - \sigma\circ w_0)\|_\infty\le \varepsilon_n\right)\\
		&\ \ge \Pi_{\rho^*}\left(r : |r - (\|\rho_0\|_\infty+1)|\le \varepsilon_n/2\right)
		\Pi_W\left(w:\|\sigma\circ w - \sigma\circ w_0\|_\infty
		\le \varepsilon_n/(2\|\rho_0\|_\infty+2)\right).
	\end{align*}
	Since, by assumption, the p.d.f.~of $\Pi_{\rho^*}$ is bounded away from zero on compact subsets of $(0,c_{\rho^*}+\log n]$, uniformly in $n$, and since $n\varepsilon_n^2\to\infty$, the first probability is bounded below by $c_2\varepsilon_n \ge e^{-n\varepsilon_n^2}$ as $n\to\infty$ for some $c_2>0$. Further, note that since $\sqrt{\sigma}$ is Lipschitz and since $\sigma\le 1$, the function $\sigma$ too is Lipschitz (with Lipschitz constant bounded by twice that of $\sqrt \sigma$), whence the second probability is greater than $\Pi_W\left(w:\| w -  w_0\|_\infty\le c_3\varepsilon_n\right)$ for some $c_3>0$. For $\varepsilon_n$ as above, provided that $c_1$ is large enough, the latter is bounded below by $e^{-n\varepsilon_n^2}$ by Lemma \ref{Lem:GPPriorMass}. This shows that condition \eqref{Eq:SmallBall} holds (with $C_1 = 2$).

	Moving onto the sieve condition \eqref{Eq:Sieves}, for $\Bcal_n$ the set defined as in Lemma \ref{Lem:GPSieves} below, take
	$$
	\Rcal_n:=\bigcup_{r\le c_{\rho^*}+\log n} r \Scal_n,
	\qquad \Scal_n :=\{\sigma\circ w, \ w\in\Bcal_n\}.
	$$
	Then, recalling the prior construction from Condition \ref{Cond:GP},
	\begin{align*}
		\Pi( \Rcal_n^c)
		&=\int_0^{c_{\rho^*}+\log n}\Pi_W\left(w : r\sigma\circ w \notin\Rcal_n \right)
		\Pi_{\rho^*}(r)dr.
	\end{align*}
	For all $r\le c_{\rho^*}+\log n$, we have $\Pi_W(w:r\sigma\circ w\notin\Rcal_n) \le  \Pi_W(w:\sigma\circ w\notin\Scal_n)\le\Pi_W(\Bcal_n^c)$. By Lemma \ref{Lem:GPSieves}, for any $C_2>1$, we may choose the sequences $\eta_n,R_n,T_n$ in the definition of $\Bcal_n$ so that $\eta_n=o(\bar\varepsilon_n/\sqrt{c_{\rho^*}+\log n})$ and $\Pi_W(\Bcal_n^c)\le e^{-C_2n\varepsilon_n^2}$. Combined with the previous display, this shows that $\Pi( \Rcal_n^c)\le e^{-C_2n\varepsilon_n^2}$. Further, for $\bar\varepsilon_n$ as above, the set $\Bcal_n$ also satisfies
	\begin{align}
		\label{Eq:AlmostThere}
		\log \Ncal(2\eta_n;\Bcal_n,\|\cdot\|_{\infty})
		\lesssim
		n\bar\varepsilon_n^2
	\end{align}
	by Lemma \ref{Lem:GPSieves}. We proceed verifying the sup-norm metric entropy inequality \eqref{Eq:SupNormEntropy}, which, as observed in Remark \ref{Rem:SupNormEntropy}, is a sufficient condition for the complexity bound in \eqref{Eq:Sieves} to hold. Since $\sqrt\sigma$ is bounded and Lipschitz by assumption we have, for any $r_1,r_2\in[0,c_{\rho^*}+\log n]$ and any $w_1,w_2\in\Bcal_n$,
	\begin{align*}
		\|\sqrt{r_1\sigma\circ w_1} - \sqrt{r_2\sigma\circ w_2}\|_\infty
		&\le |\sqrt{r_1} - \sqrt{r_2}| + c_4\sqrt{c_{\rho^*}+\log n}\|w_1 - w_2\|_\infty\\
		&\le \sqrt{|r_1 - r_2|} + c_4\sqrt{c_{\rho^*}+\log n}\|w_1 - w_2\|_\infty
	\end{align*}
	for some $c_4>0$ only depending on $\sigma$. Therefore, in view of \eqref{Eq:AlmostThere}, since $2\eta_n\le\bar\varepsilon_n/(2c_4$ $\sqrt{c_{\rho^*}+\log n})$ for all $n$ large enough,
	\begin{align*}
		\log&\Ncal\left(\bar\varepsilon_n;\sqrt{\Rcal_n},\|\cdot\|_\infty\right)\\
		&\le \log\Ncal\left(\bar\varepsilon_n/2;[0,c_{\rho^*}+\log n],\sqrt{|\cdot|}\right)
		+\log\Ncal(\bar\varepsilon_n/(2c_4\sqrt{c_{\rho^*}+\log n});\Bcal_n,\|\cdot\|_\infty)\\
		&\lesssim \log((c_{\rho^*}+\log n)/\bar\varepsilon_n) + n\bar\varepsilon_n^2
		\lesssim n\bar\varepsilon_n^2.
	\end{align*}
	
	The claim of Theorem \ref{Theo:GPRates} now follows from an application of Theorem \ref{Theo:GenContrRates}, with $\bar\varepsilon_n$ replaced by $(1\vee\sqrt{\vol(\Wcal)})\bar\varepsilon_n$ and with the choice $M=c_{\rho^*}+\log n$, noting that $\rho\land M=\rho$ for all $\rho$ in the support of $\Pi$ and $\rho_0\land M=\rho_0$ for all sufficiently large $n$, upon setting $C=c_1+2+d/2$.
\end{proof}

We conclude by noting that the prior $\Pi$ from Section \ref{Subsec:Prior}, employed in Theorem \ref{Theo:GPRates}, satisfies Condition \ref{Cond:GP}. Indeed, the sigmoid link $\sigma(t)=(1+e^{-t})^{-1}$ is a smooth bijection from $\R$ onto $(0,1)$ with $\sigma'=\sigma(1-\sigma)>0$, and $(\sqrt\sigma)'=\sqrt\sigma(1-\sigma)/2\in(0,1/2]$, whence $\sqrt\sigma$ is uniformly Lipschitz. The p.d.f.~of the truncated gamma prior $\Pi_{\rho^*}$ satisfies $\Pi_{\rho^*}(r)\ge b_{\rho^*}^{a_{\rho^*}}r^{a_{\rho^*}-1}e^{-b_{\rho^*}r}/\Gamma(a_{\rho^*})$ for all $r\in[0,c_{\rho^*}+\log n]$ and all $n$, and is thus bounded away from zero on compact subsets of $(0,c_{\rho^*}+\log n]$, uniformly in $n$. Moreover, the ARD kernel \eqref{Eq:ARD} is of the form \eqref{Eq:GenCovKernel} with $\mu$ equal to the $N(0,2I_d)$ distribution, which satisfies \eqref{Eq:SpMeas} for any $c_\mu>0$. Finally, the tilted gamma hyper-prior $\Pi_\ell$ from \eqref{Eq:HyperPrior} satisfies the requirements \eqref{Eq:H1} and \eqref{Eq:H2}, as we show in the next lemma.

\begin{lemma}\label{Lem:TiltedGamma}
	For any $a_\ell,b_\ell,\lambda>0$, the hyper-prior $\Pi_\ell$ defined in \eqref{Eq:HyperPrior} satisfies \eqref{Eq:H1} with $a_0=1$ and $k_\ell=0$, and \eqref{Eq:H2} with $c_\ell=\lambda/2$.
\end{lemma}

\begin{proof}
	Write $\Pi_\ell(l)\propto \prod_{j=1}^d l_j^{a_\ell-1}e^{-b_\ell l_j}e^{-\lambda\prod_{j=1}^d(1+l_j)}$, $l\in(0,\infty)^d$, with normalising constant lying between the fixed quantities
	$$
	c_1:=\big(e^{-b_\ell}/a_\ell\big)^de^{-2^d\lambda},
	\qquad c_2:= \big(\Gamma(a_\ell)/b_\ell^{a_\ell}\big)^d,
	$$
	the lower bound following by restricting the integral to $[0,1]^d$. To prove the validity of \eqref{Eq:H1}, let $a\in[1,\infty)^d$ and $a^*=\prod_{j=1}^da_j\ge1$, so that $a_j\le a^*$ for all $j=1,\dots,d$. For $l=(l_1,\dots,l_d)$ such that $l_j\in[a_j,2a_j]$ for all $j$, we have $\prod_{j=1}^d(1+l_j)\le\prod_{j=1}^d 3a_j=3^da^*$, $\sum_{j=1}^d l_j\le 2da^*$ and $\prod_{j=1}^dl_j^{a_\ell-1}\ge (2a^*)^{-d|a_\ell-1|}$. It follows that, since $\vol\left(\bigtimes_{j=1}^d[a_j,2a_j]\right) = a^*$,
	$$
	\Pi_\ell\big(\ell : \ell_j\in[a_j,2a_j]\ \text{for all } j=1,\dots,d\big)
	\ge c_2^{-1}(2a^*)^{-d|a_\ell-1|}e^{-(2db_\ell+3^d\lambda)a^*}
	\ge e^{-c_3a^*}
	$$
	for some $c_3>0$ only depending on $d,a_\ell,b_\ell,\lambda$. This yields \eqref{Eq:H1} with $a_0=1$ and $k_\ell=0$. For \eqref{Eq:H2}, note that $\prod_{j=1}^d(1+l_j)\ge\prod_{j=1}^d(l_j\vee1)$. Therefore, for all $R>0$,
	$$
	\Pi_\ell\Big(\ell : \prod_{j=1}^d(\ell_j\vee 1)>R\Big)
	\le c_1^{-1}e^{-\lambda R}\int_{(0,\infty)^d}\prod_{j=1}^d l_j^{a_\ell-1}e^{-b_\ell l_j}dl
	= \frac{c_2}{c_1}e^{-\lambda R}
	\le e^{-\lambda R/2}
	$$
	for all $R$ large enough, which recovers \eqref{Eq:H2} with $c_\ell=\lambda/2$.
\end{proof}

%
%
%

\subsection{Auxiliary results for hierarchical multi-bandwidth Gaussian processes}
\label{Subsec:GPAux}

The following lemma provides a lower bound, required in the proof of Theorem \ref{Theo:GPRates}, for the probability of small sup-norm neighborhoods charged by the randomly re-scaled Gaussian prior $\Pi_W$ defined in Condition \ref{Cond:GP}. The result extends the third claim of Theorem 3.1 in \citet{BPD14} to the present construction with a hyper-prior on the length-scales satisfying the requirement \eqref{Eq:H1} from Condition \ref{Cond:GP}.

\begin{lemma}\label{Lem:GPPriorMass}
	For $\alpha = (\alpha_1,\dots,\alpha_d)\in (0,\infty)^d$,
	let $w_0\in C^\alpha(\Zcal)$, and let $\Pi_W$ be a prior for $w$ constructed as in Condition \ref{Cond:GP}, with $\Pi_\ell$ any distribution on $(0,\infty)^d$ satisfying \eqref{Eq:H1}. Then, there exists a constant $K_1>0$ only depending on $w_0, d$, the spectral measure $\mu$ from \eqref{Eq:SpMeas} and $\Pi_\ell$ such that, for all sufficiently small $\varepsilon>0$,
	\begin{align*}
		\Pi_W\left(w:\|w - w_0\|_\infty\le \varepsilon \right)
		&\ge e^{-(1/\varepsilon)^{1/\alpha_0}\log^{K_1}(1/\varepsilon)},
		\qquad \alpha_0 = 1/\sum_{j=1}^d\alpha_j^{-1}.
	\end{align*}
	In particular, setting $\varepsilon_n = n^{-\alpha_0/(2\alpha_0+1)}\log^{K_2}n$ for any $K_2> K_1\alpha_0/(2\alpha_0+1)$, it holds for all sufficiently large $n$ that 
	\begin{equation}\label{Eq:GPPriorMass}
		\Pi_W\left(w:\|w - w_0\|_\infty\le \varepsilon_n \right)
		\ge e^{- n\varepsilon_n^2}.
	\end{equation}
\end{lemma}

\begin{proof}
	Let $\varepsilon>0$, and let $\ell\sim\Pi_\ell$ and $w|\ell\sim \Pi_{W_\ell}$ be as in Condition \ref{Cond:GP}. Then, we have
	\begin{equation*}
		\Pi_W\left(w:\|w - w_0\|_\infty\le \varepsilon \right)
		=\int_{(0,\infty)^d} \Pi_{W_\ell}\left(w:\|w - w_0\|_\infty\le \varepsilon \right)
		d\Pi_\ell(\ell).
	\end{equation*}
	Let $\ell^*:=\prod_{j=1}^d \ell_j$, $\bar \ell :=\max_{j=1,\dots,d} \ell_j$ and $\underline\ell := \min_{j=1,\dots,d} \ell_j$. By a combination of Lemmas 4.2 and 4.3 of \citet{BPD14}, for any fixed $\ell_0>0$ there exist constants $\varepsilon_0\in(0,1/2)$ and $c_1,c_2>0$ only depending on $w_0$, $d$ and $\mu$ such that
	$$
	\Pi_{W_\ell}\left(w:\|w - w_0\|_{\infty}
	\le \varepsilon \right)
	\ge e^{-c_1\ell^*\log^{1+d}(\bar\ell /\varepsilon)},
	$$
	for all $\varepsilon < \varepsilon_0$ and all $\ell$ such that $\underline\ell>\ell_0$ and $\sum_{j=1}^d \ell_j^{-\alpha_j}\le d\varepsilon/c_2$. Thus, taking $\ell_0=a_0$ from \eqref{Eq:H1}, letting $a_j:=(c_2/\varepsilon)^{1/\alpha_j}$, $j=1,\dots,d$, and $B_\varepsilon:=\{\ell : \ell_j\in[a_j,2a_j],\ j=1,\dots,d\}$, provided that $\varepsilon < \varepsilon_0\land c_2a_0^{-\overline\alpha}$,
	\begin{align*}
		\Pi_W&\left(w:\|w - w_0\|_\infty\le \varepsilon \right)\\
		&\ge \int_{B_\varepsilon}
		e^{-c_1 \ell^*\log^{1+d}(\bar \ell/\varepsilon)}
		d\Pi_\ell(\ell)\\
		&\ge e^{-c_1 2^d c_2^{1/\alpha_0}(1/\varepsilon)^{1/\alpha_0}\log^{1+d}(2c_2^{1/\underline\alpha}
			\varepsilon^{-1-1/\underline\alpha})}
		\Pi_\ell(B_\varepsilon)
		\ge e^{-(1/\varepsilon)^{1/\alpha_0}\log ^{c_3}(1/\varepsilon)}
		\Pi_\ell(B_\varepsilon)
	\end{align*}
	for any constant $c_3>1+d$, where $\overline\alpha:=\max_{j=1,\dots,d}\alpha_j$ and $\underline\alpha=\min_{j=1,\dots,d}\alpha_j$. Finally, since $a_j\ge a_0$ for all $j$ and $a^*=\prod_{j=1}^da_j = c_2^{1/\alpha_0}(1/\varepsilon)^{1/\alpha_0}$, \eqref{Eq:H1} gives
	$$
	\Pi_\ell(B_\varepsilon)\ge e^{-C_\ell a^*\log^{k_\ell}a^*}
	\ge e^{-(1/\varepsilon)^{1/\alpha_0}\log^{c_4}(1/\varepsilon)}
	$$
	for any constant $c_4>k_\ell$ and all $\varepsilon$ small enough. Combining the obtained estimates yields the first claim of Lemma \ref{Lem:GPPriorMass}. The second then readily follows for the given choice of $\varepsilon_n$.
\end{proof}

We next construct sieves with bounded complexity containing the bulk of the mass of the randomly re-scaled Gaussian prior $\Pi_W$ defined in Condition \ref{Cond:GP}, employed in the proof of Theorem \ref{Theo:GPRates}. Our construction is similar to the one on p.~373 of \citet{BPD14}, which is itself based on ideas from \citet{vdVvZ09}, and relies on the requirement \eqref{Eq:H2} from Condition \ref{Cond:GP}. In view of the small ball estimate \eqref{Eq:GPPriorMass} and Lemma \ref{Lem:GPSieves} below, we expect that priors based on $\Pi_W$ achieve adaptive anisotropic posterior contraction rates in other statistical models as well, along the lines discussed for example in \citet[Section 3]{vdVvZ08}.

Let $\Ccal_1$ denote the unit ball in sup-norm of $C(\Zcal)$. For each $\ell\in(0,\infty)^d$, let $\Hcal_\ell$ be the reproducing kernel Hilbert space associated to the Gaussian process $W_\ell$ from Condition \ref{Cond:GP}, and let $\Hcal_{\ell,1}$ denote its unit ball; see \citet[Section 4.1]{BPD14} for definitions and properties. For $\eta, R, T>0$, construct the sets
\begin{equation}
	\label{Eq:GPSieves}
	\Bcal:= \eta \Ccal_1 + \bigcup_{\ell\in Q_R} T\Hcal_{\ell,1},
	\qquad
	Q_R:=\Big\{\ell\in(0,\infty)^d:\prod_{j=1}^d(\ell_j\vee1)\le R\Big\}.
\end{equation}

\begin{lemma}\label{Lem:GPSieves}
	
	Let $\Pi_W$ be a prior for $w$ constructed as in Condition \ref{Cond:GP}, with $\Pi_\ell$ any distribution on $(0,\infty)^d$ satisfying \eqref{Eq:H2}. Then, there exists a constant $K>0$ only depending on $d$ and the spectral measure $\mu$ from \eqref{Eq:SpMeas} such that, for all sufficiently small $\eta$, all $R$ large enough, and all $T\ge K\sqrt{R} \log^{(1+d)/2}(R /\eta)$,
	\begin{align*}
		\Pi_W(\Bcal^c)
		&
		\le \frac{1}{2}e^{-R \log^{1+d}(R /\eta)}
		+e^{-c_\ell R},
		\qquad
		\log\Ncal(2\eta;\Bcal,\|\cdot\|_\infty)\lesssim R\log^{1+d}(2T/\eta),
	\end{align*}
	where $c_\ell>0$ is the constant from \eqref{Eq:H2}. In particular, if $\varepsilon_n=n^{-\alpha_0/(2\alpha_0+1)}\log^{K_1}n$ for some $\alpha_0, K_1>0$, then for any $K_2>0$, letting $\Bcal_n$ be as in \eqref{Eq:GPSieves} with $\eta = \eta_n = n^{-K_3}\log^{K_4}n$ for any $K_3\ge\alpha_0/(2\alpha_0+1)$ and any $K_4>0$, $R=R_n = K_5n^{1/(2\alpha_0+1)}$ $\log^{2K_1}n$ for any $K_5>K_2/(c_\ell\land1)$, and $T=T_n = K_6n^{1/(4\alpha_0+2)}\log^{(2K_1+1+d)/2}n$ for any sufficiently large $K_6>0$, we have that
	\begin{align*}
		\Pi_W(\Bcal_n^c)
		\le e^{-K_2n\varepsilon_n^2},
		\qquad \log\Ncal(2\eta_n;\Bcal_n,\|\cdot\|_\infty) 
		\lesssim n\bar\varepsilon_n^2,
	\end{align*}
	for all $n$ large enough, where $\bar\varepsilon_n = n^{-\alpha_0/(2\alpha_0+1)}\log^{(2K_1+1+d)/2}n$.
\end{lemma}

\begin{proof}
	We start with the verification of the first inequality in the first claim. With $\Pi_{W_\ell}$ and $\Pi_\ell$ as in Condition \ref{Cond:GP}, we have
	\begin{align*}
		\Pi_W(\Bcal^c)
		=\int_{(0,\infty)^d}
		\Pi_{W_\ell}(\Bcal^c)
		d\Pi_\ell(\ell)
		\le \int_{Q_R}
		\Pi_{W_\ell}(\Bcal^c)
		d\Pi_\ell(\ell)
		+\Pi_\ell(Q_R^c).
	\end{align*}
	Since $\Bcal$ contains the set $\eta\Ccal_1+T\Hcal_{\ell,1}$ for any $\ell\in Q_R$ by construction, by Borell's isoperimetric inequality \citet[Theorem 2.6.12]{GN16},
	\begin{align*}
		\Pi_{W_\ell}(\Bcal^c)
		&\le \Pi_{W_\ell}((\eta\Ccal_1+T\Hcal_{\ell,1})^c)\\
		&\le 1 - \Phi\left(T + \Phi^{-1}\left(\Pi_{W_\ell}(\eta\Ccal_1)\right)\right)
		\le 1 - \Phi\left(T + \Phi^{-1}\left(\Pi_{W_{\ell\vee1}}(\eta\Ccal_1)\right)\right),
	\end{align*}
	where $\Phi$ is the standard normal cumulative distribution function, whose inverse we denote by $\Phi^{-1}$, and where $\ell\vee1:=(\ell_1\vee1,\dots,\ell_d\vee1)$. The last inequality follows from the fact that $\Phi$ and $\Phi^{-1}$ are monotone increasing and that, for all $\ell\in(0,\infty)^d$,
	\begin{align*}
		\Pi_{W_{\ell\vee1}}(\eta\Ccal_1)
		&=\Pi_{W_1}\Big( w:\sup_{z\in [0,1]^d}|w((\ell\vee1)\odot z)|\le\eta\Big)\\
		&\le \Pi_{W_1}\Big( w:\sup_{z\in [0,1]^d}|w(\ell\odot z)|\le\eta\Big)
		=\Pi_{W_\ell}(\eta\Ccal_1),
	\end{align*}
	in view of \eqref{Eq:GenCovKernel} and of the inclusion $\{\ell\odot z, \ z\in[0,1]^d\}\subseteq\{(\ell\vee1)\odot z, \ z\in[0,1]^d\}$. For $\ell\in Q_R$, we have $\ell\vee1\in[1,\infty)^d$, $\prod_{j=1}^d(\ell_j\vee1)\le R$ and $\max_{j}(\ell_j\vee1)\le R$. Thus, provided that $\eta$ is small enough, Lemma 4.3 in \citet{BPD14} gives that, for some constant $c_1\ge1$ only depending on $d$ and $\mu$,
	\begin{align*}
		\Pi_{W_{\ell\vee1}}(\eta\Ccal_1)
		\ge e^{-c_1 R \log^{1+d}(R /\eta)},
		\qquad \ell\in Q_R.
	\end{align*}
	By the standard inequality $\Phi^{-1}(t)\ge - \sqrt{2\log(1/t)}$ holding for all $0<t<1$, cf.~\citet[Lemma 4.10]{vdVvZ09}, we then obtain 
	\begin{align*}
		\Pi_{W_\ell}(\Bcal^c)
		&\le 1 - \Phi\left(T - \sqrt{2c_1R} \log^{(1+d)/2}(R /\eta)\right),
	\end{align*}
	whence, taking $T\ge 2\sqrt{2c_1} \sqrt{R} \log^{(1+d)/2}(R /\eta)$, the standard Gaussian tail bound yields, for all $\ell\in Q_R$,
	\begin{align*}
		\Pi_{W_\ell}(\Bcal^c)
		&\le 1 - \Phi\left(\sqrt{2c_1R} \log^{(1+d)/2}(R /\eta)\right)
		\le \frac{1}{2}e^{-c_1R \log^{1+d}(R /\eta)}
		\le \frac{1}{2}e^{-R \log^{1+d}(R /\eta)}.
	\end{align*}
	Finally, by \eqref{Eq:H2}, $\Pi_\ell(Q_R^c)\le e^{-c_\ell R}$ for all $R$ large enough. Combining the obtained estimates yields the first inequality in the first claim of Lemma \ref{Lem:GPSieves}, with $K=2\sqrt{2c_1}$.

	Moving onto the second inequality, note that, by construction, $\Bcal$ is an $\eta$-enlargement in sup-norm of the set $\bigcup_{\ell\in Q_R} T\Hcal_{\ell,1}$, and therefore,
	$$
	\log \Ncal(2\eta;\Bcal,\|\cdot\|_\infty)
	\le \log \Ncal\Bigg(\eta;\bigcup_{\ell\in Q_R} T\Hcal_{\ell,1},\|\cdot\|_{\infty}\Bigg).
	$$
	Let $\Delta_{d-1}$ be the $d-1$ dimensional simplex. Then, provided that $R>1$,
	$$
	Q_R\subseteq\bigcup_{\theta\in\Delta_{d-1}}\left\{(\ell_1,\dots,\ell_d):0< \ell_j\le R^{\theta_j}, \ j=1,\dots,d\right\}.
	$$
	Indeed, for $\ell\in Q_R$, the numbers $t_j:=\log(\ell_j\vee1)/\log R\ge0$ satisfy $\sum_{j=1}^dt_j\le 1$ and $\ell_j\le R^{t_j}$, and it suffices to take any $\theta\in\Delta_{d-1}$ with $\theta_j\ge t_j$ for all $j$. Thus, using Lemma 4.5 in \citet{BPD14},
	$$
	\log \Ncal(2\eta;\Bcal,\|\cdot\|_\infty)
	\le \log \Ncal\Bigg(\eta;\bigcup_{\theta\in \Delta_{d-1}}
	\bigcup_{\ell\le R^{\theta}} T\Hcal_{\ell,1},\|\cdot\|_{\infty}\Bigg)
	\le
	c_2 R \log^{1+d}(2T/\eta),
	$$
	for some $c_2>0$ only depending on $d$ and the spectral measure $\mu$ in \eqref{Eq:SpMeas}. This concludes the verification of the first claim of Lemma \ref{Lem:GPSieves}.

	For the second claim, with the given definition of $\Bcal_n$, we have $T_n\ge K\sqrt{R_n}\log^{(1+d)/2}$ $(R_n/\eta_n)$ for all sufficiently large $n$, provided that $K_6$ is large enough, and therefore
	$$
	\Pi_W(\Bcal_n^c)
	\le \frac{1}{2} e^{-K_5 n^{1/(2\alpha_0+1)} \log^{2K_1}n}
	+e^{-c_\ell K_5 n^{1/(2\alpha_0+1)} \log^{2K_1}n}
	\le e^{-K_2n\varepsilon_n^2},
	$$
	as well as
	$$
	\log \Ncal(2\eta_n;\Bcal_n,\|\cdot\|_\infty)\lesssim n^{1/(2\alpha_0+1)} \log^{2K_1}n
	\log^{1+d}n
	=n\bar\varepsilon_n^2.
	$$
\end{proof}

%
%
%

\section{A general posterior contraction rate theorem}
\label{Sec:GenTheory}

We develop a general concentration theorem holding under abstract prior conditions resembling the standard assumptions from the asymptotic theory of Bayesian nonparametrics (e.g.~\citet{GvdV17}). This constitutes the primary tool to prove our main result on multi-bandwidth Gaussian process methods for covariate-based intensities, Theorem \ref{Theo:GPRates}. The general result is based on the Hellinger testing approach to posterior contraction rates in i.i.d.~statistical models of \citet{GGvdV00}, which we pursue in the present setting by extending ideas developed by \citet{BSvZ15} and \citet{KvZ15} for non-covariate-dependent inhomogeneous Poisson processes. Recall the metric $d_Z$ defined in \eqref{Eq:dZ}, and the notation $\Rcal$ for the set of measurable, bounded and nonnegative-valued functions defined on the covariate space $\Zcal$. Throughout this section, the latter is allowed to be any subset of $\R^d$.

\begin{theorem}\label{Theo:GenContrRates}
	Let $\rho_0\in\Rcal$ satisfy $\inf_{z\in\Zcal}\rho_0(z)>0$,  and consider data $D^{(n)}\sim P_{\rho_0}^{(n)}$ arising as described at the beginning of Section \ref{Sec:GPMethods}. Let $\Pi$ be a prior for $\rho$ supported on $\Rcal$, and assume that for a sequence $\varepsilon_n\to0$ such that $n\varepsilon_n^2\to\infty$ as $n\to\infty$ and some constant $C_1>0$ we have
	\begin{equation}\label{Eq:SmallBall}
		\Pi(\rho :
		\|\rho - \rho_0\|_\infty\le\varepsilon_n)\ge e^{-C_1 n\varepsilon_n^2}.
	\end{equation}
	Further, assume that for a sequence $\bar\varepsilon_n\to0$ as $n\to\infty$ such that $\bar\varepsilon_n\ge\varepsilon_n$, and for all $C_2>1$, there exist measurable sets $\Rcal_n\subseteq\Rcal$ and some constant $C_3>0$ such that,
	\begin{align}\label{Eq:Sieves}
		\Pi(\Rcal_n^c)\le e^{-C_2 n\varepsilon_n^2};
		\qquad \log \Ncal\left(\bar\varepsilon_n; \Rcal_n, d_Z\right)\le C_3 n\bar\varepsilon_n^2.
	\end{align}
	Then, for all sufficiently large $L>0$,
	$$
	\sup_{M>1}\Pi\left(\rho  : d_Z (\rho \land M ,\rho_0 \land M ) > L\sqrt{M}\bar\varepsilon_n
	\Big|D^{(n)}\right) \to 0
	$$
	in $P^{(\infty)}_{\rho_0}$-probability as $n\to\infty$.
\end{theorem}

The `prior mass condition' \eqref{Eq:SmallBall} entails the customary requirement that $\Pi$ put sufficient probability mass on neighborhoods of $\rho_0$ with small radius in sup-norm. In the present setting, the latter can be shown to control the Kullback-Leibler divergence and variation (cf.~Lemma \ref{Lem:KLBounds}).

The `sieve condition' \eqref{Eq:Sieves} requires that the bulk of the prior mass be concentrated on sets of suitably bounded metric entropy. The complexity bound in the second inequality is with respect to the metric $d_Z$ from \eqref{Eq:dZ}, which is natural in view of its close relationship to the Hellinger distance (cf.~Lemma \ref{Lem:HellDistBounds}). As argued in Remark \ref{Rem:SupNormEntropy} below, $d_Z$ is upper bounded by the sup-norm of the difference between square-rooted intensities. This furnishes a standard approach to verify assumption \eqref{Eq:Sieves} for a potentially large variety of nonparametric priors via analytic results on their information geometry, including the ones for hierarchical multi-bandwidth Gaussian processes derived in Section \ref{Subsec:GPAux}. On the other hand, for certain priors, sup-norm complexity bounds are known to be possibly too restrictive. This is the case, for example, for Besov-Laplace priors, which are popular in inverse problems and imaging, where they furnish a `spatially inhomogeneous' alternative to Gaussian priors; see \citet{ADH21}, as well as the discussion after Theorem 1 in \citet{G23}. Extensions of the results presented in this article to such priors may still be pursued in the stationary setting considered in Section \ref{Subsec:StatCov}, under which $d_Z$ is equivalent to an $L^2$-distance between square-rooted intensities, giving rise to weaker complexity bounds compared to the sup-norm.

\begin{remark}[Cut-offs]
	\label{Rem:Cutoff}
	The claim of Theorem \ref{Theo:GenContrRates} involves the cut-off of $\rho$ and $\rho_0$ at arbitrary levels $M>1$. This stems from the lower bound for the Hellinger distance from Lemma \ref{Lem:HellDistBounds}, and is likely an artifact of the proof arising from the presence of random covariates; see Section 7.3.2 of \citet{GvdV17} for a similar situation in nonparametric regression. We note that $M$ can be taken arbitrarily large, only impacting, through the factor $\sqrt M$, the constant pre-multiplying the rate. In particular, for $M>\|\rho_0\|_\infty$, we may replace $\rho_0\land M$ with $\rho_0$ in the claim of Theorem \ref{Theo:GenContrRates}. Further, the cut-off may be sidestepped altogether if $\Pi$ is supported on a subset of functions with values on an interval of the form $[0,\rho^*_n]$ for some slowly increasing $\rho^*_n\to\infty$: since the claim of Theorem \ref{Theo:GenContrRates} holds uniformly in $M>1$, taking $M=\rho^*_n$ then yields contraction in $d_Z(\rho,\rho_0)$ at rate $\sqrt{\rho^*_n}\,\bar\varepsilon_n$. This approach was taken for the construction of the prior in Condition \ref{Cond:GP} with the choice $\rho^*_n = c_{\rho^*}+\log n$, only slightly impacting the logarithmic factor appearing in the claim of Theorem \ref{Theo:GPRates}. On the other hand, the $n$-dependence of the truncation is not essential: if, in the prior from Section \ref{Subsec:Prior}, $\Pi_{\rho^*}$ is replaced by an untruncated $\Gamma(a_{\rho^*},b_{\rho^*})$ distribution (or, more generally, by any prior with exponential tails as in condition (2.5) of \citet{KvZ15}), the assumptions of Theorem \ref{Theo:GenContrRates} can be verified as in the proof of Theorem \ref{Theo:GPRates}, with sieves as in Section 4.2 of \citet{KvZ15}, whose maximal intensity grows polynomially in $n$. This yields contraction at the same rate, up to logarithmic factors, in $d_Z(\rho\land M,\rho_0\land M)$, for any fixed $M>1$.
\end{remark}

\begin{remark}[Sup-norm complexity bounds]\label{Rem:SupNormEntropy}
	Since $\Wcal$ is compact, we have that
	\begin{align*}
		d_Z^2(\rho_1,\rho_2)
		&\le \Enorm\left[\int_{\Wcal}\sup_{x\in\Wcal}
		\left|\sqrt{\rho_1(Z(x))} 
		-\sqrt{\rho_2(Z(x))}
		\right|^2dx\right]
		\le\vol(\Wcal)
		\|\sqrt{\rho_1} - \sqrt{\rho_2}\|_\infty^2.
	\end{align*}
	A sufficient condition for \eqref{Eq:Sieves} to hold is then that
	\begin{equation}
		\label{Eq:SupNormEntropy}
		\log \Ncal\left(\bar\varepsilon_n; \sqrt{\Rcal_n}, \|\cdot\|_\infty\right)
		\le C_3 n\bar\varepsilon_n^2,
	\end{equation}
	for some $C_3>0$, where $\sqrt{\Rcal_n}:=\{\sqrt{\rho}, \ \rho\in\Rcal_n\}$. This is similar to the metric entropy condition employed by \citet{BSvZ15} in the context of non-covariate-dependent inhomogeneous Poisson processes, except with the sup-norm replacing the $L^2$-distance. 
\end{remark}

\begin{proof}[Proof of Theorem \ref{Theo:GenContrRates}]
	We verify the conditions for concentration in Hellinger distance in i.i.d.~statistical models from Theorem 8.9 in \citet{GvdV17}. Define the neighborhoods
	$$
	B_n := \left\{ \rho : KL(p_{\rho_0},p_\rho )\le c_1\varepsilon_n^2,
	\ V(p_{\rho_0},p_{\rho} )\le c_1 \varepsilon_n^2\right\},
	\qquad c_1>0,
	$$
	where $KL$ and $V$ denote the Kullback-Leibler divergence and variation, respectively, defined as in \eqref{Eq:KLDivVar} below. Since $\varepsilon_n\to0$ as $n\to\infty$, Lemma \ref{Lem:KLBounds} implies that for all sufficiently large $n$, provided that $c_1$ is large enough, $\{\rho : \|\rho - \rho_0\|_\infty\le \varepsilon_n\} \subseteq B_n$. In view of assumption \eqref{Eq:SmallBall}, we then have $\Pi(B_n)\ge e^{- C_1n\varepsilon_n^2}$, yielding condition (8.4) of Theorem 8.9 in \citet{GvdV17}, with $\sqrt{c_1}\varepsilon_n$ in place of $\varepsilon_n$ and with constant $C_1/c_1$, since $B_n$ coincides with the neighbourhood $B_2(p_{\rho_0},\sqrt{c_1}\varepsilon_n)$ therein.

	Next, fix $C_2>C_1+4c_1$, let $\Rcal_n$ be the corresponding set from assumption \eqref{Eq:Sieves}, and define the collection of observational densities $\Pcal_n := \{p_\rho, \ \rho \in\Rcal_n\}$,
	with $p_\rho$ as in \eqref{Eq:Likelihood}. Then $\Pi(\rho : p_\rho \notin \Pcal_n)\le \Pi(\Rcal_n^c)\le e^{-(C_1+4c_1) n\varepsilon_n^2}$. Further, for $h$ the Hellinger distance defined in \eqref{Eq:HellDist}, the upper bound from Lemma \ref{Lem:HellDistBounds} implies that, letting $K:=\sqrt{8\vee c_1\vee C_3}$,
	$$
	\log \Ncal( K\bar\varepsilon_n/2; \Pcal_n, h) 
	\le\log \Ncal( \sqrt 2\bar\varepsilon_n; \Pcal_n, h) 
	\le \log \Ncal( \bar\varepsilon_n; \Rcal_n, d_Z)
	\le nK^2\bar\varepsilon_n^2.
	$$
	Since $K\bar\varepsilon_n\ge\sqrt{c_1}\varepsilon_n$, this shows that conditions (8.5) and (8.6) of Theorem 8.9 in \citet{GvdV17} are also verified, with $K\bar\varepsilon_n$ in place of $\bar\varepsilon_n$. Conclude that $\Pi(\cdot|D^{(n)})$ contracts towards $\rho_0$ in Hellinger distance at rate $\bar\varepsilon_n$, namely that for all sufficiently large $c_2>0$,
	$$
	\Pi\left(\rho : h(p_\rho ,p_{\rho_0}) > c_2\bar\varepsilon_n
	\Big|D^{(n)}\right) \to 0
	$$
	in $P^{(\infty)}_{\rho_0}$-probability as $n\to\infty$. Finally, since $1-e^{-M\vol(\Wcal)/2}\ge1-e^{-\vol(\Wcal)/2}$ for all $M>1$, the lower bound from Lemma \ref{Lem:HellDistBounds} implies that $d_Z(\rho\land M,\rho_0\land M)\le c_3\sqrt M h(p_\rho,p_{\rho_0})$ for all $\rho\in\Rcal$ and all $M>1$, with $c_3:=(\vol(\Wcal)/(2-2e^{-\vol(\Wcal)/2}))^{1/2}$. Hence, for all $L\ge c_2c_3$,
	$$
	\sup_{M>1}\Pi\left(\rho  : d_Z (\rho \land M ,\rho_0 \land M ) > L\sqrt{M}\bar\varepsilon_n
	\Big|D^{(n)}\right) 
	\le \Pi\left(\rho : h(p_\rho ,p_{\rho_0}) > c_2\bar\varepsilon_n
	\Big|D^{(n)}\right),
	$$
	and the claim of Theorem \ref{Theo:GenContrRates} follows.
\end{proof}

%
%
%

\subsection{Bounds for the Kullback-Leibler divergence and variation}

For pairs $(N,Z)$ arising as described at the beginning of Section \ref{Sec:GPMethods}, recall the expression of the observational densities $p_\rho$, $\rho\in\Rcal$, given by \eqref{Eq:Likelihood}, with dominating measure $P_1$ corresponding to the standard Poisson case. The associated Kullback-Leibler divergence and variation are defined, respectively, as
\begin{equation}
	\label{Eq:KLDivVar}
	KL(p_{\rho_0},p_{\rho} ) 
	:= E_{\rho_0}\left[ \log \frac{p_{\rho_0}(N,Z)}
	{p_{\rho}(N,Z) } \right];
	\qquad
	V(p_{\rho_0},p_{\rho} ) 
	:= Var_{\rho_0}\left[ \log \frac{p_{\rho_0}(N,Z)}
	{p_{\rho}(N,Z) } \right].
\end{equation}
The following lemma provides upper bounds for these two quantities in terms of the sup-norm distance $\|\rho - \rho_0\|_\infty$. It is based on ideas from the proof of Lemma 1 and Theorem 1 of \citet{BSvZ15}, with suitable adaptations to accommodate the presence of random covariates.

\begin{lemma}\label{Lem:KLBounds}
	Let $\rho_0\in\Rcal$ satisfy $\inf_{z\in\Zcal}\rho_0(z)>0$. Then, there exist constants $C_1,C_2>0$ only depending on $\rho_0$ such that, for all sufficiently small $\varepsilon\in(0,1)$ and all $\rho\in\Rcal$ satisfying $\|\rho - \rho_0\|_\infty\le \varepsilon$, we have
	$$
	KL(p_{\rho_0},p_{\rho} ) \le C_1\varepsilon^2;
	\qquad 
	V(p_{\rho_0},p_{\rho} ) \le C_2\varepsilon^2.
	$$
\end{lemma}

\begin{proof}
	By the tower property of conditional expectations,
	\begin{equation}
		\label{Eq:FirstKL}
		KL(p_{\rho_0},p_{\rho} ) 
		= \Enorm \left[ E_{\rho_0}\left[ \log \frac{p_{\rho_0}(N,Z)}{p_{\rho} (N,Z)}\Bigg| Z  \right] \right],
	\end{equation} 
	where the outer expectation is intended with respect to the law of $Z$. We have
	\begin{align*}
		\frac{p_{\rho_0}(N,Z)}{p_{\rho}(N,Z)}
		&=e^{\sum_{k=1}^K\log 
			\frac{\rho_0(Z(X_k))}{\rho(Z(X_k))} 
			-\int_{\Wcal} (\rho_0(Z(x)) - \rho(Z(x)))dx},
	\end{align*}
	and using the fact that, under $P_{\rho_0}$, $N|Z$ is distributed as an inhomogeneous Poisson process with intensity $\lambda_{\rho_0} = \rho_0\circ Z$, the inner expectation in \eqref{Eq:FirstKL} equals
	\begin{align*}
		-&\int_{\Wcal}(\rho_0(Z(x)) - \rho(Z(x)))dx
		+E_{\rho_0}\left[ \sum_{k=1}^K\log 
		\frac{\rho_0(Z(X_k))}{\rho(Z(X_k))} \Bigg| Z  \right]\\
		&=-\int_{\Wcal}(\rho_0(Z(x)) - \rho(Z(x)))dx
		+\int_{\Wcal} \log \frac{\rho_0(Z(x))}{\rho(Z(x))}\rho_0(Z(x))dx\\
		&=
		\int_{\Wcal}\rho_0(Z(x))
		\left( \frac{\rho(Z(x))}{\rho_0(Z(x))} 
		- 1 - \log\frac{\rho(Z(x))}{\rho_0(Z(x))}  \right)  dx
		= \int_{\Wcal}\rho_0(Z(x)) G\left(\frac{\rho(Z(x))}{\rho_0(Z(x))}\right)dx,
	\end{align*}
	having set $G(t) := t - 1 - \log t$, $t>0$, and having used the standard formula for the expectation of functionals of inhomogeneous Poisson processes. The function $G$ satisfies $|G(t)|\le 3(\sqrt t -1)^2$ for all $t\in (1/e,\infty)$ and $|G(t)|\le |\log t|$ for all $t\in(0,1/e]$. It follows that
	\begin{align*}
		E_{\rho_0}\left[ \log \frac{p_{\rho_0}(N,Z)}
		{p_{\rho}(N,Z) }\Bigg | Z \right ]
		&\le
		3\int_{\{x: \rho(Z(x))/\rho_0(Z(x)) > 1/e\}}
		\left(\sqrt{\rho(Z(x))} - \sqrt{\rho_0(Z(x))}\right)^2  dx\\
		&\quad+\int_{\{x : \rho(Z(x))/\rho_0(Z(x)) \le 1/e\}}\rho_0(Z(x))\left|\log 
		\frac{\rho(Z(x))}{\rho_0(Z(x))}\right| dx\\
		&\le 3\left\|\sqrt{\lambda_{\rho_0}} - \sqrt{\lambda_{\rho}}\right\|^2_{L^2(\Wcal)}
		+\int_{\Wcal}\rho_0(Z(x))\log ^2
		\frac{\rho(Z(x))}{\rho_0(Z(x))} dx.
	\end{align*}
	Using the fact that $1 - t\le |\log t|$ for all $t\in(0,1)$, we obtain the further upper bound
	\begin{align*}
		\left\|\sqrt{\lambda_{\rho_0}} - \sqrt{\lambda_{\rho}}\right\|^2_{L^2(\Wcal)}
		&\le \int_{\{x : \rho(Z(x))\ge\rho_0(Z(x))\}}\rho(Z(x))
		\left(\log\sqrt{\frac{\rho_0(Z(x))}
			{\rho(Z(x))}} \right)^2dx\\
		&\quad+\int_{\{x : \rho_0(Z(x))>\rho(Z(x))\}}\rho_0(Z(x))\left( \log\sqrt{\frac{\rho(Z(x))}
			{\rho_0(Z(x))}} \right)^2dx\\
		&\le \frac{1}{4}\int_{\Wcal} (\rho_0(Z(x))\vee\rho(Z(x))) \log ^2
		\frac{\rho_0(Z(x))}{\rho(Z(x))} dx.
	\end{align*}
	Combined with the second to last display, this implies
	\begin{equation}
		\label{Eq:IntermediateKL}
		E_{\rho_0}\left[ \log \frac{p_{\rho_0}(N,Z)}
		{p_{\rho}(N,Z) }\Bigg | Z \right ]
		\le 
		\frac{7}{4}\int_{\Wcal} (\rho(Z(x))\vee\rho_0(Z(x))) \log ^2
		\frac{\rho_0(Z(x))}{\rho(Z(x))} dx.
	\end{equation}
	Now for each $x\in \Wcal$, 
	by a Taylor expansion with exact remainder,
	\begin{align*}
		\log \frac{\rho_0(Z(x))}{\rho(Z(x))}
		&=  \left(\frac{\rho_0(Z(x))}{\rho(Z(x))} - 1 \right)
		-\frac{1}{2\xi_x^2}\left(\frac{\rho_0(Z(x))}{\rho(Z(x))} - 1 \right)^2,
	\end{align*}
	where $\xi_x$ lies between $\rho_0(Z(x))/\rho(Z(x))$ and $1$. Since $\rho_0$ is bounded away from zero by assumption, for all sufficiently small $\varepsilon\in(0,1)$, we have that if $\|\rho - \rho_0\|_\infty\le \varepsilon$ then necessarily $\inf_{z\in\Zcal}\rho(z)
	\ge\frac{1}{2}\inf_{z\in\Zcal}\rho_0(z)>0$. It follows that
	$$
	\left| \frac{\rho_0(Z(x))}{\rho(Z(x))} - 1\right| 
	\le \frac{1}{\inf_{z\in\Zcal}\rho(z)}\|\rho_0 - \rho\|_\infty\le c_1 \varepsilon,
	$$
	for some $c_1>0$ independent of $\rho$, $x$ and $\varepsilon$. This also implies that $\xi_x$ is itself bounded away from zero, so that
	$$
	\frac{1}{2\xi_x^2}\left(\frac{\rho_0(Z(x))}{\rho(Z(x))} - 1 \right)^2
	\le c_2\varepsilon^2 \le c_2\varepsilon,
	$$
	and $\log ^2
	(\rho_0(Z(x))/\rho(Z(x)))\le c_3\varepsilon^2$, with  $c_2,c_3>0$ independent of $\rho$, $x$ and $\varepsilon$. Finally, observing that, if $\varepsilon$ is sufficiently small, we must have $\|\rho\|_\infty\le 2\|\rho_0\|_\infty$, we obtain from \eqref{Eq:IntermediateKL} that
	\begin{align}
		\label{Eq:InnerExp}
		E_{\rho_0}\left[ \log \frac{p_{\rho_0}(N,Z)}
		{p_\rho(N,Z) }\Bigg | Z \right ]
		\le \frac{7}{2} \|\rho_0\|_\infty
		c_3\vol(\Wcal) \varepsilon^2.
	\end{align}
	Combined with \eqref{Eq:FirstKL}, this concludes the proof of the first claim of Lemma \ref{Lem:KLBounds} upon taking $C_1:= \frac{7}{2} c_3\|\rho_0\|_\infty\vol(\Wcal)$.

	The second claim is proved with a similar argument, applying the law of total variance to obtain the identity
	\begin{align*}
		V(p_{\rho_0},p_{\rho} ) 
		&= \Enorm\left[ 
		Var_{\rho_0} \left[\log \frac{p_{\rho_0}(N,Z)}
		{p_\rho(N,Z) }\Bigg | Z \right ] \right]
		+\Varnorm\left[ 
		E_{\rho_0}\left[ \log \frac{p_{\rho_0}(N,Z)}
		{p_\rho(N,Z) }\Bigg | Z \right ] \right],
	\end{align*}
	where the outer expectation and variance are intended with respect to the law of $Z$. Using again the fact that, under $P_{\rho_0}$, $N|Z$ is distributed as an inhomogeneous Poisson process with intensity $\lambda_{\rho_0}$,
	\begin{align*}
		Var_{\rho_0} \left[\log \frac{p_{\rho_0}(N,Z)}
		{p_\rho(N,Z) }\Bigg | Z \right ]
		&=
		Var_{\rho_0} \left[
		\sum_{k=1}^K\log \frac{\rho_0(Z(X_k))}{\rho(Z(X_k))}
		-\int_{\Wcal}(\rho_0(Z(x)) - \rho(Z(x)))dx 
		\Bigg | Z \right ]\\
		&=
		Var_{\rho_0} \left[
		\sum_{k=1}^K\log \frac{\rho_0(Z(X_k))}{\rho(Z(X_k))}
		\Bigg | Z \right ]
		=
		\int_{\Wcal} \log^2 \frac{\rho_0(Z(x))}{\rho(Z(x))}\rho_0(Z(x))dx.
	\end{align*}
	The bounds obtained in the first part of the proof now yield that, for all sufficiently small $\varepsilon\in(0,1)$, if $\|\rho - \rho_0\|_\infty\le \varepsilon$ we must have 
	$$
	Var_{\rho_0} \left[\log \frac{p_{\rho_0}(N,Z)}
	{p_\rho(N,Z) }\Bigg | Z \right ]
	\le c_3\|\rho_0\|_\infty\vol(\Wcal)\varepsilon^2,
	$$
	where we recall that the constant $c_3$ is independent of $\rho$, $x$ and $\varepsilon$. Further, in view of \eqref{Eq:InnerExp}, we also have
	$$
	\Varnorm\left[ 
	E_{\rho_0}\left[ \log \frac{p_{\rho_0}(N,Z)}
	{p_\rho(N,Z) }\Bigg | Z \right ] \right]
	\le c_4 \varepsilon^4 \le c_4\varepsilon^2,
	$$
	for some $c_4>0$ independent of $\rho$, $x$ and $\varepsilon$. Setting $C_2 : = c_3\|\rho_0\|_\infty\vol(\Wcal) + c_4$ yields the second claim of Lemma \ref{Lem:KLBounds}.
\end{proof}

%
%
%

\subsection{Bounds for the Hellinger distance}
\label{Subsec:HellDist}

The Hellinger distance between two observational densities
$p_{\rho_1},p_{\rho_2}$, with $\rho_1,\rho_2\in\Rcal$, defined as in \eqref{Eq:Likelihood}, is given by
\begin{equation}
	\label{Eq:HellDist}
	h\left(p_{\rho_1},p_{\rho_2}\right)
	:= \sqrt{ E_1\left[\left(\sqrt{p_{\rho_1}(N,Z)} - \sqrt{p_{\rho_2}(N,Z)} \right)^2 
		\right]},
	\qquad \rho_1,\rho_2\in\Rcal,
\end{equation}
where $E_1$ is the expectation with respect to the dominating measure $P_1$. The following lemma provides upper and lower bounds for this quantity in terms of the metric $d_Z$ from \eqref{Eq:dZ}. Its proof adapts, in the present setting, the argument to derive the first statement of Lemma 1 in \citet{BSvZ15}.

\begin{lemma}\label{Lem:HellDistBounds}
	For all $\rho_1,\rho_2\in\Rcal$ and all $M>0$ it holds that
	$$
	\frac{2(1 - e^{-\frac{1}{2}M\vol(\Wcal)})}{M\vol(\Wcal)}d_Z^2(\rho_1\land M,\rho_2\land M)
	\le
	h^2\left(p_{\rho_1},p_{\rho_2}\right)
	\le 2d_Z^2(\rho_1,\rho_2).
	$$
\end{lemma}

\begin{proof}
	We start with the well-known identity for the squared Hellinger distance,
	\begin{equation}
		\label{Eq:HellDistAff}
		h^2\left(p_{\rho_1},p_{\rho_2}\right) = 2\left[1 - a\left(p_{\rho_1},p_{\rho_2}\right) \right],
	\end{equation}
	where $a\left(p_{\rho_1},p_{\rho_2}\right):= E_1\left[ \sqrt{p_{\rho_1}(N,Z)} \sqrt{p_{\rho_2}(N,Z)}\right]$ is the Hellinger affinity. By a change of measure and the tower property of conditional expectations, the latter can be written as
	\begin{align*}
		a\left(p_{\rho_1},p_{\rho_2}\right)
		&= E_1\left[ \sqrt{ \frac{p_{\rho_1}(N,Z)}{p_{\rho_2}(N,Z)}}
		p_{\rho_2}(N,Z)\right]\\
		&=E_{\rho_2}\left[\sqrt{ \frac{p_{\rho_1}(N,Z)}{p_{\rho_2}(N,Z)}} \right]
		=\Enorm\left[
		E_{\rho_2}
		\left[\sqrt{ \frac{p_{\rho_1}(N,Z)}{p_{\rho_2}(N,Z)}} 
		\Bigg | Z \right] 
		\right],
	\end{align*}
	where the outer expectation is intended with respect to the law of $Z$. Recalling \eqref{Eq:Likelihood} and using the fact that, under $P_{\rho_2}$, $N|Z$ is distributed as an inhomogeneous Poisson point process with intensity $\lambda_{\rho_2}$, the inner expectation in the last display equals
	\begin{align*}
		&e^{-\frac{1}{2}\int_{\Wcal}(\lambda_{\rho_1}(x)-\lambda_{\rho_2}(x))dx} 
		E_{\rho_2}\left[e^{\frac{1}{2}\sum_{k=1}^K\log 
			\frac{\lambda_{\rho_1}(X_k)}{\lambda_{\rho_2}(X_k)}}\Bigg | Z \right] \\
		&\quad=
		e^{-\frac{1}{2}\int_{\Wcal}(\lambda_{\rho_1}(x)-\lambda_{\rho_2}(x))dx}
		e^{\int_{\Wcal}\left(\sqrt{\lambda_{\rho_1}(x)/\lambda_{\rho_2}(x)}-1\right)
			\lambda_{\rho_2}(x)dx}
		=e^{-\frac{1}{2}\|\sqrt{\lambda_{\rho_1}} - \sqrt{\lambda_{\rho_2}} \|^2_{L^2(\Wcal)}}.
	\end{align*}
	This implies that
	$ a\left(p_{\rho_1},p_{\rho_2}\right)=\Enorm\left[ e^{-\frac{1}{2}\|\sqrt{\lambda_{\rho_1}} - \sqrt{\lambda_{\rho_2}} \|^2_{L^2(\Wcal)}} \right]$,
	which combined with \eqref{Eq:HellDistAff} yields
	\begin{equation}\label{Eq:HellDistIdentity}
		h^2\left(p_{\rho_1},p_{\rho_2}\right)
		= 2\left(1 - 
		\Enorm
		\left[ e^{-\frac{1}{2}\|\sqrt{\lambda_{\rho_1}} - \sqrt{\lambda_{\rho_2}} \|^2_{L^2(\Wcal)}} \right]\right).
	\end{equation}
	The upper bound in the statement of Lemma \ref{Lem:HellDistBounds} then follows from Jensen's inequality and an application of the fact that $1- e^{-t}\le t$ for all $t\in\R$, whence
	\begin{align*}
		h^2\left(p_{\rho_1},p_{\rho_2}\right)
		&\le
		2\left(1 - e^{-\frac{1}{2}\Enorm\|\sqrt{\lambda_{\rho_1}} 
			- \sqrt{\lambda_{\rho_2}} \|^2_{L^2(\Wcal)}}\right)\\
		&\le 2\Enorm\left\|\sqrt{\lambda_{\rho_1}} 
		- \sqrt{\lambda_{\rho_2}} \right\|^2_{L^2(\Wcal)}
		= 2d_Z^2(\rho_1,\rho_2).
	\end{align*}

	For the lower bound, we apply ideas from the proof of Proposition 1 in \citet{B04}. We observe that for all $M>0$, 
	\begin{align*}
		\|\sqrt{\lambda_{\rho_1\land M}} - \sqrt{\lambda_{\rho_2\land M}} \|^2_{L^2(\Wcal)}
		&=\int_{\Wcal} \left( \sqrt{\rho_1(Z(x))\land M} 
		- \sqrt{\rho_2(Z(x))\land M} \right)^2dx\\
		&\le
		\int_{\Wcal}\left( \sqrt{\rho_1(Z(x))} 
		- \sqrt{\rho_2(Z(x))} \right)^2dx\\
		&=\left\|\sqrt{\lambda_{\rho_1}} 
		- \sqrt{\lambda_{\rho_2}} \right\|^2_{L^2(\Wcal)},
	\end{align*}
	whence, in view of \eqref{Eq:HellDistIdentity},
	$$
	h^2\left(p_{\rho_1},p_{\rho_2}\right)
	\ge 2\left(1 - 
	\Enorm
	\left[ e^{-\frac{1}{2}\|\sqrt{\lambda_{\rho_1\land M}} - \sqrt{\lambda_{\rho_2\land M}} \|^2_{L^2(\Wcal)}} \right]\right).
	$$
	At the same time, $\left\|\sqrt{\lambda_{\rho_1\land M}} - \sqrt{\lambda_{\rho_2\land M}} \right\|^2_{L^2(\Wcal)}\le M\vol(\Wcal)$, and by using the inequality, holding for all $0\le t_1 \le t_2$,
	$$
	e^{-t_1}\le \frac{e^{-t_2} -1}{t_2}t_1 + 1,
	$$
	cf.~\citet[p.~1043]{B04}, with the choices $t_1 = \frac{1}{2}\left\|\sqrt{\lambda_{\rho_1\land M}} - \sqrt{\lambda_{\rho_2\land M}} \right\|^2_{L^2(\Wcal)}$ and $t_2 = \frac{1}{2}M\vol(\Wcal)$, we obtain
	\begin{align*}
		h^2\left(p_{\rho_1},p_{\rho_2}\right)
		&\ge \frac{2(1 - e^{-\frac{1}{2}M\vol(\Wcal)})}{M\vol(\Wcal)}
		\Enorm\left\|\sqrt{\lambda_{\rho_1\land M}} - \sqrt{\lambda_{\rho_2\land M}} \right\|^2_{L^2(\Wcal)}  \\
		&=\frac{2(1 - e^{-\frac{1}{2}M\vol(\Wcal)})}{M\vol(\Wcal)}
		d_Z^2(\rho_1\land M,\rho_2\land M).
	\end{align*}
\end{proof}

%
%
%
%
%

\section{Proof of Proposition \ref{Prop:LowerBound}}
\label{Subsec:LowerBound}

\begin{proof}
	We use the usual reduction scheme of Theorem 2.5 in \citet{T09}, combined with a bound on the Kullback-Leibler divergence between the laws of the observations. For related lower bounds in non-covariate-based intensity estimation, see Theorems 6.5 and 6.6 in \citet{K98}. With slight abuse of notation, let $\nu_Z$ also denote the p.d.f.~of the stationary distribution of $Z$, which, by assumption, is bounded and bounded away from zero on $[0,1]^d$. For all $\rho_1,\rho_2\in\Rcal$ with $\inf_{z\in[0,1]^d}\rho_1(z)>0$ and $\inf_{z\in[0,1]^d}\rho_2(z)>0$, arguing as in the proof of Lemma \ref{Lem:KLBounds} and using the inequality $t-1-\log t\le (t-1)^2/t$, $t>0$, we have
	\begin{align*}
		KL(p_{\rho_1},p_{\rho_2})
		&=\Enorm\int_{\Wcal}\rho_1(Z(x))\,G\left(\frac{\rho_2(Z(x))}{\rho_1(Z(x))}\right)dx\\
		&\le \frac{1}{\inf_{z\in[0,1]^d}\rho_2(z)}\Enorm\int_{\Wcal}(\rho_2-\rho_1)^2(Z(x))dx
		\le \frac{\vol(\Wcal)\|\nu_Z\|_\infty}{\inf_{z\in[0,1]^d}\rho_2(z)}\|\rho_2-\rho_1\|^2_{L^2},
	\end{align*}
	where $G(t)=t-1-\log t$ and the last step follows from Fubini's theorem and the assumptions on the stationary distribution of $Z$. It follows that, letting $p^{(n)}_\rho(D^{(n)}):=\prod_{i=1}^np_\rho(N^{(i)},Z^{(i)})$ denote the likelihood of $n$ i.i.d.~observations, cf.~\eqref{Eq:Likelihood}, and defining $KL(p^{(n)}_{\rho_1},p^{(n)}_{\rho_2})$ analogously to \eqref{Eq:KLDivVar}, we have $KL(p^{(n)}_{\rho_1},p^{(n)}_{\rho_2})=n\,KL(p_{\rho_1},p_{\rho_2})\le n\vol(\Wcal)\|\nu_Z\|_\infty\|\rho_2-\rho_1\|^2_{L^2}/\inf_{z\in[0,1]^d}\rho_2(z)$.

	We next construct a finite collection of intensities in the class $\{\rho_0\in C^\alpha([0,1]^d):\rho_0\ge0,\ \|\rho_0\|_{C^\alpha}\le R\}$, which serve as `hypotheses' in the reduction scheme of \citet[Theorem 2.5]{T09}. 
	These are obtained by perturbing the constant intensity $\rho\equiv1$ with small bumps supported on the cells of a grid in $[0,1]^d$, whose sides along each direction $j=1,\dots,d$ are scaled according to the corresponding smoothness $\alpha_j$. 
	
	%
	
	Let $K_0:\R\to[0,\infty)$ be an infinitely differentiable function supported on $[0,1]$, not identically zero, and set $K(u):=\prod_{j=1}^dK_0(u_j)$, $u\in\R^d$. For $h\in(0,1)$, select direction-specific bandwidths $h_j:=h^{1/\alpha_j}$ (wider along smoother directions), $j=1,\dots,d$, and let $m_j:=\lfloor 1/h_j\rfloor$ be the number of disjoint intervals of length $h_j$ fitting along the $j^{\text{th}}$ direction. We index by $k=(k_1,\dots,k_d)\in\Kcal:=\bigtimes_{j=1}^d\{0,\dots,m_j-1\}$ the disjoint cells $\bigtimes_{j=1}^d[k_jh_j,(k_j+1)h_j)\subseteq[0,1]^d$. Their number $m:=|\Kcal|=\prod_{j=1}^dm_j$ satisfies $2^{-d}h^{-1/\alpha_0}\le m\le h^{-1/\alpha_0}$ for all $h$ small enough, since $\sum_{j=1}^d\alpha_j^{-1}=\alpha_0^{-1}$. For $\omega=(\omega_k)_{k\in\Kcal}\in\{0,1\}^m$ and $L>0$, define the hypothesis
	$$
	\rho_\omega(z):=1+Lh\sum_{k\in\Kcal}\omega_kK\left(\frac{z_1-k_1h_1}{h_1},\dots,\frac{z_d-k_dh_d}{h_d}\right),
	\qquad z\in[0,1]^d.
	$$
	Since the summands have disjoint supports, and since the $r^{\text{th}}$ derivative along the $j^{\text{th}}$ coordinate of each summand scales as $Lhh_j^{-r}$, with $h_j^{\alpha_j}=h$, a standard computation shows that $\|\rho_\omega\|_{C^\alpha([0,1]^d)}\le 1+c_1L$ for some $c_1>0$ only depending on $\alpha$ and $K_0$. Further, $\rho_\omega\ge1$. Thus, for all $\omega\in\{0,1\}^m$, $\rho_\omega$ belongs to the relevant class of hypotheses provided that $L\le1$ and $R\ge1+c_1$.

	Now, for all $\omega,\omega'\in\{0,1\}^m$,
	$$
	\|\rho_\omega-\rho_{\omega'}\|^2_{L^2}
	=L^2h^2\|K\|^2_{L^2(\R^d)}\prod_{j=1}^dh_j\ \sum_{k\in\Kcal}|\omega_k-\omega'_k|
	=L^2h^{2+1/\alpha_0}\|K\|^2_{L^2(\R^d)}\,H(\omega,\omega'),
	$$
	where $H$ is the Hamming distance. By the Varshamov-Gilbert bound \citep[Lemma 2.9]{T09}, for $m\ge8$ there exists $\Omega\subseteq\{0,1\}^m$ containing $\omega^{(0)}:=(0,\dots,0)$, with $|\Omega|\ge 2^{m/8}+1$ and $H(\omega,\omega')\ge m/8$ for all distinct $\omega,\omega'\in\Omega$. For such pairs, using $m\ge 2^{-d}h^{-1/\alpha_0}$, we obtain
	$$
	\|\rho_\omega-\rho_{\omega'}\|_{L^2}\ge c_2Lh,
	\qquad c_2:=2^{-(d+3)/2}\|K\|_{L^2(\R^d)}.
	$$
	On the other hand, using the above upper bound on the Kullback-Leibler divergence with $\rho_1 = \rho_\omega$ and $\rho_2=\rho_{\omega^{(0)}}\equiv1$ and the upper bound $m\le h^{-1/\alpha_0}$, 
	$$
	KL\big(p^{(n)}_{\rho_\omega},p^{(n)}_{\rho_{\omega^{(0)}}}\big)
	\le n\vol(\Wcal)\|\nu_Z\|_\infty L^2h^{2+1/\alpha_0}\|K\|^2_{L^2(\R^d)}\,m
	\le c_3nL^2h^2,
	$$
	where $c_3:=\vol(\Wcal)\|\nu_Z\|_\infty\|K\|^2_{L^2(\R^d)}$. Taking $h=h_n:=n^{-\alpha_0/(2\alpha_0+1)}$, so that $nh_n^2=h_n^{-1/\alpha_0}\le 2^dm$, and $\log|\Omega|\ge (m/8)\log2$, yields, for all $\omega\in\Omega$,
	$$
	KL\big(p^{(n)}_{\rho_\omega},p^{(n)}_{\rho_{\omega^{(0)}}}\big)
	\le c_32^dL^2m
	\le \frac{1}{16}\log|\Omega|,
	$$
	provided that $L\le L_0:=(\log2/(2^{d+7}c_3))^{1/2}\land1$. Theorem 2.5 in \citet{T09}, applied with the $L^2$-distance and separation $c_2Lh_n$, then implies that there exists a constant $c_4>0$ such that, for all $n$ large enough,
	$$
	\inf_{\hat\rho_n}\sup_{\substack{\rho_0\in C^\alpha([0,1]^d):\\ \rho_0\ge0,\ \|\rho_0\|_{C^\alpha}\le R}}
	P^{(n)}_{\rho_0}\Big(\|\hat\rho_n-\rho_0\|_{L^2}\ge \tfrac{1}{2}c_2Lh_n\Big)\ge c_4.
	$$
	The claim of Proposition \ref{Prop:LowerBound} then follows by Markov's inequality.
\end{proof}

%
%
%
%
%

\section{Further simulation results}
\label{Sec:MoreSym}

In this section, we complement the numerical simulation studies from Section~\ref{Sec:Simulations} of the main article. We first provide further details on the implementation of the competing methods (Section~\ref{S:Sec:SimDetails}). We then validate the MCMC sampling setup employed in the simulation study against longer multi-chain runs, and report diagnostics for the MCMC algorithm described in Section~\ref{Subsec:Algorithm} (Section~\ref{S:Sec:Diagnostics}), as well as a sensitivity analysis with respect to the numerical approximations and the prior hyper-parameters (Section~\ref{Subsec:Sensitivity}). Finally, we present additional experiments with different ground truths, and investigate the performance of our approach in the presence of purely spatial effects and over-parametrization (Sections~\ref{S:Sec:More1DExp}--\ref{Subsec:Overparam}).

%
%
%

\subsection{Implementation details for kernel and parametric methods}
\label{S:Sec:SimDetails}

%
%
%

\subsubsection*{Kernel estimators} 

The averaged kernel estimators \eqref{Eq:AverageKernel} are computed on a regular binning of the covariate space $[0,1]^d$, with $B=1000$ bins for $d=1$ and $B\times B=160\times160$ bins for $d=2$. For each replicate $i$, the numerators in the individual estimators \eqref{Eq:IndividualKernels} are computed by binning the values $Z^{(i)}(X^{(i)}_k)$, while the denominators $g^{(i)}$ are computed by binning the values of $Z^{(i)}$ at the pixel centres used for the spatial quadrature, each weighted by the pixel area. The binned values are then smoothed with the Gaussian product kernel from Section~\ref{Sec:Simulations}, $\kappa_h(u)=\prod_{j=1}^d e^{-u_j^2/(2h_j^2)}$, $u\in\R^d$, whose normalizing constant is omitted since it cancels in the ratio. The estimates are computed at the bin centres and interpolated linearly ($d=1$) or bilinearly ($d=2$) elsewhere.

The bandwidths are selected over 25 log-equispaced values in $[0.003,1]$ for $d=1$, and over a $12\times12$ grid of log-equispaced values in $[0.008,2]^2$ for $d=2$. The oracle choice minimizes the relative $L^2$-error to $\rho_0$, while the cross-validated choice maximizes the leave-one-replicate-out Poisson log-likelihood
\begin{equation}
	\label{Eq:KernelCV}
	\mathrm{CV}(h)=\frac1n\sum_{i=1}^n\bigg[\sum_{k=1}^{K^{(i)}}\log\hat\rho^{(-i)}_h
	\big(Z^{(i)}(X^{(i)}_k)\big)-\int_{\Wcal}\hat\rho^{(-i)}_h\big(Z^{(i)}(x)\big)\,dx\bigg],
\end{equation}
where $\hat\rho^{(-i)}_h=\frac1{n-1}\sum_{j\neq i}\hat\rho^{(j)}_h$ and $\hat\rho^{(j)}_h$ denotes the individual estimator \eqref{Eq:IndividualKernels} with bandwidths $h$. Both terms in \eqref{Eq:KernelCV} are evaluated via the same binning as above.

Across the experiments of Section~\ref{Sec:Simulations}, the selected bandwidths lie in the interior of the grids, with the exception of the bivariate setting with $n=50$, where the largest bandwidth in the second coordinate was selected by the oracle in two of the 50 experiments and by cross-validation in six, corresponding to estimates that are nearly constant along it. This indicates that the employed bandwidth grids are sufficiently wide, in particular allowing for adequately small bandwidths to capture the finer variations of the ground truths.

%
%
%

\subsubsection*{Log-polynomial models} 

The log-polynomial models from Section~\ref{Sec:Simulations} are fitted by maximizing the likelihood \eqref{Eq:Likelihood} over the coefficients of $P_k$, which include all interaction terms for $d=2$, yielding $k+1$ coefficients for $d=1$ and $(k+1)(k+2)/2$ for $d=2$. The integrals over $\Wcal$ are computed via the same pixel quadrature as for the posterior, cf.~Section~\ref{Subsec:SimComputation}. Since the log-likelihood is concave in the coefficients, the maximizer is computed by Newton's method with step-halving, initialized at the constant intensity matching the total number of events, and stopped when all coefficient updates fall below $10^{-9}$ in absolute value. Convergence was attained in all fits.

%
%
%

\subsection{Validation of the sampling setup and MCMC diagnostics}
\label{S:Sec:Diagnostics}

In the simulation study of Section~\ref{Sec:Simulations}, each dataset is fitted via a single chain of 60,000 iterations of the MCMC algorithm from Section~\ref{Subsec:Algorithm}, discarding the first 20,000 as burn-in and retaining every tenth of the remaining ones; see Section~\ref{Subsec:SimComputation}. Across the experiments, the step-size adaptation during the burn-in yielded acceptance rates between $28.2\%$ and $32.1\%$ for the pCN updates of the latent function, and between $40.7\%$ and $47.7\%$ for the length-scale updates, close to the respective targets of $30\%$ and $44\%$.

To validate this setup, for five datasets per setting and sample size (40 datasets in total), we ran two longer independent chains, initialized at independent prior draws, with a burn-in of 50,000 iterations and retaining every tenth of the remaining ones. These were extended in blocks of 50,000 iterations until the rank-normalized $\hat R$ was below $1.01$ and the bulk and tail ESS were at least $200$ for all monitored quantities, namely $\rho^*$, $\ell$, the log-likelihood, the mean intensity and the values of $\rho$ at a fixed grid of covariates. Their final lengths are reported in Table~\ref{Tab:Validation}. We compared the posterior means obtained from the single chains with those from the long runs. Across all datasets, the $L^2$-norm of their difference was at most $0.7\%$ of the $L^2$-norm of the long-run posterior mean, and their relative $L^2$-estimation errors differed by at most $0.005$ in absolute value, corresponding to at most $3.7\%$ of the estimation error of the long-run posterior mean.

\begin{table}[H]
	\centering
	\small
	\begin{tabular}{ll|ccc}
		& & Long runs & Max.~difference of & Max.~change in \\
		& $n$ & (iterations) & posterior means & estimation error \\
		\hline
		\textbf{Univariate} & 50 & $2\times100{,}000$ & $0.61\%$ & $0.0047$ ($2.7\%$) \\
		\textbf{setting} & 250 & $2\times100{,}000$ & $0.21\%$ & $0.0007$ ($1.1\%$) \\
		& 500 & $2\times(100{,}000$--$250{,}000)$ & $0.16\%$ & $0.0007$ ($1.5\%$) \\
		& 1000 & $2\times(150{,}000$--$250{,}000)$ & $0.19\%$ & $0.0007$ ($2.6\%$) \\
		\hline
		\textbf{Bivariate} & 50 & $2\times200{,}000$ & $0.57\%$ & $0.0026$ ($1.5\%$) \\
		\textbf{setting} & 250 & $2\times(200{,}000$--$250{,}000)$ & $0.35\%$ & $0.0005$ ($0.6\%$) \\
		& 500 & $2\times(200{,}000$--$250{,}000)$ & $0.66\%$ & $0.0023$ ($3.7\%$) \\
		& 1000 & $2\times(250{,}000$--$300{,}000)$ & $0.38\%$ & $0.0005$ ($0.9\%$) \\
		\hline
	\end{tabular}
	\caption{Validation against long two-chain runs over 5 datasets per setting and sample size: lengths of the long runs (range across experiments), maximal difference of the posterior means (relative to the $L^2$-norm of the long-run posterior means), and maximal absolute change in the relative $L^2$-estimation error (with the maximal percentage change relative to the estimation error of the long-run posterior means in brackets).}
	\label{Tab:Validation}
\end{table}

Figures~\ref{Fig:Trace1D} and \ref{Fig:Trace2D} display, for one dataset with $n=1000$ in the univariate and bivariate scenarios, respectively, the trace plots of the intensity ceiling $\rho^*$, the length-scales and the log-likelihood for the two chains of the long runs. The figures show that the two chains, initialized at independent prior draws, explore the same regions and exhibit no trends after the burn-in. The length-scales display stronger autocorrelation than $\rho^*$ and the log-likelihood, in line with the slower mixing discussed in Section~\ref{Subsec:SimComputation}, which motivated the use of the long runs for their posterior summaries and for validation.

\begin{figure}[H]
	\centering
	\includegraphics[width=\linewidth]{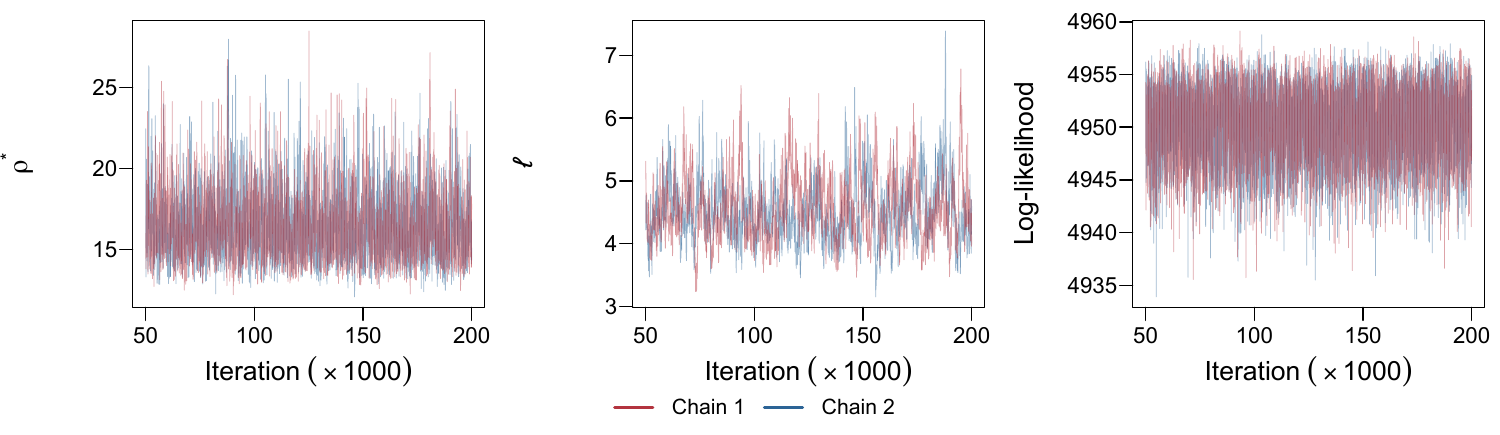}
	\caption{Univariate setting, $n=1000$: trace plots of $\rho^*$, $\ell$ and the log-likelihood for the two chains of the long run, after burn-in and thinning.}
	\label{Fig:Trace1D}
\end{figure}

\begin{figure}[H]
	\centering
	\includegraphics[width=\linewidth]{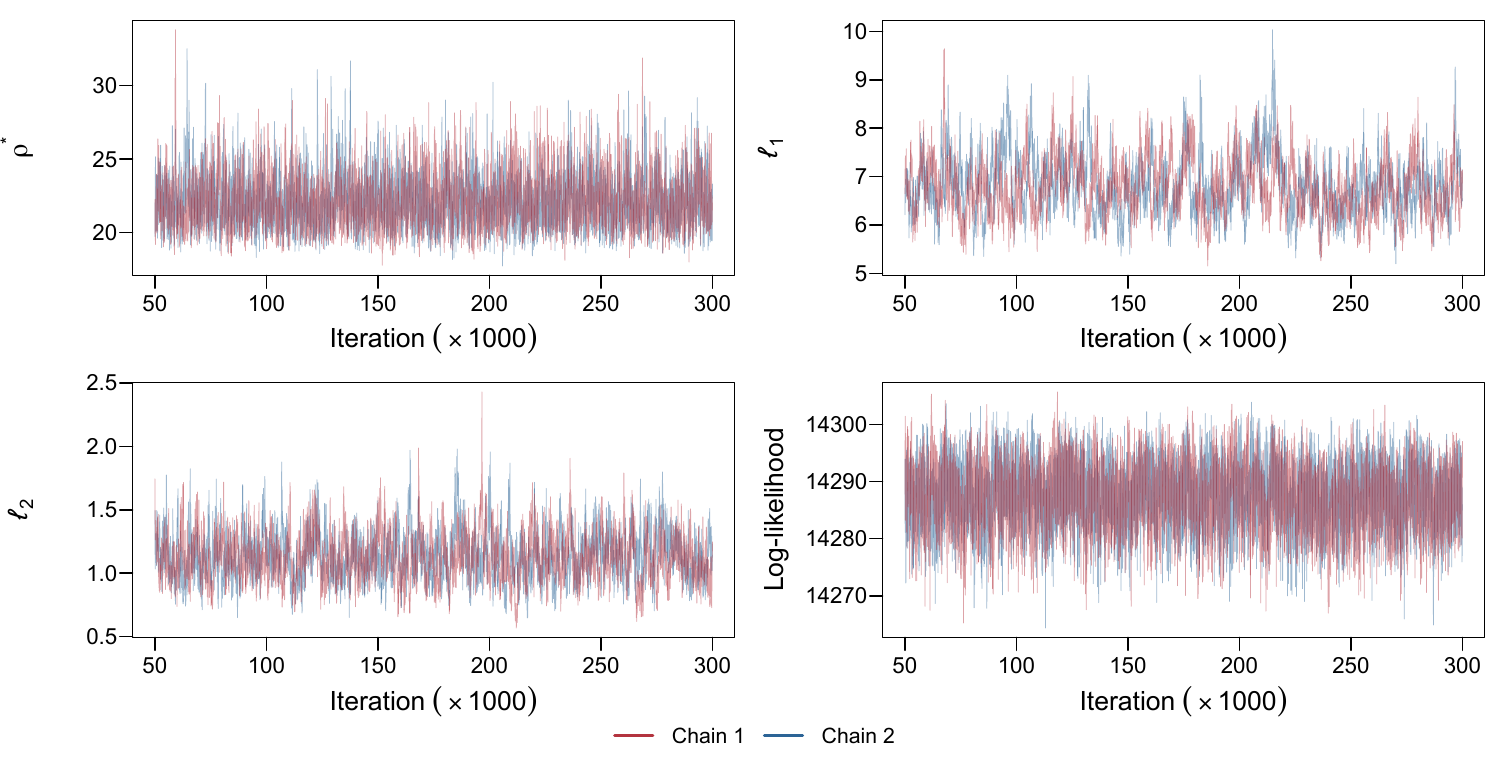}
	\caption{Bivariate setting, $n=1000$: trace plots of $\rho^*$, $\ell_1$, $\ell_2$ and the log-likelihood for the two chains of the long run, after burn-in and thinning.}
	\label{Fig:Trace2D}
\end{figure}

%
%
%

\subsection{Sensitivity analysis}
\label{Subsec:Sensitivity}

To assess the sensitivity of the proposed multi-bandwidth Gaussian process methods for covariate-based intensities to the numerical approximations underlying the MCMC sampler described in Section~\ref{Subsec:Algorithm} and to the choice of the prior hyper-parameters, we conducted a one-factor-at-a-time analysis around the numerical experiments of Section~\ref{Sec:Simulations}. For both the univariate and bivariate scenarios, and for $n\in\{50,1000\}$ (representing the smallest and largest sample sizes, respectively), we refitted the model over five datasets each, changing one factor at a time, as listed in Table~\ref{Tab:SensDesign}, with the same MCMC sampling setup as in the simulation study. To investigate the sensitivity to the spatial quadrature, we simulated new datasets on a finer grid of $200\times200$ pixels, and approximated the integrals over $\Wcal$ on $q\times q$ sub-sampled pixels, $q\in\{25,50,100\}$, with $q=50$ (the setting of Section~\ref{Subsec:SimComputation}) as the baseline.

\begin{table}[H]
	\centering
	\small
	\begin{tabular}{l|c|c}
		\textbf{Factor} & \textbf{Univariate setting} & \textbf{Bivariate setting} \\
		\hline
		Discretization of $w$ in \eqref{Eq:Discret} & $V\in\{100,\mathbf{200},400\}$ & $15\times15$, $\mathbf{20\times20}$, $35\times35$ grids \\
		Binning of the covariate values & $500$, $\mathbf{2000}$, $8000$ bins & $100^2$, $\mathbf{200^2}$, $400^2$ bins \\
		\hline
		Spatial quadrature & \multicolumn{2}{c}{$q\times q$ pixels, $q\in\{25,\mathbf{50},100\}$} \\
		\hline
		Truncation of $\Pi_{\rho^*}$ & \multicolumn{2}{c}{$c_{\rho^*}\in\{\mathbf{100},\infty\}$} \\
		Rate of $\Pi_{\rho^*}$ & \multicolumn{2}{c}{$b_{\rho^*}\in\{0.01,\mathbf{0.1},2\}$} \\
		Tilt of $\Pi_\ell$ & \multicolumn{2}{c}{$\lambda\in\{0.05,\mathbf{0.1},0.25\}$} \\
		Rate of $\Pi_\ell$ & \multicolumn{2}{c}{$b_\ell\in\{1.5,\mathbf{2.5},4\}$} \\
		\hline
	\end{tabular}
	\caption{Factors of the sensitivity analysis, with the baseline values used in Section~\ref{Subsec:SimComputation} reported in bold.}
	\label{Tab:SensDesign}
\end{table}

Table~\ref{Tab:Sensitivity} reports, for each configuration, the relative $L^2$-distance, in percent, between the resulting posterior mean and the baseline one (columns `Post.~Mean Change'), to be compared with the Monte Carlo error reported at the bottom of each panel, as well as the absolute and percentage change in the relative $L^2$-estimation error with respect to the baseline settings of Section~\ref{Subsec:SimComputation} (columns `Est.~Error Change').

The numerical approximations have minor effects: most changes in the error are within $1\%$, and the posterior means differ from the baseline ones by amounts comparable to the Monte Carlo error. In particular, the finer computational settings change the error in both directions, with shifts of the posterior mean of the same order as the Monte Carlo error, indicating that the numerical approximations at the baseline settings are already negligible compared to the statistical and Monte Carlo errors. The only notable exception among the numerical approximations is the coarser $15\times15$ grid for $w$ in the bivariate setting with $n=1000$, which increases the error by $3.3\%$ and shifts the posterior mean by $1.6\%$. Here, the posterior mean of $\ell_1$ is about $6.7$ (Table~\ref{Tab:Hyper}), so that the correlation length $1/\ell_1\approx0.15$ of $w$ in the first variable spans only about two spacings of the $15\times15$ grid, against about three for the baseline $20\times20$ grid, and more than four in the remaining experiments of this sensitivity analysis. This suggests checking a posteriori that the correlation lengths $1/\ell_j$, $j=1,\dots,d$, implied by the posterior estimates of the length-scales are sufficiently large compared to the resolution of the grid employed for the discretization of $w$, and refining the grid otherwise. Note that for $n=1000$ the relative changes in the error are taken with respect to small baseline errors; e.g., the change of $+1.4\%$ for the coarser quadrature with $q=25$ in the univariate setting corresponds to an absolute increase of $0.0005$, with a shift of the posterior mean within the Monte Carlo error.

Removing the truncation of the intensity ceiling prior $\Pi_{\rho^*}$ (i.e., setting $c_{\rho^*}=\infty$) leaves the results unchanged, since the baseline truncation level $c_{\rho^*}+\log n$, with $c_{\rho^*}=100$, lies far above the posterior distribution of $\rho^*$ in all scenarios. This suggests that the truncation is not needed in practice, although our theory covers an untruncated prior only for the cut-off distance $d_Z(\rho\land M,\rho_0\land M)$, cf.~Remark~\ref{Rem:BoundInt}. As a rule of thumb, we recommend choosing the truncation level conservatively, well above any plausible value of $\|\rho_0\|_\infty$, so that it does not affect the posterior while retaining the guarantees of Theorem~\ref{Theo:GPRates}.

A flatter prior on $\rho^*$ ($b_{\rho^*}=0.01$) changes the error by at most $3.1\%$, and the posterior mean by about $1\%$ for $n=50$ and by an amount comparable to the Monte Carlo error for $n=1000$. In contrast, the strongly informative choice $b_{\rho^*}=2$ pulls the posterior distributions of $\rho^*$ below the maximum of $\rho_0$ for $n=50$ in both the univariate and bivariate settings, resulting in changes in the posterior means larger than $8.5\%$ and increases in the error of over $20\%$, the largest recorded throughout the sensitivity analysis. For $n=1000$, the shifts of the posterior means reduce to only about $2\%$, and the error increases by at most $0.007$ in absolute value (although by $8.5\%$ and $21.6\%$ in relative terms, owing to the small baseline errors), suggesting that the influence of the prior on $\rho^*$ vanishes as the sample size grows, in agreement with Theorem~\ref{Theo:GPRates}, whose conclusion holds for any fixed choice of the hyper-parameters.

No noticeable effect was recorded for the tilt parameter $\lambda$, with changes in the posterior mean within the Monte Carlo error. Finally, the rate $b_\ell$ of the length-scale hyper-prior mostly influences the experiments with $n=50$, in particular in the univariate setting. There, a smaller rate yields greater prior mass on larger length-scales, i.e.~on functions with more abrupt variations, reducing the bias towards flatter estimates discussed in Section~\ref{Subsec:SimResults}. In the univariate setting, the error decreases by $9.5\%$ for $b_\ell=1.5$ and increases by $13\%$ for $b_\ell=4$, with shifts of the posterior mean of about $2\%$; in the bivariate setting, the changes in the error are about $4\%$. For $n=1000$, the effect of $b_\ell$ on the error is again below $2\%$.

\begin{table}[H]
	\centering
	\small
	\setlength{\tabcolsep}{5pt}
	\begin{tabular}{l|cc|cc}
		& \multicolumn{2}{c|}{$n=50$} & \multicolumn{2}{c}{$n=1000$} \\
		& Post.~Mean & Est.~Error & Post.~Mean & Est.~Error \\
		\textbf{Factor} & Change & Change & Change & Change \\
		\hline
		\multicolumn{5}{c}{\textbf{Univariate setting}} \\
		$w$: $V=100$ & $0.3\%$ & $-0.0006$ ($-0.4\%$) & $0.1\%$ & $-0.0004$ ($-1.3\%$) \\
		$w$: $V=400$ & $0.4\%$ & $-0.0005$ ($-0.5\%$) & $0.2\%$ & $-0.0004$ ($-1.0\%$) \\
		Binning: $500$ bins & $0.2\%$ & $-0.0010$ ($-0.7\%$) & $0.1\%$ & $-0.0001$ ($-0.4\%$) \\
		Binning: $8000$ bins & $0.2\%$ & $-0.0012$ ($-0.9\%$) & $0.1\%$ & $-0.0001$ ($-0.3\%$) \\
		\hline
		Quadrature: $q=25$ & $0.3\%$ & $+0.0009$ ($+0.6\%$) & $0.1\%$ & $+0.0005$ ($+1.4\%$) \\
		Quadrature: $q=100$ & $0.2\%$ & $+0.0005$ ($+0.5\%$) & $0.1\%$ & $+0.0001$ ($+0.4\%$) \\
		\hline
		$c_{\rho^*}=\infty$ & $0.0\%$ & $0.0000$ ($0.0\%$) & $0.0\%$ & $0.0000$ ($0.0\%$) \\
		$b_{\rho^*}=0.01$ & $1.0\%$ & $-0.0049$ ($-3.1\%$) & $0.2\%$ & $-0.0002$ ($-0.5\%$) \\
		$b_{\rho^*}=2$ & $10.8\%$ & $+0.0611$ ($+27.6\%$) & $1.9\%$ & $+0.0068$ ($+21.6\%$) \\
		$\lambda=0.05$ & $0.3\%$ & $-0.0027$ ($-1.7\%$) & $0.1\%$ & $-0.0001$ ($-0.5\%$) \\
		$\lambda=0.25$ & $0.2\%$ & $+0.0012$ ($+1.1\%$) & $0.1\%$ & $-0.0001$ ($-0.4\%$) \\
		$b_\ell=1.5$ & $2.1\%$ & $-0.0140$ ($-9.5\%$) & $0.2\%$ & $+0.0006$ ($+1.7\%$) \\
		$b_\ell=4$ & $2.5\%$ & $+0.0206$ ($+13.0\%$) & $0.4\%$ & $+0.0001$ ($+0.2\%$) \\
		\hline
		Monte Carlo error & $0.4\%$ & -- & $0.2\%$ & -- \\
		\hline
		\multicolumn{5}{c}{\textbf{Bivariate setting}} \\
		$w$: $15\times15$ grid & $0.5\%$ & $+0.0012$ ($+0.6\%$) & $1.6\%$ & $+0.0020$ ($+3.3\%$) \\
		$w$: $35\times35$ grid & $0.4\%$ & $0.0000$ ($0.0\%$) & $0.7\%$ & $-0.0002$ ($-0.4\%$) \\
		Binning: $100^2$ bins & $0.3\%$ & $+0.0005$ ($+0.3\%$) & $0.3\%$ & $-0.0003$ ($-0.7\%$) \\
		Binning: $400^2$ bins & $0.3\%$ & $+0.0001$ ($0.0\%$) & $0.3\%$ & $+0.0002$ ($+0.3\%$) \\
		\hline
		Quadrature: $q=25$ & $0.7\%$ & $+0.0006$ ($+0.4\%$) & $0.4\%$ & $+0.0001$ ($+0.2\%$) \\
		Quadrature: $q=100$ & $0.2\%$ & $+0.0006$ ($+0.4\%$) & $0.3\%$ & $-0.0003$ ($-0.5\%$) \\
		\hline
		$c_{\rho^*}=\infty$ & $0.0\%$ & $0.0000$ ($0.0\%$) & $0.0\%$ & $0.0000$ ($0.0\%$) \\
		$b_{\rho^*}=0.01$ & $1.2\%$ & $+0.0009$ ($+0.4\%$) & $0.3\%$ & $-0.0001$ ($-0.3\%$) \\
		$b_{\rho^*}=2$ & $8.7\%$ & $+0.0411$ ($+20.6\%$) & $1.5\%$ & $+0.0049$ ($+8.5\%$) \\
		$\lambda=0.05$ & $0.3\%$ & $-0.0002$ ($-0.1\%$) & $0.2\%$ & $-0.0003$ ($-0.6\%$) \\
		$\lambda=0.25$ & $0.4\%$ & $+0.0019$ ($+1.0\%$) & $0.3\%$ & $+0.0002$ ($+0.4\%$) \\
		$b_\ell=1.5$ & $3.6\%$ & $-0.0076$ ($-3.9\%$) & $0.4\%$ & $-0.0003$ ($-0.5\%$) \\
		$b_\ell=4$ & $2.0\%$ & $+0.0089$ ($+4.4\%$) & $0.5\%$ & $+0.0007$ ($+1.2\%$) \\
		\hline
		Monte Carlo error & $0.4\%$ & -- & $0.4\%$ & -- \\
		\hline
	\end{tabular}
	\caption{Sensitivity analysis, medians over 5 datasets, for the univariate (top) and bivariate (bottom) settings. Post.~Mean Change: $L^2$-distance between the posterior means under each configuration and under the baseline, in \% of the $L^2$-norm of the latter. Est.~Error Change: absolute change of the relative $L^2$-estimation error with respect to the baseline (with relative change in percent in brackets). Monte Carlo error: Monte Carlo error of the posterior mean change, median over the settings.}
	\label{Tab:Sensitivity}
\end{table}

%
%
%

\subsection{Additional experiments with univariate covariates}
\label{S:Sec:More1DExp}

On the window $\Wcal = [-1/2,1/2]^2$, we take the univariate covariate random field $Z$ from Section \ref{Subsec:DataGen}, and consider two additional ground truths, respectively defined as
\begin{equation}
	\label{Eq:1DTruth2}
	\rho_0(z) = 2 e^{3 (1-z) -1}, \qquad z\in[0,1],
\end{equation}
and
\begin{equation}
	\label{Eq:1DTruth3}
	\rho_0(z) = 2 + 2 P(z; 3/4) - 2P(z; 1/4), \qquad z\in[0,1],
\end{equation}
where $P(z;c) = 1 - p(16|z-c|/3)$ is a smooth bump function centered at $c$ and supported on $[c-3/16,c+3/16]$, with $p(t) = 6t^5 - 15 t^4+10t^3$ for $t\in[0,1]$ and $p(t)=1$ for $t>1$. The former is log-affine, cf.~Figure \ref{Fig:1DResults2} (top row), and thus is correctly specified by the log-polynomial models, while the latter exhibits both positive and negative deviations from a flat baseline, cf.~Figure \ref{Fig:1DResults2} (bottom row). For both, we repeated the univariate experiments of Section~\ref{Sec:Simulations}, keeping all other specifications unchanged, namely the sample sizes and the number of datasets per sample size, the competing methods, the prior hyper-parameters, the numerical approximations and the MCMC settings of Section~\ref{Subsec:SimComputation}. Figure~\ref{Fig:1DResults2} displays the estimates obtained on representative datasets for $n\in\{50,250,$ $1000\}$, and Table~\ref{Tab:1DTab2} reports the relative $L^2$-errors. 

\begin{figure}[H]
	\centering
	\includegraphics[width=\textwidth]{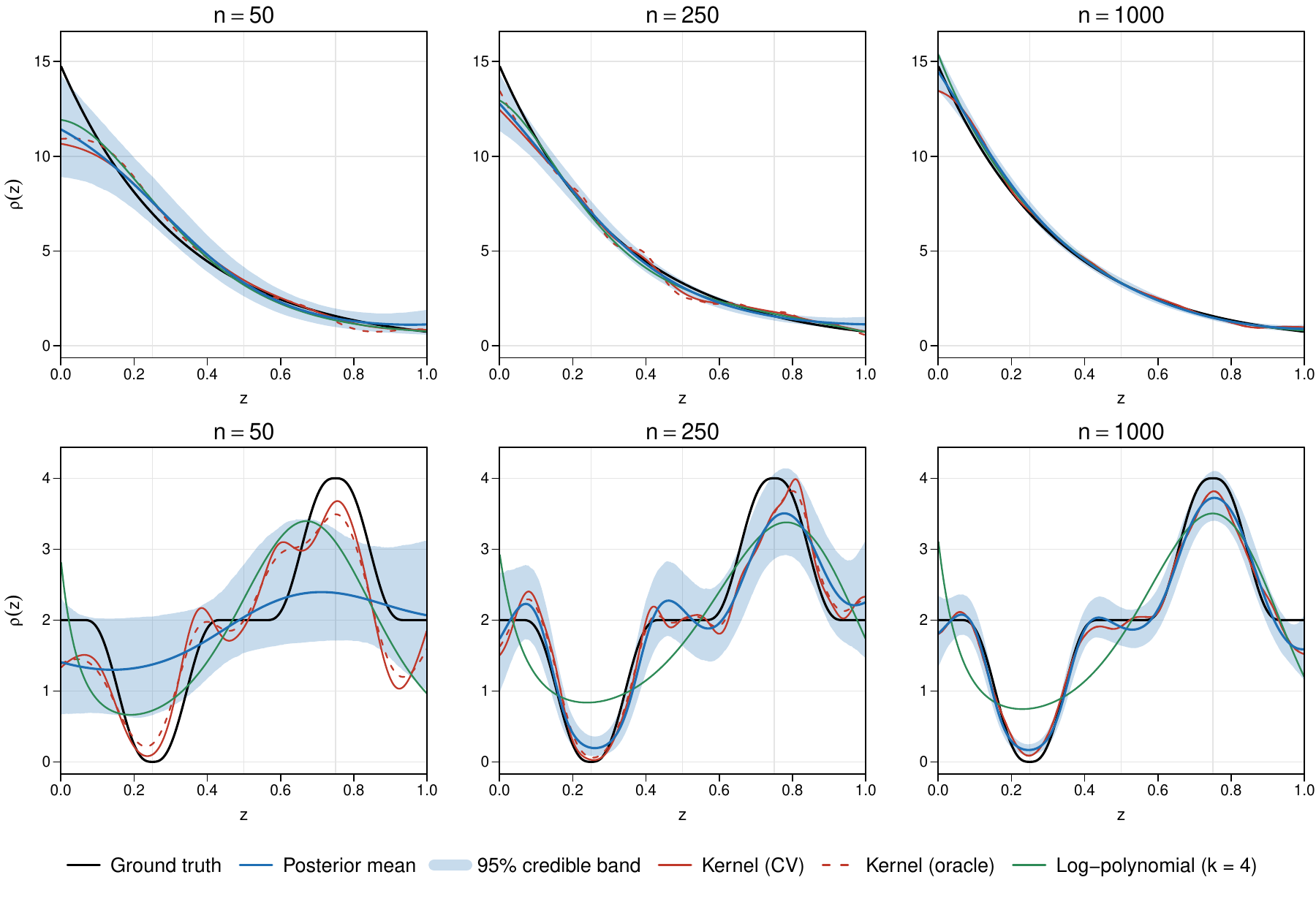}
	\caption{Univariate setting with ground truths \eqref{Eq:1DTruth2} (top row) and \eqref{Eq:1DTruth3} (bottom row), $n=50,250,1000$: posterior mean with pointwise $95\%$-credible intervals, kernel estimates with oracle and cross-validated bandwidths, and log-polynomial estimate of degree 4. The ground truths are shown in black.}
	\label{Fig:1DResults2}
\end{figure}

\begin{table}[H]
	\centering
	\small
	\begin{tabular}{l|cccc}
		& $n=50$ & $n=250$ & $n=500$ & $n=1000$ \\
		\hline
		\noalign{\vskip 3pt}
		\multicolumn{5}{c}{\textbf{Ground truth \eqref{Eq:1DTruth2}}}\\
		\noalign{\vskip 3pt}
		Posterior mean & 0.118 (0.043) & 0.063 (0.021) & 0.047 (0.016) & 0.033 (0.010) \\
		\noalign{{\color{gray!60}\hrule height 0.3pt}}
		Kernel, CV & 0.141 (0.058) & 0.079 (0.023) & 0.059 (0.017) & 0.045 (0.011) \\
		Kernel, oracle & 0.111 (0.040) & 0.068 (0.020) & 0.053 (0.016) & 0.040 (0.010) \\
		\noalign{{\color{gray!60}\hrule height 0.3pt}}
		Log-polynomial, $k=1$ & \textbf{0.065} (0.048) & \textbf{0.034} (0.022) & \textbf{0.025} (0.017) & \textbf{0.017} (0.009) \\
		Log-polynomial, $k=2$ & 0.085 (0.050) & 0.042 (0.021) & 0.030 (0.017) & 0.020 (0.010) \\
		Log-polynomial, $k=3$ & 0.100 (0.053) & 0.046 (0.022) & 0.034 (0.017) & 0.025 (0.011) \\
		Log-polynomial, $k=4$ & 0.109 (0.057) & 0.051 (0.022) & 0.036 (0.018) & 0.027 (0.012) \\
		\hline
		\noalign{\vskip 3pt}
		\multicolumn{5}{c}{\textbf{Ground truth \eqref{Eq:1DTruth3}}}\\
		\noalign{\vskip 3pt}
		Posterior mean & 0.338 (0.058) & 0.140 (0.031) & 0.093 (0.016) & \textbf{0.070} (0.015) \\
		\noalign{{\color{gray!60}\hrule height 0.3pt}}
		Kernel, CV & 0.239 (0.066) & 0.138 (0.034) & 0.102 (0.017) & 0.083 (0.019) \\
		Kernel, oracle & \textbf{0.220} (0.060) & \textbf{0.121} (0.027) & \textbf{0.091} (0.015) & 0.072 (0.016) \\
		\noalign{{\color{gray!60}\hrule height 0.3pt}}
		Log-polynomial, $k=1$ & 0.409 (0.025) & 0.391 (0.004) & 0.389 (0.002) & 0.389 (0.002) \\
		Log-polynomial, $k=2$ & 0.405 (0.034) & 0.382 (0.014) & 0.375 (0.007) & 0.375 (0.005) \\
		Log-polynomial, $k=3$ & 0.275 (0.049) & 0.223 (0.015) & 0.212 (0.005) & 0.209 (0.005) \\
		Log-polynomial, $k=4$ & 0.288 (0.047) & 0.233 (0.017) & 0.219 (0.009) & 0.214 (0.006) \\
		\hline
	\end{tabular}
	\caption{Mean relative $L^2$-errors (and between-dataset standard deviations) over 50 independent datasets for each ground truth and sample size. 
	}
	\label{Tab:1DTab2}
\end{table}

For the log-affine ground truth \eqref{Eq:1DTruth2}, the correctly specified log-polynomial model of degree~1 attains the lowest errors, as expected; among the nonparametric methods, the posterior mean resulting from the proposed Bayesian approach is comparable to the kernel estimator with oracle bandwidth for $n=50$, and the most accurate from $n=250$ onwards. For the ground truth \eqref{Eq:1DTruth3}, all log-polynomial models are misspecified, and their errors saturate around $0.21$ or above. For $n=50$, the posterior mean is less accurate than the kernel estimators, owing to the bias towards flatter estimates discussed in Section~\ref{Subsec:SimResults}, which is more pronounced for this more variable ground truth (cf.~Figure~\ref{Fig:1DResults2}, bottom left), and could be mitigated by a smaller rate $b_\ell$, as suggested by the sensitivity analysis of Section~\ref{Subsec:Sensitivity}. The gap closes as $n$ grows: the posterior mean is on par with the oracle kernel estimator for $n=500$, and the most accurate method for $n=1000$.

%
%
%

\subsection{Additional experiments with bivariate covariates}
\label{S:Sec:More2DExp}

We next take the bivariate covariate process $Z=(Z_1,Z_2)$ from Section~\ref{Subsec:DataGen}, and consider the additional ground truths, defined for $z\in[0,1]^2$,
\begin{equation}
	\label{Eq:2DTruth3}
	\rho_0(z) = \max\big\{0,\ 30 - 90\, f_{SN}\big(z; (0.3,0.3), 0.5 I_2, (-1,-1)\big)\big\},
\end{equation}
and
\begin{equation}
	\label{Eq:2DTruth2}
	\rho_0(z) = 6 f_{SN}\big(z; (0.3,0.8), 0.03 I_2, (-1,-1)\big)
	+ 14 f_{SN}\big(z; (0.7,0.2), 0.05 I_2, (3,-2)\big),
\end{equation}
where $f_{SN}(\cdot\,;\xi,\Omega,\alpha)$ denotes the bivariate skew-normal density with location $\xi$, scale matrix $\Omega$ and shape $\alpha$, and $I_2$ is the $2\times2$ identity matrix. Three independent realizations of $Z$ are shown in Figure~\ref{Fig:2DPattern}, illustrating the different length-scales of the two covariates $Z_1$ and $Z_2$. The first ground truth above increases smoothly from a region of zero intensity near the origin towards the top-right corner of the covariate space $[0,1]^2$, cf.~Figure~\ref{Fig:2DResults3}, while the second consists of two separated skewed bumps of different heights and widths, cf.~Figure~\ref{Fig:2DResults2}.

\begin{figure}[H]
	\centering
	\begin{subfigure}[b]{\textwidth}
		\centering
		\includegraphics[width=\textwidth]{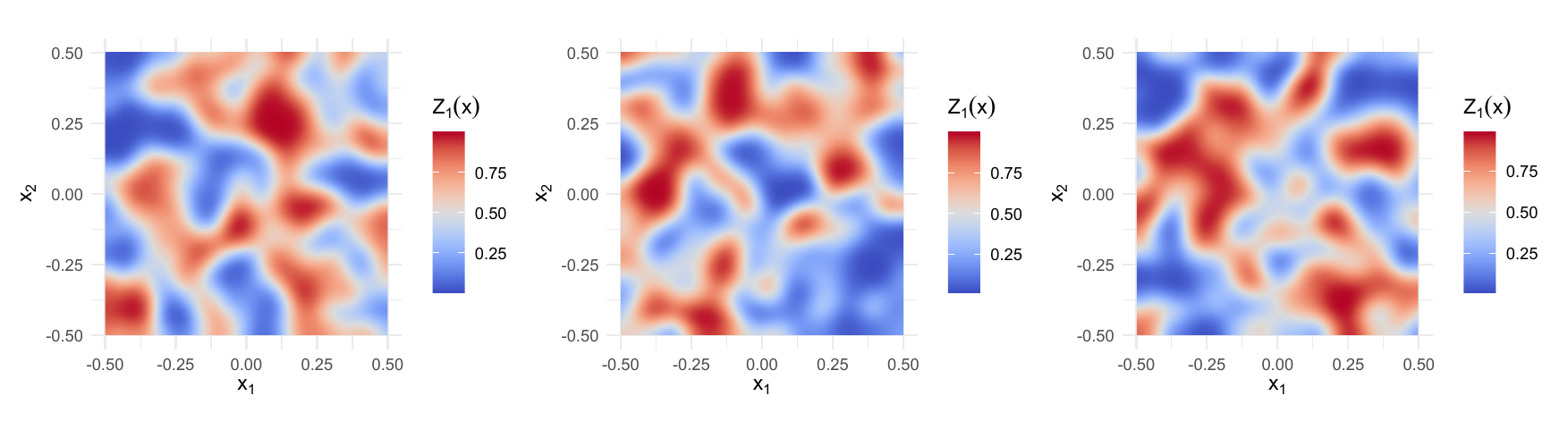}
	\end{subfigure}
	\begin{subfigure}[b]{\textwidth}
		\centering
		\includegraphics[width=\textwidth]{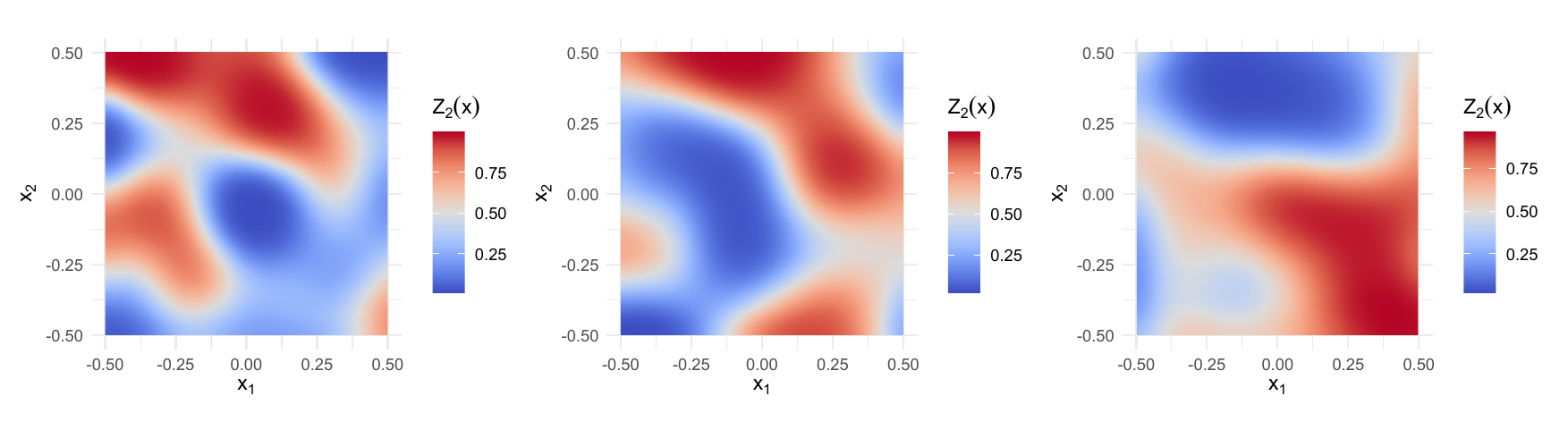}
	\end{subfigure}
	\caption{Three independent realizations of the bivariate covariate process: $Z_1$ (top row) and $Z_2$ (bottom row).}
	\label{Fig:2DPattern}
\end{figure}

For both ground truths, we repeated the bivariate experiments of Section~\ref{Sec:Simulations}, keeping all other specifications unchanged, as in the additional experiments presented in Section~\ref{S:Sec:More1DExp}. Figures~\ref{Fig:2DResults3} and~\ref{Fig:2DResults2} display the estimates obtained on representative datasets, and Table~\ref{Tab:2DResults3} reports the relative $L^2$-errors.

Across all combinations of ground truth and sample size, the posterior mean resulting from the proposed multi-bandwidth Gaussian process method attains the lowest errors among all considered procedures. This includes the smallest sample size, $n=50$, in contrast to the bivariate experiments of Section~\ref{Sec:Simulations}. For the ground truth \eqref{Eq:2DTruth3}, the errors for the posterior mean are about $30\%$ lower than those of the kernel estimator with oracle bandwidth, while the log-polynomial model of degree~4 approaches the accuracy of the posterior mean for larger sample sizes. For the more variable ground truth \eqref{Eq:2DTruth2}, the kernel estimators, which employ a single bandwidth per coordinate over the whole covariate space, struggle to reconstruct both bumps, and their errors are roughly twice those of the posterior mean for $n\in\{250,500,1000\}$; the log-polynomial models are misspecified, with errors saturating at $0.257$ or above.

\begin{figure}[H]
	\centering
	\includegraphics[width=\textwidth]{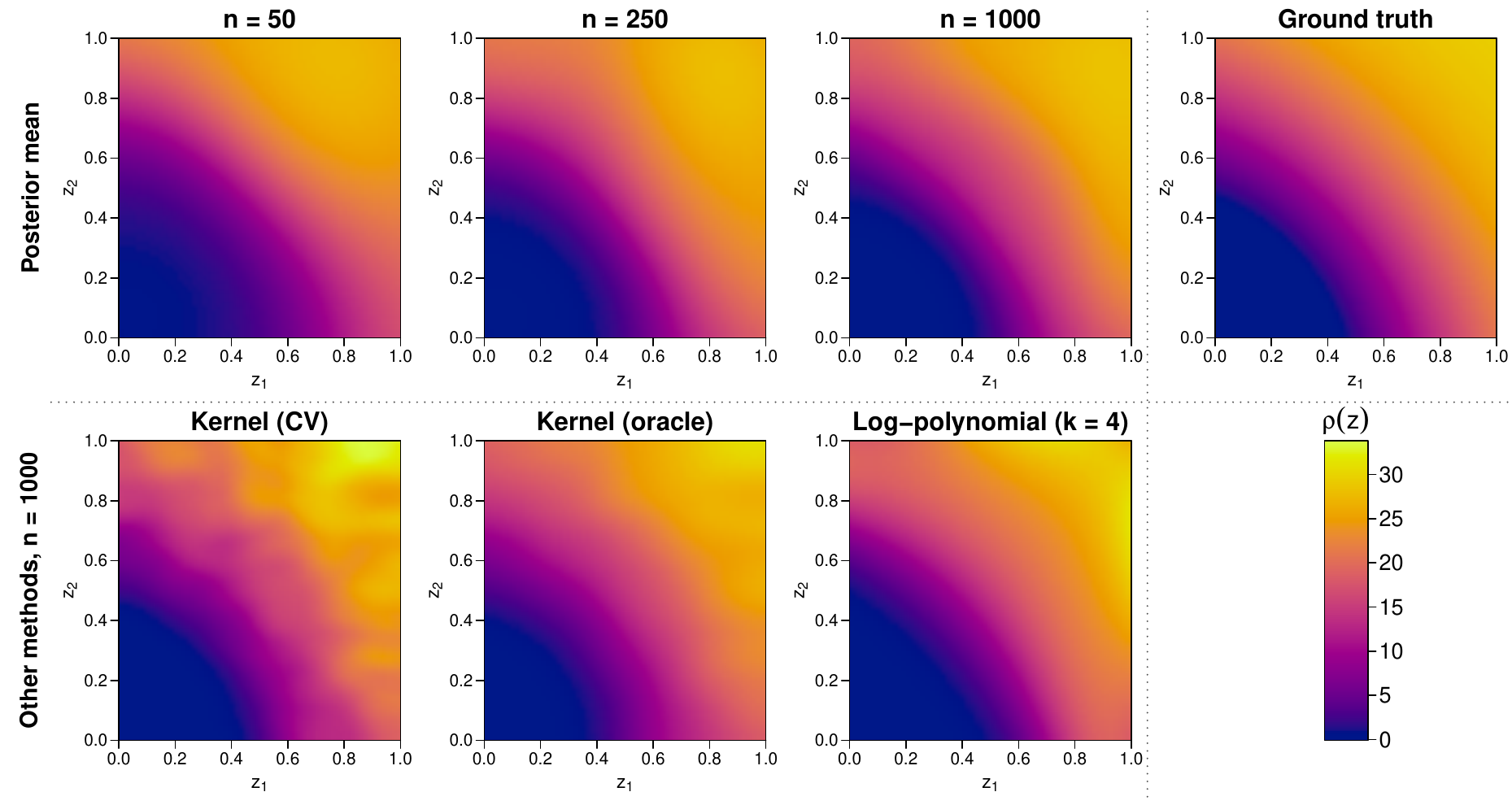}
	\caption{Bivariate setting with ground truth \eqref{Eq:2DTruth3}. Top row: posterior means for $n=50,250,1000$ and the ground truth. Bottom row: kernel estimates with cross-validated and oracle bandwidths, and log-polynomial estimate of degree 4, for $n=1000$.}
	\label{Fig:2DResults3}
\end{figure}

\begin{figure}[H]
	\centering
	\includegraphics[width=\textwidth]{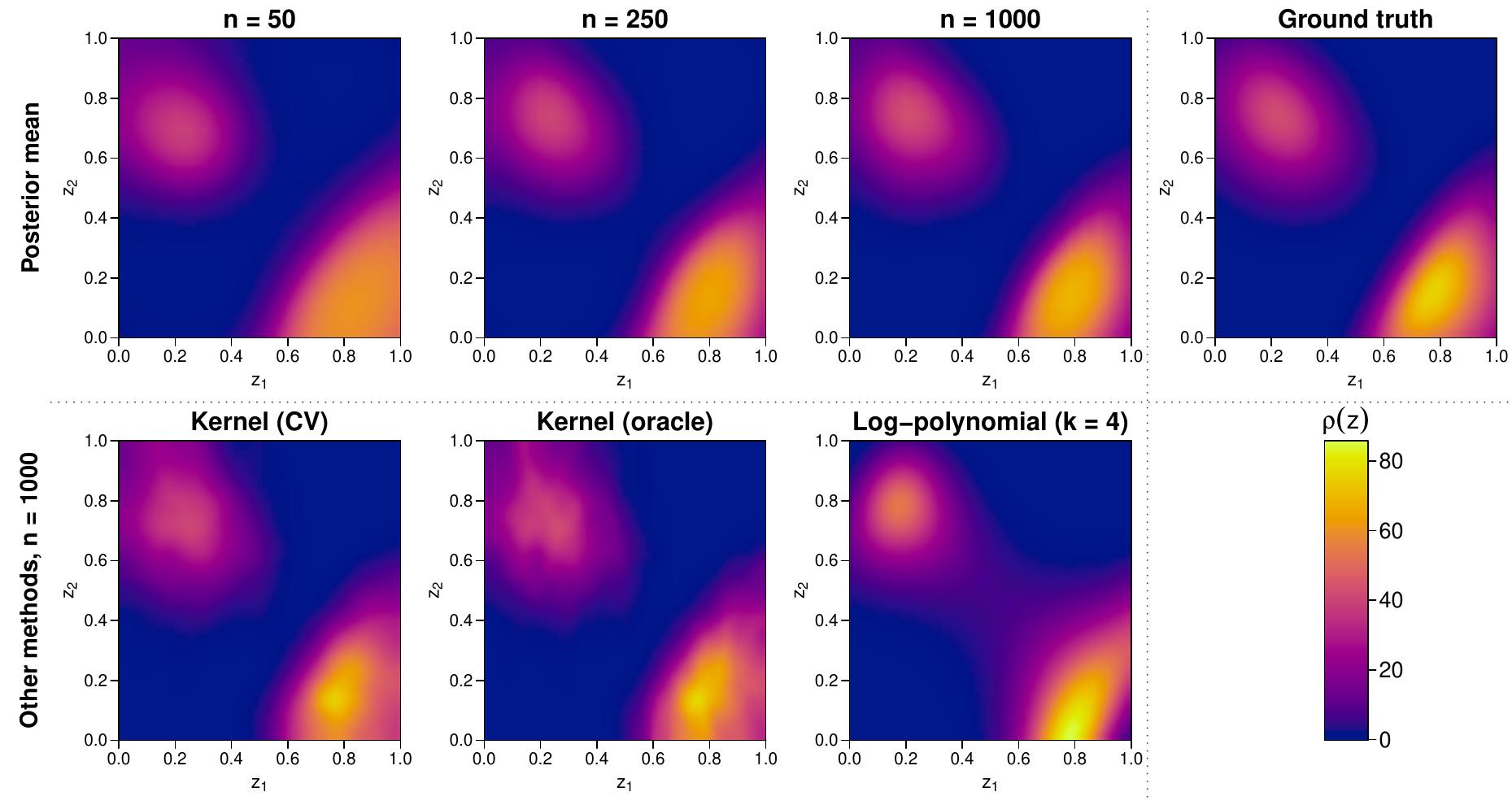}
	\caption{Bivariate setting with ground truth \eqref{Eq:2DTruth2}. Top row: posterior means for $n=50,250,1000$ and the ground truth. Bottom row: kernel estimates with cross-validated and oracle bandwidths, and log-polynomial estimate of degree 4, for $n=1000$.}
	\label{Fig:2DResults2}
\end{figure}

\begin{table}[H]
	\centering
	\small
	\begin{tabular}{l|cccc}
		& $n=50$ & $n=250$ & $n=500$ & $n=1000$ \\
		\hline
		\noalign{\vskip 3pt}
		\multicolumn{5}{c}{\textbf{Ground truth \eqref{Eq:2DTruth3}}}\\
		\noalign{\vskip 3pt}
		Posterior mean & \textbf{0.089} (0.017) & \textbf{0.052} (0.010) & \textbf{0.042} (0.007) & \textbf{0.034} (0.005) \\
		\noalign{{\color{gray!60}\hrule height 0.3pt}}
		Kernel, CV & 0.149 (0.038) & 0.100 (0.019) & 0.077 (0.011) & 0.065 (0.009) \\
		Kernel, oracle & 0.123 (0.022) & 0.077 (0.012) & 0.064 (0.009) & 0.053 (0.006) \\
		\noalign{{\color{gray!60}\hrule height 0.3pt}}
		Log-polynomial, $k=1$ & 0.287 (0.021) & 0.283 (0.008) & 0.283 (0.006) & 0.283 (0.004) \\
		Log-polynomial, $k=2$ & 0.145 (0.022) & 0.121 (0.007) & 0.119 (0.005) & 0.119 (0.004) \\
		Log-polynomial, $k=3$ & 0.132 (0.026) & 0.079 (0.009) & 0.073 (0.006) & 0.070 (0.004) \\
		Log-polynomial, $k=4$ & 0.143 (0.027) & 0.069 (0.008) & 0.059 (0.005) & 0.051 (0.006) \\
		\hline
		\noalign{\vskip 3pt}
		\multicolumn{5}{c}{\textbf{Ground truth \eqref{Eq:2DTruth2}}}\\
		\noalign{\vskip 3pt}
		Posterior mean & \textbf{0.204} (0.024) & \textbf{0.088} (0.013) & \textbf{0.067} (0.008) & \textbf{0.052} (0.006) \\
		\noalign{{\color{gray!60}\hrule height 0.3pt}}
		Kernel, CV & 0.229 (0.035) & 0.157 (0.022) & 0.137 (0.018) & 0.117 (0.011) \\
		Kernel, oracle & 0.218 (0.029) & 0.153 (0.019) & 0.134 (0.016) & 0.114 (0.010) \\
		\noalign{{\color{gray!60}\hrule height 0.3pt}}
		Log-polynomial, $k=1$ & 0.728 (0.014) & 0.723 (0.007) & 0.723 (0.004) & 0.722 (0.003) \\
		Log-polynomial, $k=2$ & 0.586 (0.024) & 0.576 (0.007) & 0.576 (0.006) & 0.575 (0.005) \\
		Log-polynomial, $k=3$ & 0.577 (0.044) & 0.556 (0.013) & 0.557 (0.010) & 0.554 (0.007) \\
		Log-polynomial, $k=4$ & 0.285 (0.024) & 0.260 (0.007) & 0.257 (0.006) & 0.257 (0.004) \\
		\hline
	\end{tabular}
	\caption{Mean relative $L^2$-errors (and between-dataset standard deviations) over 50 independent datasets for each ground truth and sample size. 
	}
	\label{Tab:2DResults3}
\end{table}

%
%
%

\subsection{Experiments with deterministic covariates}
\label{Subsec:ExpDetCovariates}

Motivated by Remark~\ref{Rem:DetCov}, we also consider an example with both random and deterministic covariates. Specifically, we take $Z=(Z_1,Z_2)$, where $Z_1$ is the univariate covariate random field from Sections~\ref{Subsec:DataGen} and \ref{S:Sec:More1DExp}, and $Z_2(x)=1/2+x_1$, $x\in\Wcal = [-1/2,1/2]^2$, is a deterministic covariate, given by the (shifted) first spatial coordinate and accounting for a purely spatial effect. In this scenario, we consider the ground truth
\begin{equation}
	\label{Eq:2DTruthCoord}
	\rho_0(z) = 15\, z_2\, f_{SN}(z_1; 0.8, 0.3, -5), \qquad z\in[0,1]^2,
\end{equation}
cf.~Figure~\ref{Fig:2DResultsCoord}. All other specifications of the bivariate experiments of Section~\ref{Sec:Simulations} were kept unchanged. Figure~\ref{Fig:2DResultsCoord} displays the estimates obtained on representative datasets, and Table~\ref{Tab:2DResultsCoord} reports the relative $L^2$-errors.

The posterior mean recovers both the dependence of the intensity on the random covariate and the linear spatial trend, the latter being clearly visible from $n=250$ onwards. For $n=50$ and $n=250$, the log-polynomial model of degree~4 attains the lowest errors, with the posterior mean comparable to the kernel estimators. As $n$ grows, the posterior mean becomes the most accurate method, matching the log-polynomial model for $n=500$ and improving on all competitors by about $20\%$ for $n=1000$.

\begin{figure}[H]
	\centering
	\includegraphics[width=\textwidth]{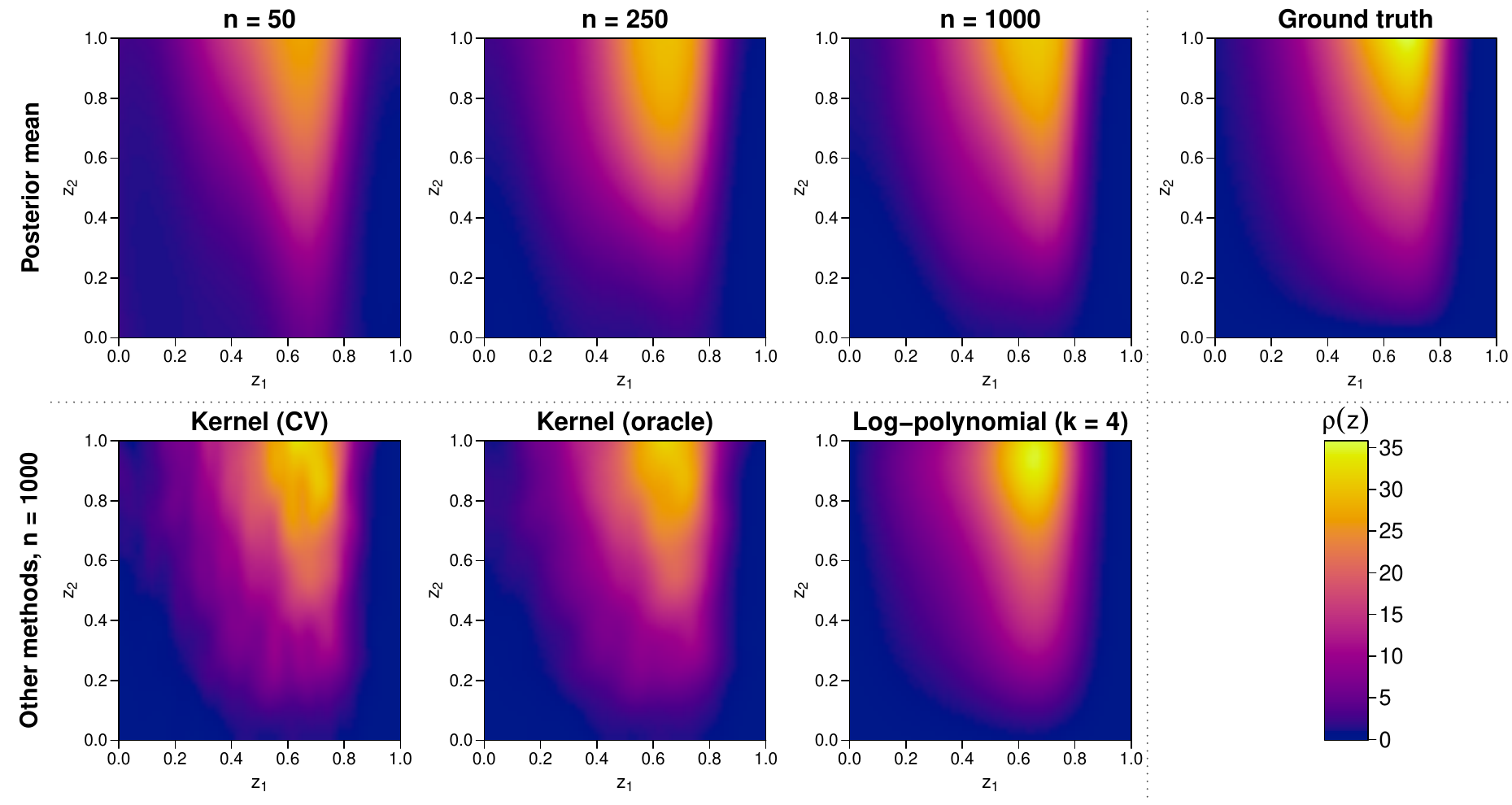}
	\caption{Setting with random and deterministic covariates and ground truth \eqref{Eq:2DTruthCoord}. Top row: posterior means for $n=50,250,1000$ and the ground truth. Bottom row: kernel estimates with cross-validated and oracle bandwidths, and log-polynomial estimate of degree 4, for $n=1000$.}
	\label{Fig:2DResultsCoord}
\end{figure}

\begin{table}[H]
	\centering
	\small
	\begin{tabular}{l|cccc}
		& $n=50$ & $n=250$ & $n=500$ & $n=1000$ \\
		\hline
		Posterior mean & 0.188 (0.036) & 0.100 (0.010) & \textbf{0.077} (0.012) & \textbf{0.056} (0.010) \\
		\noalign{{\color{gray!60}\hrule height 0.3pt}}
		Kernel, CV & 0.188 (0.037) & 0.105 (0.015) & 0.085 (0.012) & 0.073 (0.008) \\
		Kernel, oracle & 0.171 (0.031) & 0.101 (0.011) & 0.080 (0.011) & 0.070 (0.009) \\
		\noalign{{\color{gray!60}\hrule height 0.3pt}}
		Log-polynomial, $k=1$ & 0.637 (0.009) & 0.634 (0.003) & 0.633 (0.002) & 0.633 (0.001) \\
		Log-polynomial, $k=2$ & 0.265 (0.025) & 0.251 (0.010) & 0.246 (0.007) & 0.247 (0.005) \\
		Log-polynomial, $k=3$ & 0.191 (0.030) & 0.155 (0.012) & 0.148 (0.007) & 0.144 (0.005) \\
		Log-polynomial, $k=4$ & \textbf{0.162} (0.032) & \textbf{0.094} (0.016) & \textbf{0.077} (0.008) & 0.068 (0.006) \\
		\hline
	\end{tabular}
	\caption{Mean relative $L^2$-errors (and between-dataset standard deviations) over 50 independent datasets for each sample size.}
	\label{Tab:2DResultsCoord}
\end{table}

%
%
%

\subsection{Experiments with over-parametrization}
\label{Subsec:Overparam}

Lastly, we assess the robustness of our approach to over-parametrization, by considering a model with an additional `spurious' covariate that has no effect on the true data-generating mechanism. Specifically, we take the datasets of the univariate experiments of Section~\ref{Sec:Simulations}, generated from the covariate random field $Z_1$ and the ground truth \eqref{Eq:1DTruth}, and fit all methods based on the bivariate covariate $Z=(Z_1,Z_2)$, where $Z_2(x)=1/2+x_1$, $x\in\Wcal$, is the deterministic covariate of Section~\ref{Subsec:ExpDetCovariates}. The corresponding (over-parametrized) ground truth is thus
\begin{equation}
	\label{Eq:2_overparam}
	\rho_0(z) = 5 f_{SN}(z_1; 0.8, 0.3, -5), \qquad z\in[0,1]^2,
\end{equation}
which is constant in $z_2$, cf.~Figure~\ref{Fig:2DResultsMiss}. All other specifications of the bivariate experiments of Section~\ref{Sec:Simulations} were kept unchanged. Figure~\ref{Fig:2DResultsMiss} displays the estimates obtained on representative datasets, and Table~\ref{Tab:2DResultsOverparam} reports the relative $L^2$-errors.

\begin{figure}[H]
	\centering
	\includegraphics[width=\textwidth]{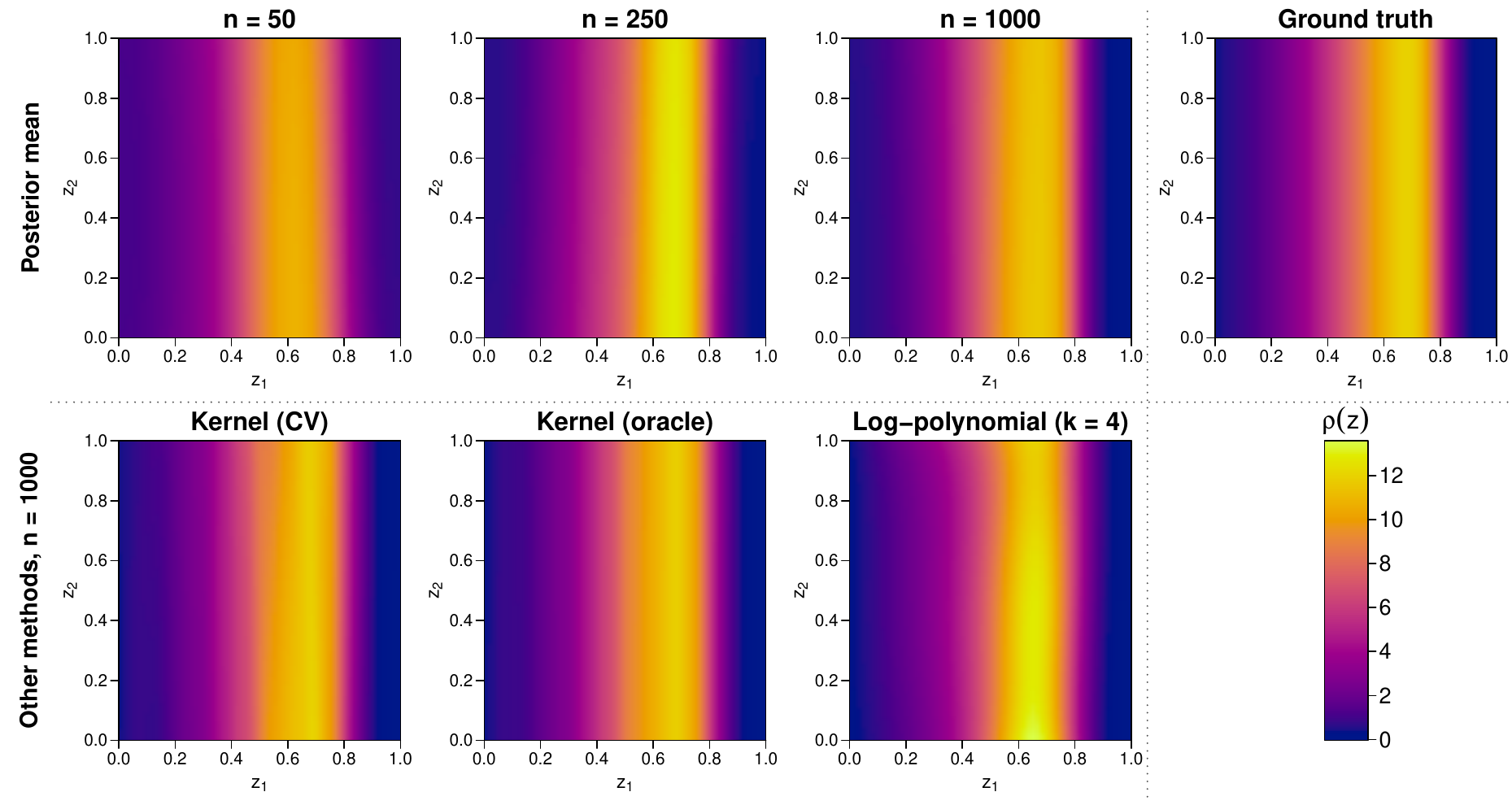}
	\caption{Setting with over-parametrization and ground truth \eqref{Eq:2_overparam}. Top row: posterior means for $n=50,250,1000$ and the ground truth. Bottom row: kernel estimates with cross-validated and oracle bandwidths, and log-polynomial estimate of degree 4, for $n=1000$.}
	\label{Fig:2DResultsMiss}
\end{figure}

\begin{table}[H]
	\centering
	\small
	\begin{tabular}{l|cccc}
		& $n=50$ & $n=250$ & $n=500$ & $n=1000$ \\
		\hline
		Posterior mean & 0.149 (0.040) & \textbf{0.056} (0.013) & \textbf{0.047} (0.015) & \textbf{0.036} (0.008) \\
		\noalign{{\color{gray!60}\hrule height 0.3pt}}
		Kernel, CV & 0.144 (0.044) & 0.076 (0.020) & 0.062 (0.014) & 0.043 (0.010) \\
		Kernel, oracle & \textbf{0.122} (0.034) & 0.065 (0.015) & 0.053 (0.012) & 0.040 (0.009) \\
		\noalign{{\color{gray!60}\hrule height 0.3pt}}
		Log-polynomial, $k=1$ & 0.627 (0.008) & 0.625 (0.003) & 0.624 (0.002) & 0.624 (0.002) \\
		Log-polynomial, $k=2$ & 0.276 (0.026) & 0.244 (0.010) & 0.244 (0.007) & 0.241 (0.006) \\
		Log-polynomial, $k=3$ & 0.210 (0.035) & 0.156 (0.013) & 0.148 (0.011) & 0.143 (0.007) \\
		Log-polynomial, $k=4$ & 0.203 (0.040) & 0.106 (0.014) & 0.088 (0.016) & 0.074 (0.006) \\
		\hline
	\end{tabular}
	\caption{Mean relative $L^2$-errors (and between-dataset standard deviations) over 50 independent datasets for each sample size.}
	\label{Tab:2DResultsOverparam}
\end{table}

The posterior mean correctly captures the absence of an effect of $Z_2$. In our experiments, the posterior means of the length-scale $\ell_2$ (medians over the datasets) decrease from $0.20$ for $n=50$ to $0.09$ for $n=1000$, corresponding to correlation lengths $1/\ell_2$ far exceeding the size of the covariate space, while those of $\ell_1$ remain close to the values reported in Table~\ref{Tab:Hyper} for the univariate experiments. This suggests that the hyper-prior $\Pi_\ell$, whose complexity penalty favours small length-scales, combined with the ARD covariance kernel, may possess dimension-reduction properties, effectively switching off covariates whose effect is not supported by the data; cf.~Section~\ref{Subsec:Prior}. Accordingly, the estimation errors are only slightly larger, by at most $0.004$, than those obtained on the same datasets with the univariate model, cf.~Table~\ref{Tab:Errors}. The kernel estimator with oracle bandwidth is similarly robust, selecting the largest bandwidth along $Z_2$ in the majority of experiments, whereas the log-polynomial models deteriorate more markedly, with the error of the model of degree~4 increasing from $0.114$ to $0.203$ for $n=50$.

%
%
%

\section{Expanded applications to the Canadian wildfire dataset}
\label{Sec:MoreRealData}

In this section, we expand the application to the Canadian wildfire dataset developed in Section~\ref{Sec:RealData}. We present further results for the province of Ontario, and repeat the workflow for the provinces of British Columbia and Saskatchewan. For the latter two, we use the same pre-processing, prior and MCMC sampler as for Ontario in Section~\ref{Sec:RealData}.

%
%
%

\subsection{Further results for the Ontario dataset}
\label{S:Sec:MoreOntario}

In addition to the exploratory univariate analysis of Section~\ref{Subsec:1DOntario} and the full analysis of Section~\ref{Subsec:3DOntario}, we also fit a model jointly based on temperature and precipitation, which our investigations, in accordance with the literature, e.g.~\citet{BGMMM20}, indicate as the two most influential meteorological factors. In Figure~\ref{Fig:2DOntario}, we plot the obtained posterior mean (central panel) and the kernel estimate with cross-validated bandwidths (right panel). Both place greater intensities at high temperatures and low precipitation levels, with the kernel estimate displaying more local variability. These findings are in line with the full analysis of Section~\ref{Subsec:3DOntario}, cf.~Figure~\ref{Fig:3DOntarioMarginals}, where the inclusion of the average wind speed as an additional covariate was observed to affect the overall intensity level, but to generally preserve the distribution of the risk across temperature and precipitation.

\begin{figure}[H]
	\centering
	\includegraphics[width=\textwidth]{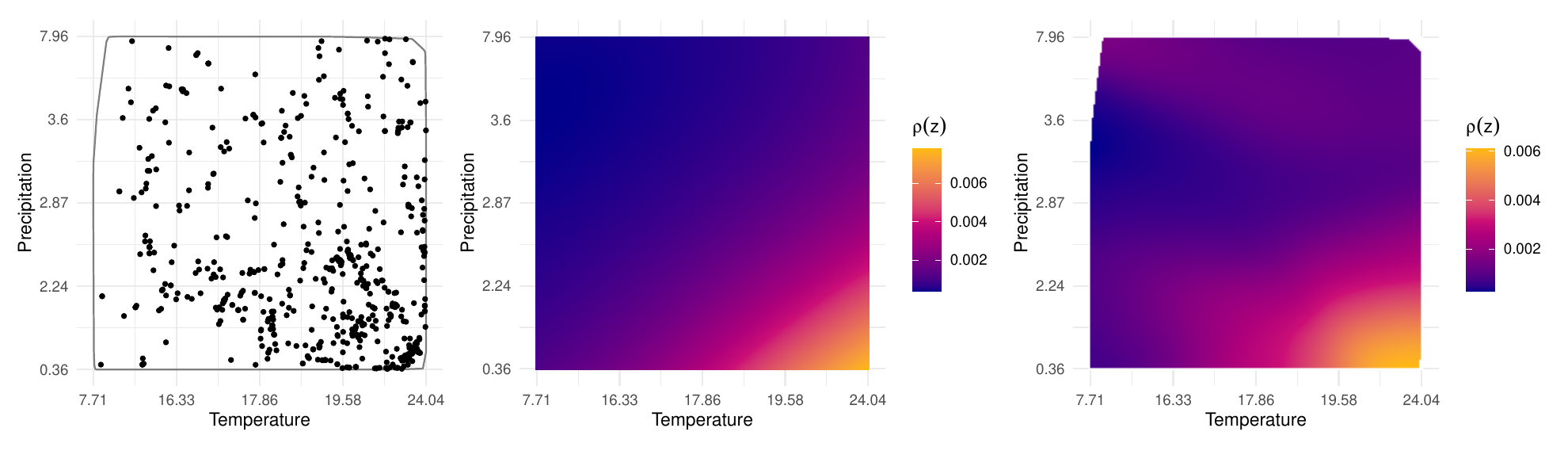}
	\caption{Left panel: average temperatures and precipitation levels at the Ontario locations where a wildfire was detected over the considered time period. The grey line is the convex hull of the covariate values observed across the whole province and time period, including locations without wildfires. Central panel: posterior mean of the wildfire intensity as a function of the two covariates. Right panel: kernel estimate with cross-validated bandwidths, restricted to the convex hull.}
	\label{Fig:2DOntario}
\end{figure}

Returning to the full analysis based on temperature, precipitation level and wind speed, cf.~Section~\ref{Subsec:3DOntario}, in Figure~\ref{Fig:3DOntarioSpatialAdd} we display the plug-in posterior means of the yearly spatial intensity for six additional years across the time period, specifically 2006, 2008, 2010, 2011, 2018 and 2019. Note that, for visual clarity of the individual plots, the color scales differ across the panels. Years with a small number of wildfires, like 2008 and 2010 (second and third panel, respectively), are assigned uniformly lower intensity values, while in years with more events, such as 2018 and 2019 (fifth and sixth panel, respectively), the areas of elevated intensity largely overlap with the main clusters of detected wildfires.

\begin{figure}[H]
	\centering
	\includegraphics[width=\textwidth]{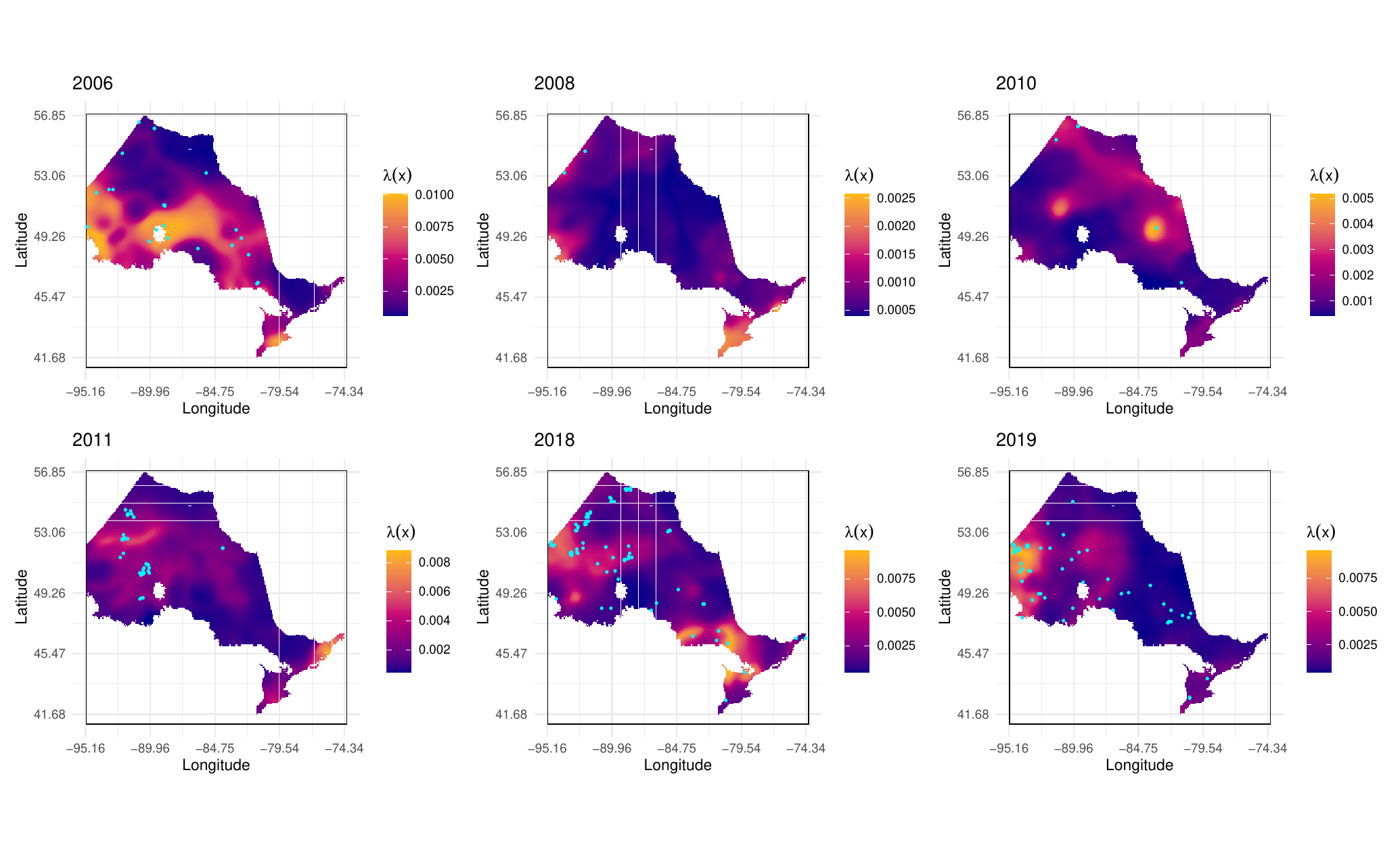}
	\caption{Left to right, top to bottom: plug-in posterior means of the spatial intensity based on average temperature, precipitation level and wind speed in Ontario, for the years 2006, 2008, 2010, 2011, 2018 and 2019.}
	\label{Fig:3DOntarioSpatialAdd}
\end{figure}

%
%
%

\subsection{Results for the British Columbia dataset}
\label{S:Sec:BritishColumbia}

The British Columbia and Saskatchewan datasets are structured similarly to the Ontario one described in Section~\ref{Sec:RealData}, each comprising $n=19$ spatial point patterns with the aggregate locations of the wildfires detected in June from 2004 to 2022, and as many three-dimensional spatial covariate fields with the location-specific average temperatures, precipitation levels and wind speeds. As for Ontario, we observe a strong variability in the yearly number of events, as well as in the range and fluctuations of the covariates. Figure~\ref{Fig:CovBT} showcases two individual observations.

\begin{figure}[H]
	\centering
	\begin{subfigure}[b]{\textwidth}
		\centering
		\includegraphics[width=0.8\textwidth]{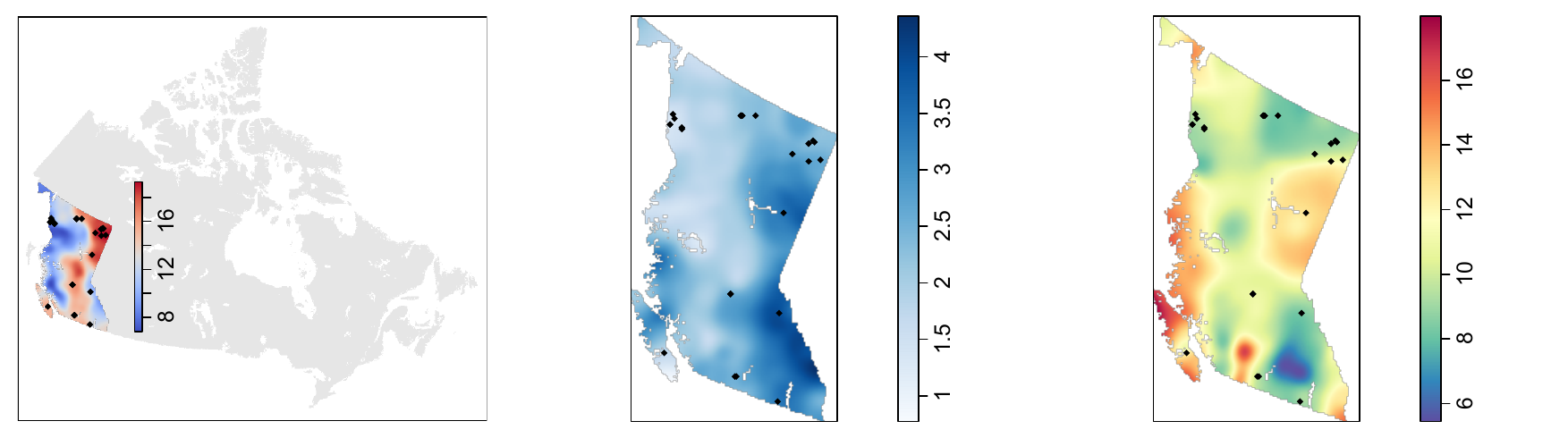}
	\end{subfigure}
	\begin{subfigure}[b]{\textwidth}
		\centering
		\includegraphics[width=0.8\textwidth]{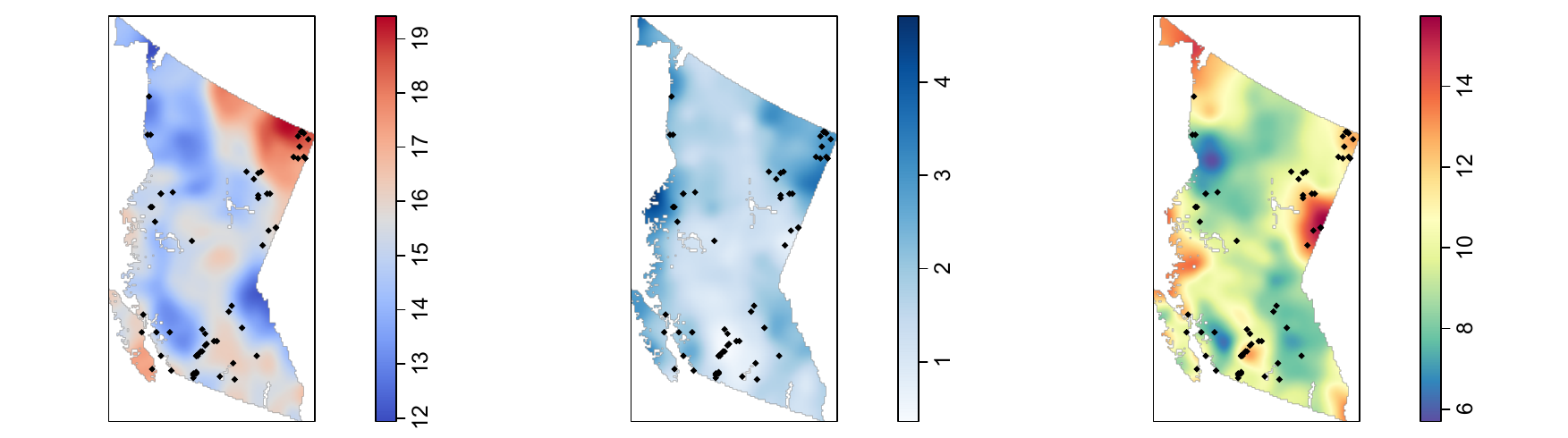}
	\end{subfigure}
	\caption{Top row, left to right: average temperatures (in Celsius), precipitation (in mm) and wind speeds (in km/h) in British Columbia during June 2013. Bottom row: Data for 2021. Black dots indicate the wildfires (respectively, $28$ and $64$ in total).}
	\label{Fig:CovBT}
\end{figure}

We again first perform a preliminary analysis, investigating the influence of each individual covariate on the wildfire intensity. The resulting posterior means are displayed in Figure~\ref{Fig:1DBTColumbia}, alongside kernel estimates. Those for the temperature and the precipitation are qualitatively in line with the ones for Ontario, cf.~Figure~\ref{Fig:1DOntario}, with the intensity increasing sharply between about $12^\circ$C and $18^\circ$C, and decreasing for precipitation levels between about $1.5$ and $2.5$ mm, while no clear effect of the wind speed is detected.

\begin{figure}[H]
	\centering
	\includegraphics[width=\textwidth]{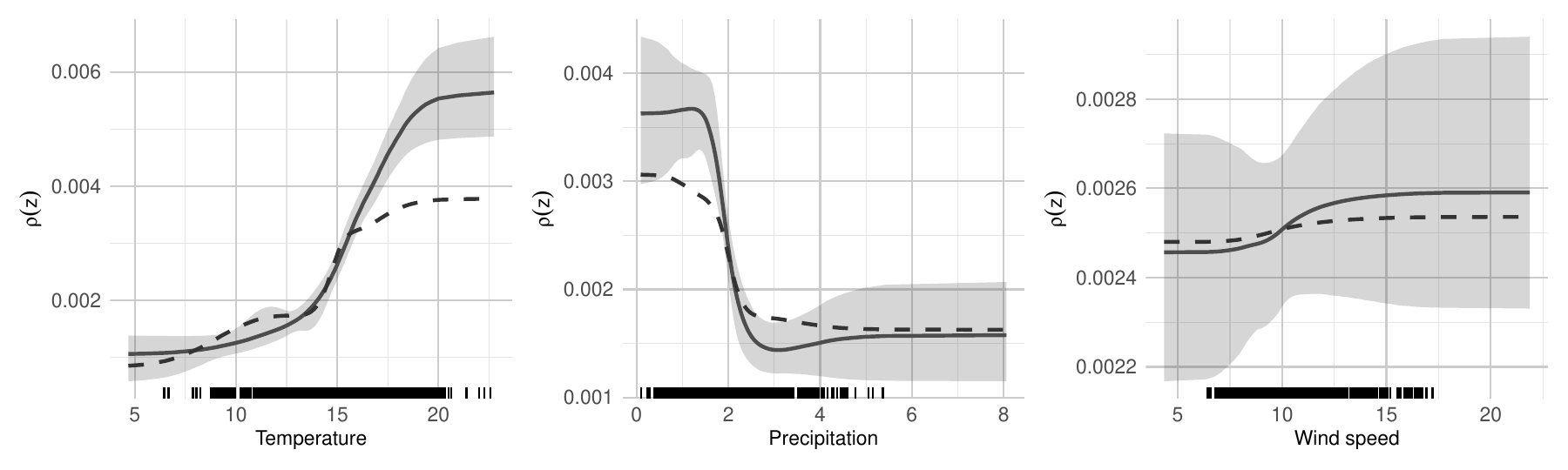}
	\caption{Left to right: posterior means (solid line) and point-wise $95\%$-credible intervals (shaded region) for the intensity as a function of the average temperature, precipitation and wind speed. The dashed lines show the kernel estimates with cross-validated bandwidths, and the ticks mark the covariate values at the observed wildfires.}
	\label{Fig:1DBTColumbia}
\end{figure}

The plug-in posterior means of the yearly spatial intensity resulting from the full analysis based on the three covariates are shown in Figure~\ref{Fig:3DBTSpatial3}, for the same selection of years 2006, 2008, 2010, 2011, 2013, 2015, 2018, 2019 and 2021 considered for Ontario. We again note that, in years with more events, such as 2015 and 2018, the areas of elevated intensity largely overlap with the main clusters of detected wildfires.

\begin{figure}[H]
	\centering
	\includegraphics[width=\textwidth]{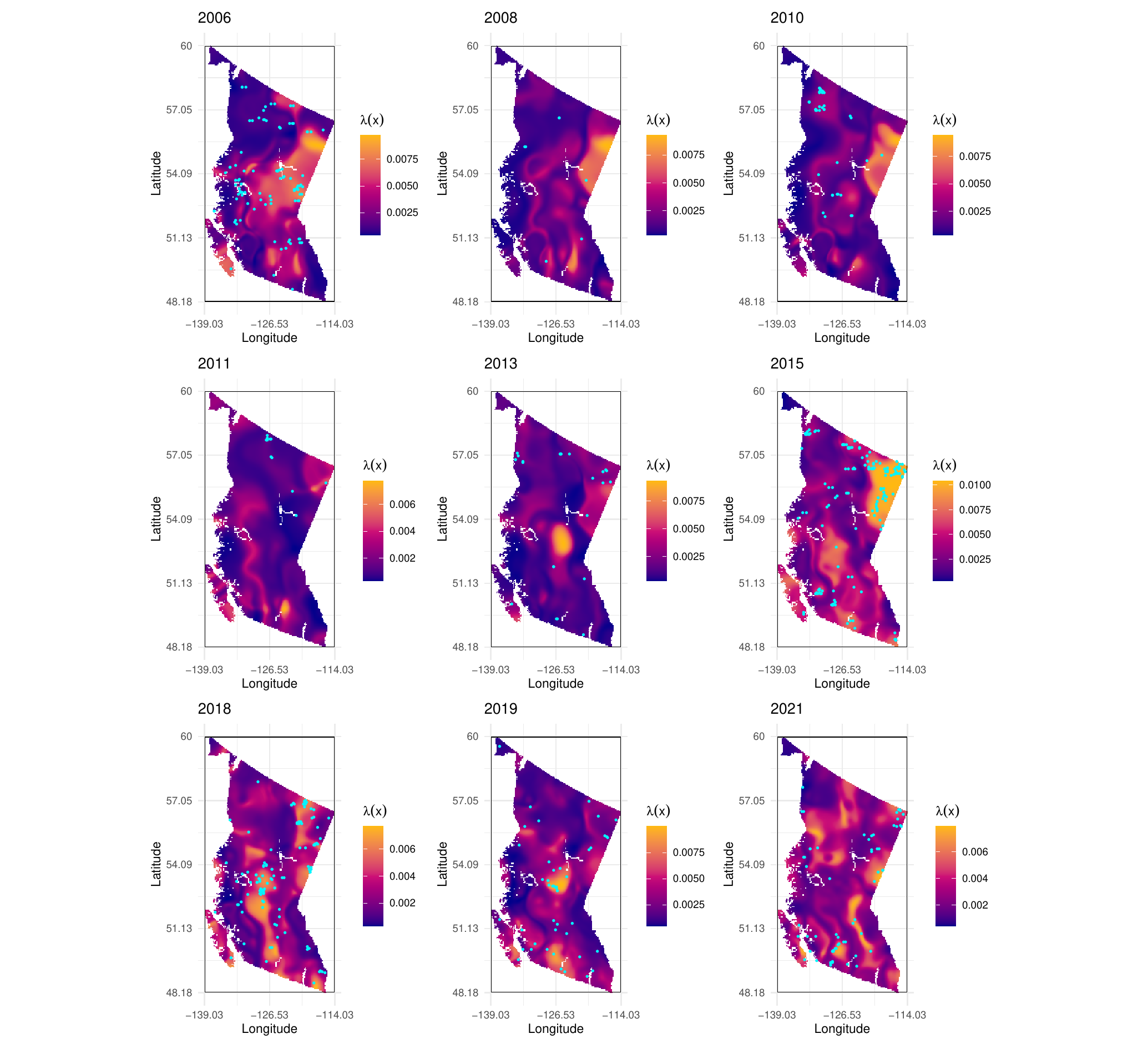}
	\caption{Left to right, top to bottom: plug-in posterior means of the spatial intensity based on average temperature, precipitation level and wind speed, for the years 2006, 2008, 2010, 2011, 2013, 2015, 2018, 2019 and 2021.}
	\label{Fig:3DBTSpatial3}
\end{figure}

%
%
%

\subsection{Results for the Saskatchewan dataset}
\label{S:Sec:Saskatchewan}

We conclude with an overview of the results obtained for the Saskatchewan dataset, for which Figure~\ref{Fig:CovSaksa} showcases two observations of the wildfire locations and meteorological covariates for two different years.

\begin{figure}[H]
	\centering
	\begin{subfigure}[b]{\textwidth}
		\centering
		\includegraphics[width=0.8\textwidth]{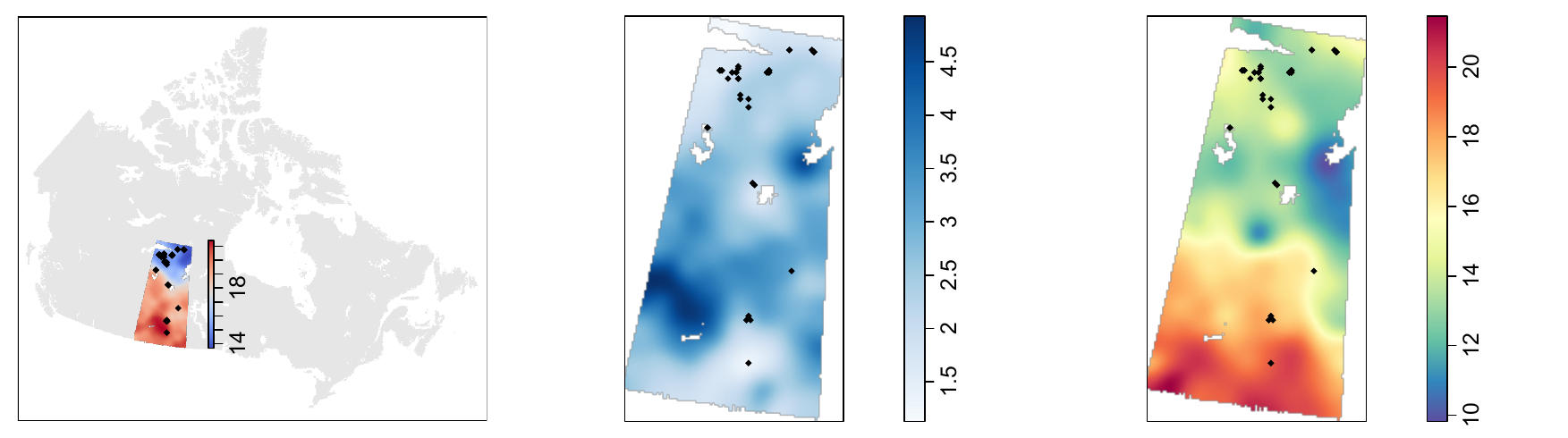}
	\end{subfigure}
	\begin{subfigure}[b]{\textwidth}
		\centering
		\includegraphics[width=0.8\textwidth]{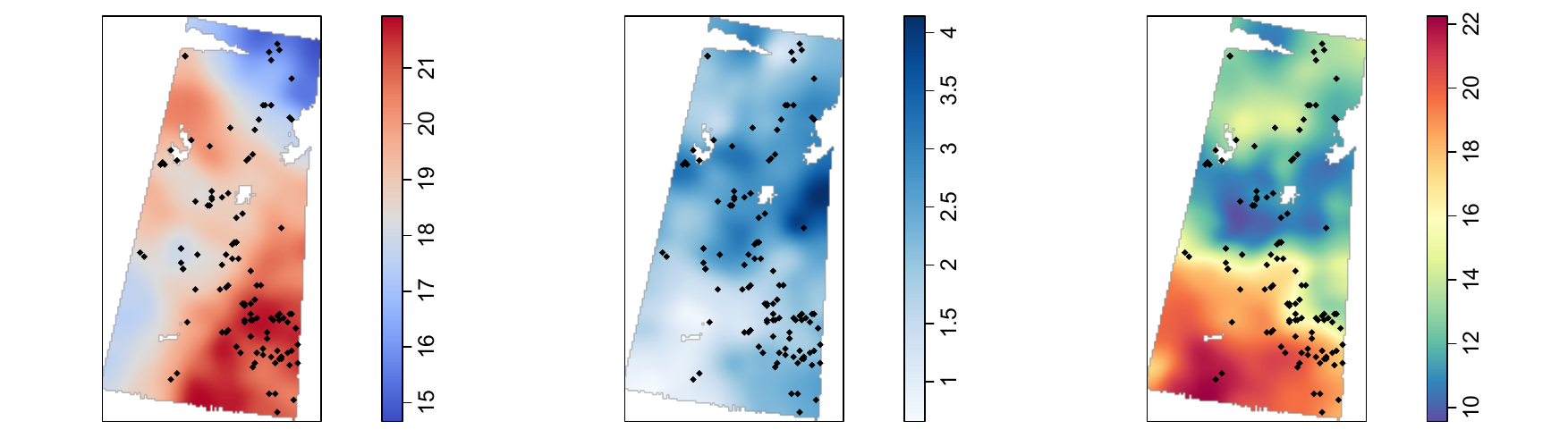}
	\end{subfigure}
	\caption{Top row, left to right: average temperatures (in Celsius), precipitation (in mm) and wind speeds (in km/h) in Saskatchewan during June 2013. Bottom row: Data for 2021. Black dots indicate the wildfires (respectively, $40$ and $114$ in total).}
	\label{Fig:CovSaksa}
\end{figure}

Figure~\ref{Fig:1DSaska} shows the results of the preliminary univariate analysis. The posterior mean based on the temperature increases up to about $15^\circ$C, with fluctuations at higher temperatures that the nearly constant kernel estimate does not show. As for Ontario, a strong negative effect of the precipitation level is captured, particularly up to about $2$ mm. For the wind speed, the posterior mean peaks around $13$ km/h, as in Figure~\ref{Fig:1DOntario} (right panel), but is overall higher in the left tail than in the right one. Together with the results for British Columbia, this suggests some heterogeneity across the regions in the effect of wind speed. The plug-in posterior means of the yearly spatial intensity resulting from the full analysis are shown in Figure~\ref{Fig:3DsasSpatial}, for the same years as above.

\begin{figure}[H]
	\centering
	\includegraphics[width=\textwidth]{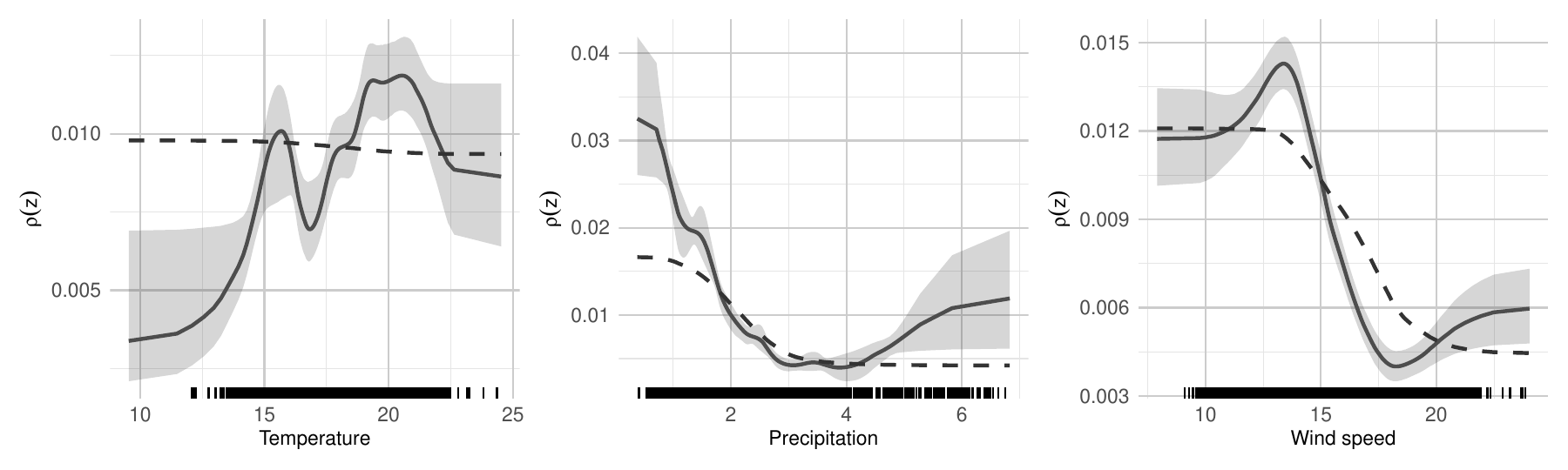}
	\caption{Left to right: posterior means (solid line) and point-wise $95\%$-credible intervals (shaded region) for the intensity as a function of the average temperature, precipitation and wind speed. The dashed lines show the kernel estimates with cross-validated bandwidths, and the ticks mark the covariate values at the observed wildfires.}
	\label{Fig:1DSaska}
\end{figure}

\begin{figure}[H]
	\centering
	\includegraphics[width=\textwidth]{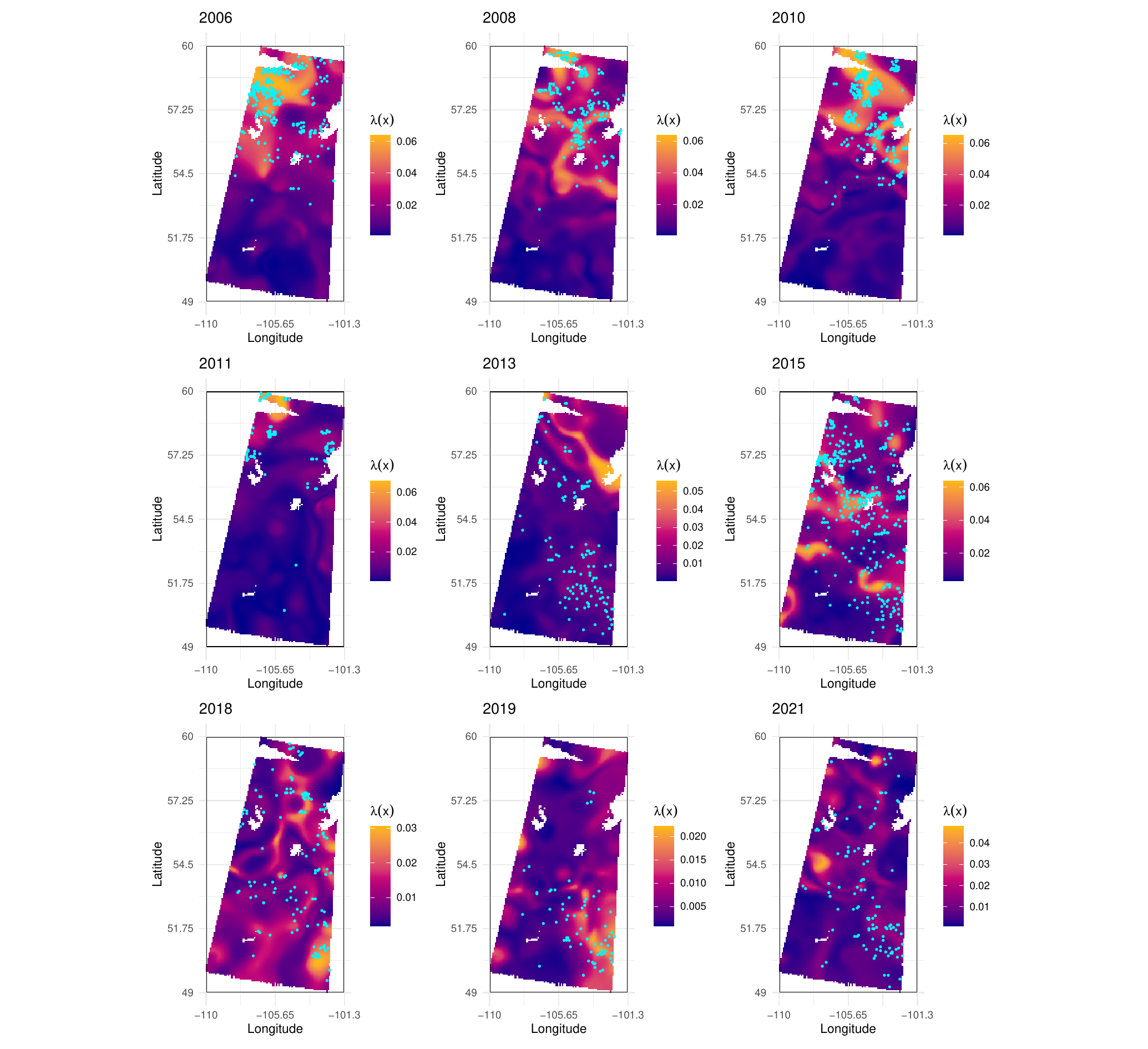}
	\caption{Left to right, top to bottom: plug-in posterior means of the spatial intensity based on average temperature, precipitation level and wind speed, for the years 2006, 2008, 2010, 2011, 2013, 2015, 2018, 2019 and 2021.}
	\label{Fig:3DsasSpatial}
\end{figure}

\end{document}